\documentclass[a4paper]{article}
\usepackage[affil-it]{authblk}
\usepackage{amsfonts}     
\usepackage{graphicx}        
\usepackage{enumerate}
\usepackage[utf8]{inputenc}
\usepackage{amsmath}
\usepackage{mathtools}
\usepackage{bm}
\usepackage{comment}
\usepackage{algorithm}
\usepackage{algorithmic}
\usepackage{amssymb}
\usepackage{mathrsfs}
\usepackage{placeins}
\usepackage{float} 
\usepackage{hyperref}
\usepackage{xcolor}
\usepackage{amsthm}
\usepackage{physics}
\usepackage{enumitem}
\usepackage{mdframed}
\usepackage{cancel}
\usepackage{makecell}
\usepackage{tabularx}
\usepackage{subcaption}
\usepackage{pgffor}
\usepackage[numbers,sort&compress]{natbib}

\allowdisplaybreaks

\setlist{nosep}

\newtheorem{theorem}{Theorem}[section]
\newtheorem{lemma}[theorem]{Lemma}

\newtheorem{corollary}[theorem]{Corollary}
\newtheorem{proposition}[theorem]{Proposition}

\theoremstyle{definition}

\newcommand{\vecI}{I}

\newcommand{\sumI}{\left| \vecI \right|}

\newcommand{\makevec}[1]{#1}

\newcommand{\dimn}{N} 

\DeclareMathOperator*{\argmin}{arg\,min}
\DeclareMathOperator{\Haf}{\textnormal{Haf}}

\begin{document}
	\title{Estimating the Percentage of GBS Advantage in Gaussian Expectation Problems\thanks{This work is supported by ReNewQuantum from the ERC Synergy program under grant agreement No. 810573, CLUSTEC from the Horizon Europe program No. 101080173,  QCI.DK from the Digital Europe program No. 101091659, TopQC2X from Danmarks Innovationsfond Instrument Grand Solutions Award No. 3200-00030B, and the Simons Foundation collaboration grant on New Structures in Low-Dimensional Topology.}}
	\author{Jørgen Ellegaard Andersen$^{1,2}$ \& Shan Shan$^{1}$ \\
		{\small $^{1}$Center for Quantum Mathematics, University of Southern Denmark \\
			$^{2}$ Danish Institute of Advanced Study, University of Southern Denmark} \\
	}
	\maketitle
	
	\begin{abstract}
		Gaussian Boson Sampling (GBS), which can be realized with a photonic quantum computing model, perform some special kind of sampling tasks. In \cite{shan2024companion}, we introduced algorithms that use GBS samples to approximate Gaussian expectation problems. We found a non-empty open subset of the problem space where these algorithms achieve exponential speedup over the standard Monte Carlo (MC) method. This speedup
		is defined in terms of the guaranteed sample size to
		reach the same accuracy $\epsilon$ and success probability $\delta$ under the $(\epsilon, \delta)$ multiplicative error approximation scheme.
		In this paper, we enhance our original approach by optimizing the average photon number in the GBS distribution to match the specific Gaussian expectation problem. We provide updated estimates of the guaranteed sample size for these improved algorithms and quantify the proportion of problem space where they outperform MC. Numerical results indicate that the proportion of the problem space where our improved algorithms have an advantage is substantial, and the advantage gained is significant. Notably, for certain special cases, our methods consistently outperform MC across nearly 100\% of the problem space. 
	\end{abstract}
	
	\section{Introduction}
The Gaussian expectation problem $\mathcal{I}_{\Haf}$, introduced in \cite{shan2024companion}, involves computing the integral
\begin{align}
	\mu_{\Haf} = \int_{\mathbb{R}^N} f(x) g(x) \dd x,
	\label{eq:I}
\end{align}
where
\begin{equation*}
	g(x) = (2\pi)^{-N/2} (\det B)^{-1/2} \exp(-\frac{1}{2}x^\intercal B^{-1} x)
\end{equation*}
is the Gaussian probability density function with zero mean and covariance matrix $B$ (real, symmetric and positive definite), and
\begin{align*}
	f(\makevec{x})= \sum_{k=0}^K \sum_{\sumI = k} a_{\vecI} \makevec{x}^{\vecI}, ~~~~ a_I \in \mathbb{R}
\end{align*}
is either a multivariate polynomial ($K < \infty$) or a formal power series ($K = \infty$) with 
\begin{gather*}
	\vecI = (i_1, \dots, i_\dimn), ~~
	\sumI = i_1 + i_2 \dots + i_N, ~~
	x^I = x_1^{i_1} \dots x_N^{i_N}.
\end{gather*}
Using Wick's theorem, $\mu_{\Haf}$ can be expressed as a weighted sum of matrix hafnians
\begin{align}
	\mu_{\Haf} = \sum_{k=0}^K \sum_{\sumI = 2k} a_{\vecI} \Haf(B_I),
	\label{eq:muhaf}
\end{align}
subject to certain integrability conditions for $K = \infty$ (details in \cite{shan2024companion}). The $(\epsilon, \delta)$-multiplicative error approximation problem $\mathcal{I}^{\times}_{\Haf}(\epsilon, \delta)$ is to find $e$ such that for any $0 < \epsilon, \delta < 1$,
\begin{equation}
	P\left(\lvert \mu_{\Haf} - e \rvert > \epsilon \vert \mu_{\Haf} \vert \right) < \delta, 
	\label{eq:adderror}
\end{equation}
provided that $\mu_{\Haf} \neq 0.$

An important family of special cases $\mathcal{I}_{\Haf^2}$, is defined by the integral
\begin{align}
	\mu_{\Haf^2} = \int_{\mathbb{R}^{2N}} {f}(p, q) {g}(p, q) \, \dd p \dd q, \label{eq:I2}
\end{align}
where
\begin{align}
	{g}(p, q) = (2\pi)^{-N} (\det B)^{-1} \exp(-\frac{1}{2} [p, q]^\intercal (B \oplus B)^{-1} [p, q]) \label{eq:I2a}
\end{align}
is the Gaussian probability density function with zero mean and covariance matrix $B \oplus B$,
and
\begin{align*}
	{f}(p, q) = \sum_{k=0}^K \sum_{\sumI = k} a_{\vecI} p^I q^I, ~~~~~a_I \in \mathbb{R}.
\end{align*}
Again, using Wick's theorem, $\mu_{\Haf^2}$ can be written as 
\begin{align}
	\label{eq:muhafsq}
	\mu_{\Haf^2} = \sum_{k=0}^K \sum_{\sumI = 2k} a_{\vecI} \Haf(B_I)^2.
\end{align}
The $(\epsilon, \delta)$-multiplicative error approximation problem $\mathcal{I}^\times_{\Haf^2}(\epsilon, \delta)$ is defined similarly as in \eqref{eq:adderror}.

In \cite{shan2024companion}, we presented solutions to $\mathcal{I}^{\times}_{\Haf}(\epsilon, \delta)$ and $\mathcal{I}^\times_{\Haf^2}(\epsilon, \delta)$ using samples from a Gaussian Boson Sampling (GBS) device \cite{aaronson2011computational, hamilton2017gaussian, zhong2020quantum, madsen2022quantum}, provided that the eigenvalues of $B$ are strictly less than 1. We introduced the probability estimator (GBS-P) for solving $\mathcal{I}^{\times}_{\Haf}(\epsilon, \delta)$ and the importance estimator (GBS-I) for solving $\mathcal{I}^{\times}_{\Haf^2}(\epsilon, \delta)$. We proved that there exists a non-empty open subset of the problem space such that the GBS estimators achieve exponential speedup over the standard Monte Carlo (MC) method. Here, the speedup is defined in terms of the guaranteed sample size, as derived using Markov's inequality or Chebyshev's inequality, to reach the same accuracy $\epsilon$ and success probability $\delta$.

In this paper, we enhance our original methods by optimizing the average photon number in the GBS distribution to match the specific Gaussian expectation problem. To this end, we first introduce a partition by degree of $\mu_{\Haf}$, given by
\begin{align}
	\mu_{\Haf} = \mu_{\Haf, 0} + \mu_{\Haf, 1} + \mu_{\Haf, 2} + \dots + \mu_{\Haf, K},
	\label{eq:partition}
\end{align}
where
\begin{align*}
	\mu_{\Haf, k} = \sum_{\sumI = 2k} a_{\vecI} \Haf(B_I).
\end{align*}
For each degree $\mu_{\Haf, k}$, we define the corresponding multiplicative error approximation problem $\mathcal{I}^{\times}_{\Haf, k}(\epsilon, \delta)$. By De Morgan's Law and the union bound, solving $\mathcal{I}^{\times}_{\Haf}(\epsilon, \delta)$ reduces to solving each individual $\mathcal{I}^{\times}_{\Haf, k}(\epsilon, \delta)$ (see Lemma \ref{lem:kslice} for details). We focus on  $\mathcal{I}^{\times}_{\Haf, K}(\epsilon, \delta)$, with a slight abuse of notation where
\begin{align*}
	\mathcal{I}^{\times}_{\Haf}(\epsilon, \delta) = \mathcal{I}^{\times}_{\Haf, K}(\epsilon, \delta) ~~~~\textnormal{and}~~~~\mu_{\Haf} = \mu_{\Haf, K}.
\end{align*} 
The same partitioning strategy applies to $\mu_{\Haf^2}$, leading to the corresponding problem $\mathcal{I}^{\times}_{\Haf^2, K}(\epsilon, \delta)$. Again, by an abuse of notation, we set
\begin{align*}
	\mathcal{I}^{\times}_{\Haf^2}(\epsilon, \delta) = \mathcal{I}^{\times}_{\Haf^2, K}(\epsilon, \delta) ~~~~\textnormal{and}~~~~\mu_{\Haf^2} = \mu_{\Haf^2, K}.
\end{align*}

We refine our estimator for each degree $K$ problem by using samples from an optimized GBS distribution. The key idea is to select a distribution $P_{tB}$ where the average photon number equals $2K$. Precisely, let $0 < \lambda_N \leq \dots \leq \lambda_1 <1$ be the eigenvalues of $B$ and let $t \in (0, \lambda_1^{-1})$. The GBS distribution for the scalar-multiplied matrix $tB$ is given by
\begin{align*}
	P_{tB}(I) & = \frac{d_t}{I!} \Haf((tB)_I)^2, ~~~~ 	d_t= \prod_{n = 1}^N \sqrt{1 - t^2\lambda_n^2}.
\end{align*}
Here, $(tB)_I$ denotes the sub- or super-matrix of $tB$ determined by $I$.
We define the \textit{average photon number} of samples from $P_{tB}$ to be
\begin{align*}
	m_{tB} = \mathbb{E}_{I \sim P_{tB}}[\vert I \vert],
\end{align*}
and we want
\begin{align*}
	m_{tB} = 2K.
\end{align*}
Let $I_1, \dots, I_n$ be i.i.d samples from $P_{tB}$. The GBS-P estimator is defined as a weighted average of the discrete approximation of the probability distribution
\begin{align*}
	\mathcal{E}^{\textnormal{GBS-P}}_{n} = \sum_{\vert I \vert = 2K} a_I t^{-K} \sigma_I \sqrt{\tfrac{I!}{d_t}} \sqrt{\tfrac{S^{(I)}_n}{n}}
\end{align*}
where $S^{(I)}_n$ counts the occurrences of $I$ and $\sigma_I$ is the sign of $\Haf(B_I)$. Here, we assume that there is an efficient method for computing $\sigma_I$. We will return to this part in a future publication. For $\mathcal{I}^\times_{\Haf^2}(\epsilon, \delta)$, we define the GBS-I estimator to be
\begin{align*}
	\mathcal{E}^\textnormal{GBS-I}_n = \frac{1}{n} \sum_{i = 1}^n \frac{ I_i!}{d_t} a_{I_i} t^{-2K}.
\end{align*}

As in \cite{shan2024companion}, GBS-I is an unbiased estimator. Therefore, it follows directly from the weak law of large numbers (WLLN) that $\mathcal{E}^\textnormal{GBS-I}_n$ solves $\mathcal{I}^\times_{\Haf^2}(\epsilon, \delta)$ for a sufficiently large $n$. Using Chebyshev's inequality, we compute the guaranteed sample size of GBS-I to be
\begin{align*}
	n^{\textnormal{GBS-I}}_{\Haf^2} = \frac{	V^{\textnormal{GBS-I}}_{\Haf^2} }{\delta \epsilon^2 \mu_{\Haf^2}^2}, ~~~~ V^{\textnormal{GBS-I}}_{\Haf^2} = \frac{t^{-4K}}{d_t} \sum_{\vert I \vert = 2K} a^2_I I! \Haf(B_I)^2  - \mu_{\Haf^2}^2 .
\end{align*}
The estimator, GBS-P, on the other hand, is biased. So, we cannot apply the WLLN directly. We instead perform a leading-order asymptotic analysis of the squared estimation error as $n$ grows, using techniques very similar to the ones used in \cite{andersen2006asymptotics}.
\begin{theorem}
	\label{thrm:asymp}
	Let $\mathcal{E}^{\textnormal{GBS-P}}_n$ and $\mu_{\Haf}$ be defined as before, then the large $n$ leading order asymptotics of the squared estimation error is given by
	\begin{align*}
		\mathbb{E} \left[ \vert \mathcal{E}^{\textnormal{GBS-P}}_n  - \mu_{\Haf} \vert^2 \right] \sim \frac{1}{4n}\left( \frac{t^{-2K}}{d_t} \sum_{|I| = 2K } a^2_I I! - \mu_{\Haf} ^2\right).
	\end{align*}
\end{theorem}
\noindent
The proof of Theorem \ref{thrm:asymp} is provided in Section \ref{sec:algorithm}.
We remark here that it is an interesting problem to derive the asymptotic expansion to all orders in $n$. We will likely obtain divergent series, but the techniques developed in \cite{andersen2020resurgence, kontsevich2022analyticity, kontsevich2024holomorphic} can give exact analytic expression. We shall in future work return to this part.
Using Theorem \ref{thrm:asymp}, it's straightforward that
\begin{align*}
	\lim_{n \rightarrow \infty} \mathbb{E} \left[ \vert \mathcal{E}^{\textnormal{GBS-P}}_n  - \mu_{\Haf} \vert^2 \right] = 0,
\end{align*}
and therefore $\mathcal{E}^{\textnormal{GBS-P}}_n$ solves  $\mathcal{I}^{\times}_{\Haf}(\epsilon, \delta)$ for sufficiently large $n$. 
Using Markov's inequality, we can further derive an approximate guaranteed sample size, that is
\begin{align*}
	n^{\textnormal{GBS-P}}_{\Haf} = \frac{V^{\textnormal{GBS-P}}_{\Haf} }{\delta \epsilon^2 \mu_{\Haf}^2}, ~~~~ V^{\textnormal{GBS-P}}_{\Haf} \sim \frac{1}{4}\left( \frac{t^{-2K}}{d_t} \sum_{|I| = 2K } a^2_I I! -  \mu_{\Haf}^2\right).
\end{align*}

We compare the efficiency of the new GBS estimators with the standard Monte Carlo estimator, as defined in \cite{shan2024companion}. Recall that
\begin{align*}
	\mathcal{E}^{\textnormal{MC}}_n = \sum_{i=1}^n \frac{1}{n} f(X_i),
\end{align*}
where $X_1, \dots, X_n$ are i.i.d samples from the Gaussian distribution $g$. Since MC is also unbiased, $\mathcal{E}^{\textnormal{MC}}_n $ solves $\mathcal{I}^{\times}_{\Haf}(\epsilon, \delta)$ or $\mathcal{I}^{\times}_{\Haf^2}(\epsilon, \delta)$ for sufficiently large $n$. Furthermore,
Chebyshev's inequality yields the following guaranteed sample sizes:
\begin{align*}
	&n_{\Haf}^{\textnormal{MC}} = \frac{V^{\textnormal{MC}}_{\Haf} }{\delta \epsilon^2 \mu_{\Haf}^2}, ~~~~~ V^{\textnormal{MC}}_{\Haf} = \sum_{\substack{|I|, |I'| = 2K}}  a_I a_{I'} \Haf(B_{I+I'}) - \mu_{\Haf}^2 \\
	&n_{\Haf^2}^{\textnormal{MC}} = \frac{V^{\textnormal{MC}}_{\Haf^2} }{\delta \epsilon^2 \vert \mu_{\Haf^2} \vert^2}, ~~~~ V^{\textnormal{MC}}_{\Haf^2} = \sum_{\substack{|I|, |I'| = 2K}} a_I a_{I'} \Haf(B_{I + I'})^2 - \mu_{\Haf^2}^2.
\end{align*}

To estimate the percentage of the problem space where our improved GBS estimators outperform MC, we define
\begin{align}
	\begin{aligned}
		P^{\textnormal{GBS-P}}_{\Haf}(N, K) & = \frac{\int_{\mathcal{P}_{\Haf}(N, K)} H(n^{\textnormal{MC}}_{\Haf} - n^{\textnormal{GBS-P}}_{\Haf}) \, d\mu}{\int_{\mathcal{P}_{\Haf}(N, K)} 1 \, d\mu}, \\
		P^{\textnormal{GBS-I}}_{\Haf^2}(N, K) & = \frac{\int_{\mathcal{P}_{\Haf^2}(N, K)} H(n^{\textnormal{MC}}_{\Haf^2} - n^{\textnormal{GBS-I}}_{\Haf^2}) \, d\mu}{\int_{\mathcal{P}_{\Haf^2}(N, K)} 1 \, d\mu}.
	\end{aligned}
	\label{eq:pp}
\end{align}
Here, $\mathcal{P}_{\Haf}(N, K)$ is the subset of $a_I$'s and $B$ for which $n^{\textnormal{GBS-P}}_{\Haf}$ and $n^{\textnormal{MC}}_{\Haf}$ are well-defined, and $\mathcal{P}_{\Haf^2}(N, K)$ is the subset for which $n^{\textnormal{GBS-I}}_{\Haf^2}$ and $n^{\textnormal{MC}}_{\Haf^2}$ are well-defined. Furthermore, $d\mu$ denotes the volume form and $H$ is the heavyside step function. The precise definition of the space and volume form is provided in Section \ref{subsec:space}.
Numerical approximations of  $P^{\textnormal{GBS-P}}_{\Haf}(N, K)$ and $P^{\textnormal{GBS-I}}_{\Haf^2}(N, K)$ are made for various $N$ and $K$. A detailed description of the methods used for these computations is provided in Section \ref{subsec:emp}. 

\begin{figure}
	\centering
	\includegraphics[width=\linewidth]{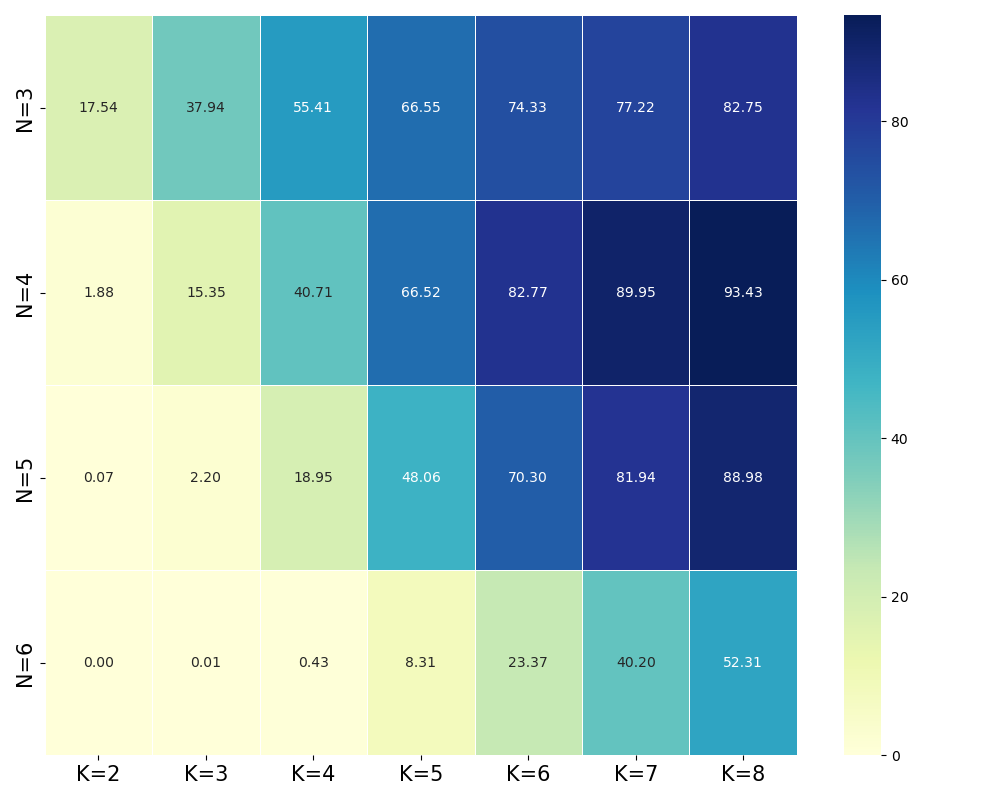}
	\caption{Percentage of the problem space where GBS-P is more efficient than the plain MC method for various $N$ and $K$.}
	\label{fig:heatmap_GBSP}
\end{figure}
Figure \ref{fig:heatmap_GBSP} reports $P^{\textnormal{GBS-P}}_{\Haf}(N, K) \times 100\%$ for $N = 3,\dots, 6$ and $K = 2, \dots, 8$. For fixed $N$, the percentage increases with $K$. This is consistent with Theorem 1.3 and Theorem 1.4 of \cite{shan2024companion}, confirming the growing advantage of GBS-P over MC for large $K$. For a fixed $K$, the percentage decreases with $N$, suggesting limited benefit of GBS-P for small $K$ and large $N$. However, the upper triangular portion of the figure shows a high percentage, indicating that GBS-P is effective when $N$ and $K$ are correlated. This aligns with Theorem 1.5 of \cite{shan2024companion}, where GBS-P provides an exponential speedup over MC with respect to $N$, provided $K \geq \zeta N^2$ for some positive constant $\zeta >0$. We also analyze the ratio $V^{\textnormal{MC}}_{\Haf} / V^{\textnormal{GBS-P}}_{\Haf}$ to quantify the efficiency gains. In Appendix B, we present histograms and  box plots of this ratio. When MC outperforms GBS-P, the advantage is typically around 10 to 20 times greater. When GBS-P outperforms MC, the efficiency gain depends on the value of $K$. When $K$ is large, the advantage is typically around 100 times greater. When $K$ is smaller, the advantage by GBS-P is often very limited. 

\begin{figure}
	\centering
	\includegraphics[width=\linewidth]{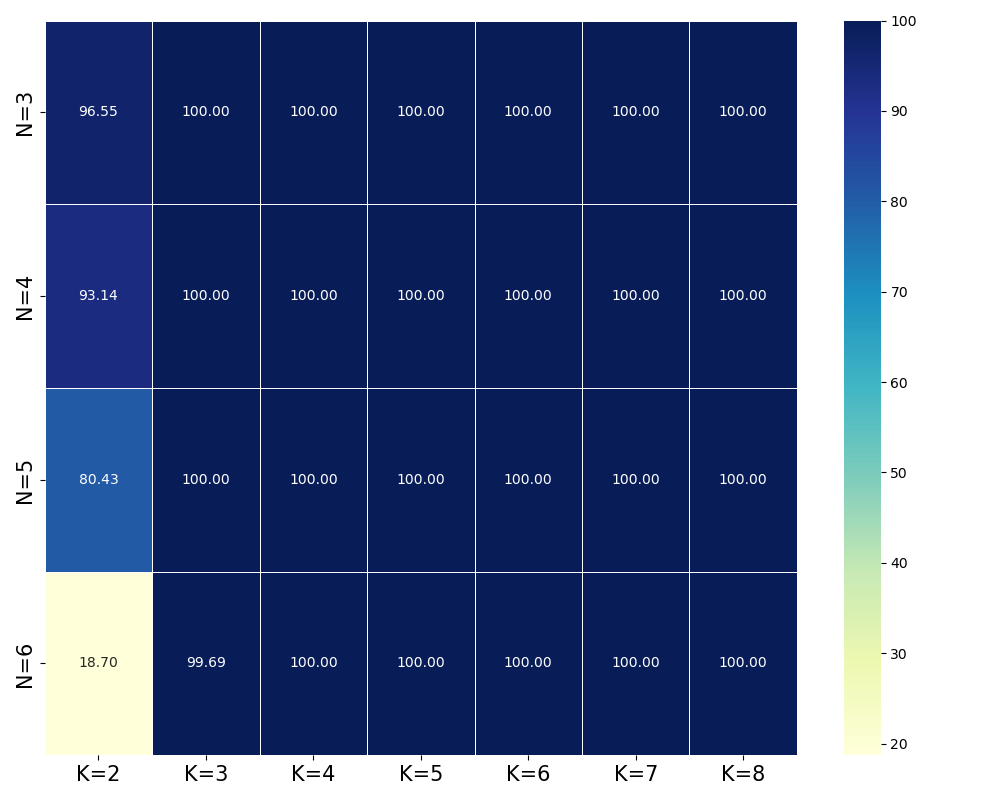}
	\caption{Percentage of the problem space where GBS-I is more efficient than the plain MC method for various $N$ and $K$.}
	\label{fig:heatmap_GBSI}
\end{figure}

Figure \ref{fig:heatmap_GBSI} reports $P^{\textnormal{GBS-I}}_{\Haf^2}(N, K) \times 100\%$ for $N = 3,\dots, 6$ and $K = 2, \dots, 8$. 
We observe that the majority of cases approach nearly 100\%, except for $K=2$, the percentage decreases to 18.70\% at $N = 6$. Additionally, 
the analysis of the efficiency gain suggests that GBS-I provides a larger efficiency gain than MC in all cases. For large $K$, the advantage can reach up to $10^{10}$ times greater. This behavior -- GBS-I outperforms MC for large $K$, mirroring the behavior seen in GBS-P-- suggests that both methods are particularly effective in capturing contributions from the higher-order terms in the formal power series expansion of $f$. The GBS-based estimators and MC provide complementary roles over different regimes of $K$.
Building on this insight, we propose a hybrid strategy for Gaussian expectation problems. Using the partition strategy in \eqref{eq:partition}, one can apply the standard MC method for problems with small $K$, where its simplicity and efficiency remain competitive. Then, switch to GBS estimators for large $K$. 

The rest of the paper is organized as follows. In Section \ref{sec:gbsdist}, we review the GBS sampling distribution and provide an explicit formula for the average photon number in terms of the eigenvalues of $B$. In Section \ref{sec:algorithm}, we describe the improved GBS methods in detail and prove Theorem \ref{thrm:asymp}. The precise definition of \eqref{eq:pp} and its numerical approximation algorithm is provided in Section \ref{sec:percentage}.

	\section{The GBS distribution and mean photon numbers}
\label{sec:gbsdist}
Let $B$ be a real, symmetric, positive definite matrix given as in \eqref{eq:I}. Suppose the eigenvalues $\lambda_1 \geq \dots \geq \lambda_N$ of $B$ are strictly less than 1. Then, GBS can be programmed to sample from the following distribution \cite{hamilton2017gaussian, kruse2019detailed}:
\begin{align}
	P_B(I) = \frac{d}{I!} \Haf(B_I)^2, ~~~~~ I = (i_1, \dots, i_N) \in \mathbb{N}^N
	\label{eq:gbsfirstappear}
\end{align} 
where
\begin{align}
	I! &= i_1! i_2! \dots i_N!, \notag \\
	d &= \prod_{n = 1}^N \sqrt{1 - \lambda_n^2}. \label{eq:ddef}
\end{align}
The $N$-tuples $I$'s correspond to the photon patterns output by the GBS device. 
For a specific pattern $I = (i_1, \dots, i_N)$, the total number of photons in $I$ is given by
\begin{align*}
	\vert I \vert = i_1 + i_2 + \dots + i_N.
\end{align*}
We are interested in the \textit{average photon number} of samples from $P_B$, that is
\begin{align*}
	m_B = \mathbb{E}_{I \sim P_B}[\vert I \vert].
\end{align*}
From Lemma \ref{lem:meanphoton}, $m_B$ can be computed explicitly from the eigenvalues of $B$, that is
\begin{align}
	m_B =\sum_{n = 1}^N \frac{\lambda_n^2}{1- \lambda_n^2}. \label{eq:meanphoton}
\end{align}

Consider $tB$ with $t \in (0, \lambda_1^{-1})$. It gives rise to a new GBS sampling distribution $P_{tB}$,
\begin{align}
	P_{tB}(I) & = \frac{d_t}{I!} \Haf((tB)_I)^2 \label{eq:gbsfirstappear2}
\end{align}
where
\begin{align*}
	d_t= \prod_{n = 1}^N \sqrt{1 - t^2\lambda_n^2}.
\end{align*}
Let $m_{tB}$ denote the average photon number of $P_{tB}$. In the next proposition, we show that $m_{tB}$ can be set to any positive real number by tuning $t$ over $ (0, \lambda_1^{-1})$.

\begin{proposition}
	\label{lem:meanphoton_t}
	For any $m > 0$, there exists $t \in (0, \lambda_1^{-1})$ such that $m_{tB} = m$.
\end{proposition}

\begin{proof}
	Since $tB$ has eigenvalues $t\lambda_1 \geq \dots \geq t\lambda_N$, we compute that
	\begin{align}
		m_{tB} =\sum_{n=1}^N \frac{t^2\lambda_n^2}{1- t^2\lambda_n^2}. \label{eq:meanphoton2}
	\end{align}
	We further note that $m_{tB}$ as a function of $t$ is surjective onto $(0, \infty)$, and this completes the proof.
\end{proof}

\noindent
In practice, to find such a scalar factor $t$, a simple grid search over $(0, \lambda_1^{-1})$ will give satisfactory results. 

	\section{The GBS algorithms}
\label{sec:algorithm}
\subsection{The probability estimation method: GBS-P}
The first estimator, GBS-P, uses a direct probability estimation. We begin by assuming that there is an efficient method for computing the sign of the matrix hafnians. 
Let $ I_1, \dots, I_n $ be i.i.d. samples from $ P_{tB} $ where $ m_{tB} = 2K$. We further define
\begin{equation*}
	S^{(J)}_n = \sum_{l=1}^n 1_J (I_l), ~~~~~~~~~ 1_J(I) = \begin{cases}
		1 & I = J \\
		0 & I \neq J
	\end{cases}.
\end{equation*}
Next, take a weighted average of the square root of ${S^{(J)}_n}/{n}$ to define
\begin{equation*}
	\mathcal{E}^{\textnormal{GBS-P}}_{n} = \sum_{\vert J \vert = 2K} \alpha_J \sqrt{\tfrac{S^{(J)}_n}{n}},
\end{equation*}
where
\begin{align*}
	\alpha_J &=  a_J t^{-K} \sigma_J \sqrt{\frac{J!}{d_t}}, \\
	\sigma_J &= \textnormal{sign}(\Haf(B_J)).
\end{align*}

We now prove Theorem \ref{thrm:asymp}. 

\begin{proof}[Proof of Theorem \ref{thrm:asymp}]
	Let $p_J = P_{tB}(J) = \frac{d_t}{J!}\Haf((tB)_J)^2$.
	Then,
	\begin{align*}
		\mu_{\Haf} & = \sum_{\vert J \vert = 2K} a_J \Haf(B_J) \\
		& = \sum_{\vert J \vert = 2K} a_J t^{-K} \Haf((tB)_J) \\
		& = \sum_{\vert J \vert = 2K} \alpha_J \sqrt{p_J}.
	\end{align*}
	For notational convenience, we define
	\begin{align}
		e_{J, J'} =  \mathbb{E} \left[ \left( \sqrt{\frac{S^{(J)}_n}{n}} - \sqrt{p_J}\right) \left( \sqrt{\frac{S^{(J')}_n}{n}} - \sqrt{p_{J'}}\right) \right]. \label{eq:ejjdef}
	\end{align}
	Then,
	\begin{align}
		& \mathbb{E} \left[ \vert \mathcal{E}^{\textnormal{GBS-P}}_n  - \mu_{\Haf} \vert^2 \right] \notag \\ 
		= &  \sum_{|J|, |J'|  = 2K } \alpha_J\alpha_{J'}  \mathbb{E} \left[ \left( \sqrt{\tfrac{S^{(J)}_n}{n}} - \sqrt{p_J}\right) \left( \sqrt{\tfrac{S^{(J')}_n}{n}} - \sqrt{p_{J'}}\right) \right] \notag \\
		= & \sum_{|J|, |J'|  = 2K } \alpha_J\alpha_{J'} e_{J, J'} \notag \\
		=&\sum_{\substack{|J|, |J'|  = 2K \\ J \neq J'}} \alpha_J\alpha_{J'} e_{J,J'} + \sum_{|J| = 2K } \alpha_J^2  e_{J, J} \notag  \\
		\sim & \frac{1}{4n} \left(\sum_{\substack{|J|, |J'|  = 2K \\ J \neq J'}} -\alpha_J\alpha_{J'}\sqrt{p_Jp_{J'}} + \sum_{|J| = 2K } \alpha_J^2 (1-p_J) \right) \label{eq:asympgbsp} \\
		= &  \frac{1}{4n} \left( \sum_{|J| = 2K } \alpha_J^2 - \vert  \mu_{\Haf} \vert ^2\right) \notag \\
		= & \frac{1}{4n}\left( \frac{t^{-2K}}{d_t} \sum_{|J| = 2K } a^2_J J! - \vert  \mu_{\Haf} \vert ^2\right) \notag
	\end{align}
	The leading order asymptotics \eqref{eq:asympgbsp} follows from 
	Lemmas \ref{lem:ejj} and \ref{lem:ejjprime}, presented near the end of this section.
\end{proof}

\begin{theorem}
	The GBS-P estimator $\mathcal{E}^{\textnormal{GBS-P}}_n$ solves $\mathcal{I}^\times_{\Haf}(\epsilon, \delta)$ with
	\begin{align*}
		n \geq \frac{	V^{\textnormal{GBS-P}}_{\Haf} }{\delta \epsilon^2 \mu_{\Haf}^2} =:  n^{\textnormal{GBS-P}}_{\Haf}
	\end{align*}
	where
	\begin{align*}
		V^{\textnormal{GBS-P}}_{\Haf} \sim \frac{1}{4}\left( \frac{t^{-2K}}{d_t} \sum_{|J| = 2K } a^2_J J! - \vert  \mu_{\Haf} \vert ^2\right).
	\end{align*}
\end{theorem}

\begin{proof}
It follows from Theorem \ref{thrm:asymp} that $\mathcal{E}^{\textnormal{GBS-P}}_n$ solves $\mathcal{I}^\times_{\Haf}(\epsilon, \delta)$ for large enough $n$, since
\begin{align*}
	\lim_{n \rightarrow \infty} \mathbb{E} \left[ \vert \mathcal{E}^{\textnormal{GBS-P}}_n  - \mu_{\Haf} \vert^2 \right] = 0.
\end{align*}
To determine the required size of $n$, we use Markov's inequality. 
We observe that
\begin{align*}
	P(\vert \mathcal{E}^{\textnormal{GBS-P}}_n  - \mu_{\Haf} \vert > \epsilon \vert \mu_{\Haf}  \vert ) 
	& = P(\vert \mathcal{E}^{\textnormal{GBS-P}}_n  - \mu_{\Haf} \vert^2 > \epsilon^2 \vert \mu_{\Haf}  \vert^2 ) \\
	& < \frac{\mathbb{E} \left[ \vert \mathcal{E}^{\textnormal{GBS-P}}_n  - \mu_{\Haf} \vert^2 \right]}{\epsilon^2 \mu_{\Haf}^2  } \\
	& \sim \frac{1}{n} \frac{1}{\epsilon^2 \mu_{\Haf}^2  } \frac{1}{4}\left( \frac{t^{-2K}}{d_t} \sum_{|J| = 2K } a^2_J J! - \vert  \mu_{\Haf} \vert ^2\right).
\end{align*}
This completes the proof.
\end{proof}

We refer to $n^{\textnormal{GBS-P}}_{\Haf}$ as the \textit{guaranteed sample size} for GBS-P to solve $\mathcal{I}^\times_{\Haf}(\epsilon, \delta)$. 
Finally, we show that the optimal $t$ that minimizes $n^{\textnormal{GBS-P}}_{\Haf}$ is indeed the one that satisfies $m_{tB} = 2K$. 

\begin{proposition}
	\label{coro:optimalt}
	A real value $t_0$ is the unique solution to
	\begin{align*}
		t_0 = \argmin_{t\in (0, \lambda_1^{-1})} \mathbb{E} \left[ \vert \mathcal{E}^{\textnormal{GBS-P}}_n  - \mu_{\Haf} \vert^2 \right]
	\end{align*}
	if and only if $m_{t_0B} = 2K$.  
\end{proposition}
\begin{proof}
	Since $\alpha_J =  a_J t^{-K} \sigma_J \sqrt{\frac{J!}{d_t}}$, we have
	\begin{align}
		\mathbb{E} \left[ \vert \mathcal{E}^{\textnormal{GBS-P}}_n  - \mu_{\Haf} \vert^2 \right] = \frac{t^{-2K}}{{d_t}} \sum_{|J|, |J'|  = 2K } a_J  a_{J'} \sigma_J \sigma_{J'}  \sqrt{J!} \sqrt{J'!} e_{J, J'} \label{eq:gbspmina}
	\end{align}
	where $e_{J, J'}$ is defined as in \eqref{eq:ejjdef}.
	To minimize \eqref{eq:gbspmina} over $t \in (0, \lambda_1^{-1})$, it is equivalent to compute
	\begin{align}
		t_0 = \argmin_{t\in (0, \lambda_1^{-1})} \gamma(t)  \label{eq:gbspminb}
	\end{align}
	where
	\begin{align*}
		\gamma(t) = \log \frac{t^{-2K}}{{d_t}} = \log \frac{t^{-2K}}{\prod_{n=1}^N \sqrt{1 - t^2 \lambda_n^2}} =  -2K \log t - \frac{1}{2} \sum_{n=1}^N \log (1 - t^2 \lambda_n^2).
	\end{align*}
	If $t_0$ is a solution to \eqref{eq:gbspminb}, then it must be a critical point of $\gamma(t)$. In other words,
	\begin{align}
		\gamma'(t_0) = -\frac{2K}{t_0} + t_0 \sum_{n=1}^N \frac{\lambda_n^2}{1 - t_0^2 \lambda_n^2} = 0,
	\end{align}
	and therefore
	\begin{align*}
		m_{t_0B} =\sum_{n=1}^N \frac{t_0^2 \lambda_n^2}{1 - t_0^2 \lambda_n^2} = 2K.
	\end{align*}
	To prove the other direction, we first note that if $0< t < t_0$, then
	\begin{align*}
		\sum_{n=1}^N \frac{t^2 \lambda_n^2}{1 - t^2 \lambda_n^2} = \sum_{n=1}^N \frac{1}{\frac{1}{t^2 \lambda_n^2} - 1} < \sum_{n=1}^N \frac{1}{\frac{1}{t_0^2 \lambda_n^2} - 1} = m_{t_0}B = 2K,
	\end{align*}
	and therefore
	\begin{align*}
		\gamma'(t) t = -2K + \sum_{n=1}^N \frac{t^2 \lambda_n^2}{1 - t^2 \lambda_n^2} < 0
	\end{align*}
	Hence, $\gamma'(t) < 0$. Similarly, we can check that $\gamma'(t) > 0$ for all $t_0 < t < \lambda_1^{-1}$. Therefore, $t_0$ is the unique minimal solution in $(0, \lambda_1^{-1})$.
\end{proof}

Now we state and prove the lemmas required by the proof of Theorem \ref{thrm:GBS-Icov}. 

\begin{lemma}
	\label{lem:ejj}
	Let  $p_J = P_{tB}(J)$ and let $e_{J, J}$ be defined as in \eqref{eq:ejjdef}. Then the large $n$ leading order asymptotics of  $e_{J, J}$ is
	\begin{align*}
		e_{J, J} \sim \frac{1-p_J}{4n}.
	\end{align*}
\end{lemma}

\begin{proof}
	We start by computing that
	\begin{align*}
		e_{J, J} & = \sum_{m=0}^{n}\left( \sqrt{\tfrac{m}{n}} - \sqrt{p_J} \right)^2  P\left(\tfrac{S^{(J)}_n}{n} = m\right).
	\end{align*}
	For convenience, let
	\begin{align*}
		t_m = \left( \sqrt{\tfrac{m}{n}} - \sqrt{p_J} \right)^2  P\left(\tfrac{S^{(J)}_n}{n} = m\right).
	\end{align*}
	Then, of course, 
	\begin{align*}
		e_{J, J}  = \sum_{m=0}^{n} t_m.
	\end{align*}
	Note that
	\begin{align*}
		P\left(\tfrac{S^{(J)}_n}{n} = m\right) = \binom{n}{m} p_J^m (1-p_J)^{n-m}.
	\end{align*}
	At $m = 0$ or $m = n$, we obtain
	\begin{align*}
		t_0 &= p_J(1-p_J)^n \\
		t_n &= (1-\sqrt{p_J})^2 p_J^n.
	\end{align*}
	Both decay exponentially in $n$. For the rest of the terms, we compute that
	\begin{align}
		e'_{J, J}
		&= \sum_{m=1}^{n-1} t_m \notag \\
		& = \frac{1}{n} \sum_{m=1}^{n-1} \left( \sqrt{m} - \sqrt{n p_J} \right)^2 \binom{n}{m} p_J^m (1-p_J)^{n-m}  \notag \\
		& = \frac{1}{n} (1-p_J)^{n} \Gamma(n+1) \sum_{m=1}^{n-1} \frac{\left( \sqrt{m} - \sqrt{n p_J} \right)^2}{\Gamma(m+1) \Gamma(n-m+1)} \left(\frac{p_J} {1-p_J} \right)^m. \label{eq:cauchysum}
	\end{align}
	Here, \eqref{eq:cauchysum} is obtained from expanding the binomial coefficient and rewriting $n!  = \Gamma(n+1)$. 
	
	We now apply techniques similar to those in \cite{andersen2006asymptotics} for analyzing the values of colored Jones polynomials and quantum invariants, particularly using Cauchy's Residue Theorem and the asymptotic analysis of oscillatory integrals.
	We start by applying Cauchy's Residue Theorem to the sum in \eqref{eq:cauchysum}. Let $\gamma_n$ be a counterclockwise closed rectangular path on the complex plane, as illustrated in Figure \ref{fig:gamma_n}. 
	Let us further define $h \in \mathcal{M}(\mathbb{C}_+)$, where $\mathbb{C}_+ = \{ z \in \mathbb{C} \mid \textnormal{Re}(z)>0 \}$, such that
	\begin{align*}
		h(z) = \frac{\left(\sqrt{z} - \sqrt{np_J } \right)^2}{\Gamma(z+1) \Gamma(n-z+1)} \left(\frac{p_J} {1-p_J} \right)^z \cot(\pi z).
	\end{align*}
	The poles of $h(z)$ inside $\gamma_n$ are at $z = 1, 2 \dots, n-1$. Note that these poles are simple and the residue at each pole is given by
	\begin{align*}
		\text{Res}(h, m) = \frac{\left( \sqrt{m} - \sqrt{n p_J} \right)^2}{\Gamma(m+1) \Gamma(n-m+1)} \left(\frac{p_J} {1-p_J} \right)^m \frac{1}{\pi}.
	\end{align*}
	Cauchy's Residue theorem implies that
	\begin{align*}
		\int_{\gamma_n} h(z) \, dz & = 2\pi i  \sum_{m=1}^{n-1} \text{Res}(h, m) \\
		& = 2 i  \sum_{m=1}^{n-1} \frac{\left( \sqrt{m} - \sqrt{n p_J} \right)^2}{\Gamma(m+1) \Gamma(n-m+1)} \left(\frac{p_J} {1-p_J} \right)^m.
	\end{align*}
	Therefore, 
	\begin{align}
		e_{J, J}
		& =  \frac{1}{2in} (1-p_J)^{n} \Gamma(n+1) \int_{\gamma_n} h(z) \, dz. \label{eq:ejj2} 
	\end{align}
	
	\begin{figure}
		\centering
		\includegraphics[width=0.7\linewidth]{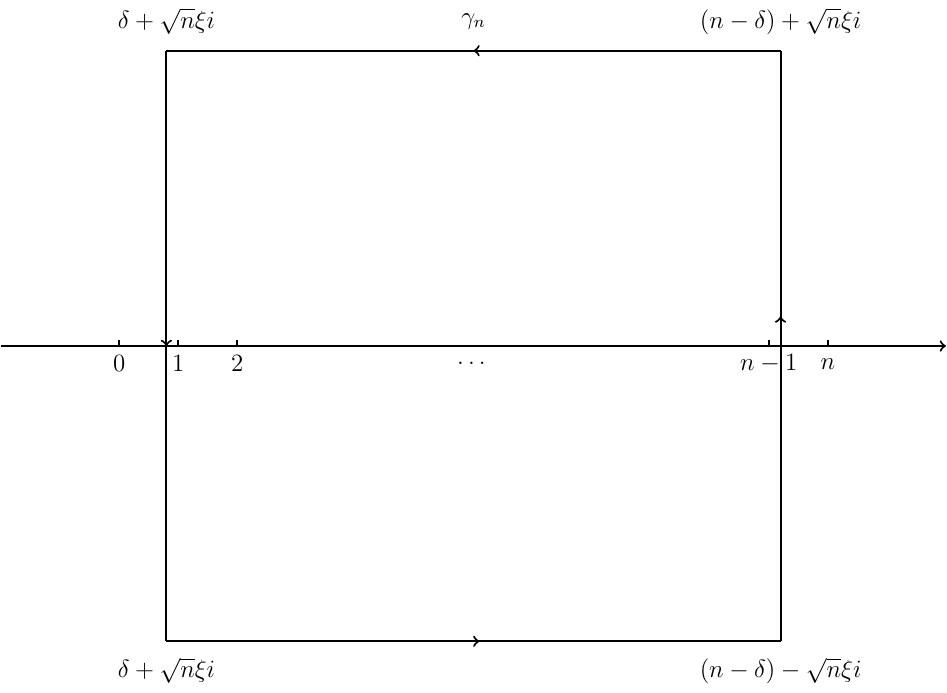}
		\caption{The rectangular contour $\gamma_n$ used in the proof of Lemma \ref{lem:ejj}. Here, $\delta$ and $\xi$ are small positive real numbers which are both less than 1.}
		\label{fig:gamma_n}
	\end{figure}

	We now compute the leading order asymptotics of the contour integral in \eqref{eq:ejj2}. Define
	\begin{align}
		I_n = \int_{\gamma_n} h(z) \, dz.
	\end{align}
	We first make a change of variable $y = \frac{z}{n}$. Let $\hat{\gamma}_n$ be the contour after the change of variable (see Figure \ref{fig:hat_gamma_n}). We then get
	\begin{align*}
		I_n & = n \int_{\hat{\gamma}_n} h(ny) \, dy \\
		& = n \int_{\hat{\gamma}_n}  \frac{\left(\sqrt{ny} - \sqrt{np_J } \right)^2}{\Gamma(ny+1) \Gamma(n-ny+1)}  \left(\frac{p_J} {1-p_J} \right)^{ny}  \cot(n \pi y) \, dy \\
		& =  n^2 \int_{\hat{\gamma}_n}  \frac{\left(\sqrt{y} - \sqrt{p_J } \right)^2}{\Gamma(ny+1) \Gamma(n-ny+1)} e^{\alpha n y} \cot(n \pi y) \, dy,
	\end{align*}
	where
	\begin{align}
		\alpha & = \log \frac{p_J}{1-p_J}. \label{eq:alpha}
	\end{align}
	
	\begin{figure}
		\centering
		\includegraphics[width=0.7\linewidth]{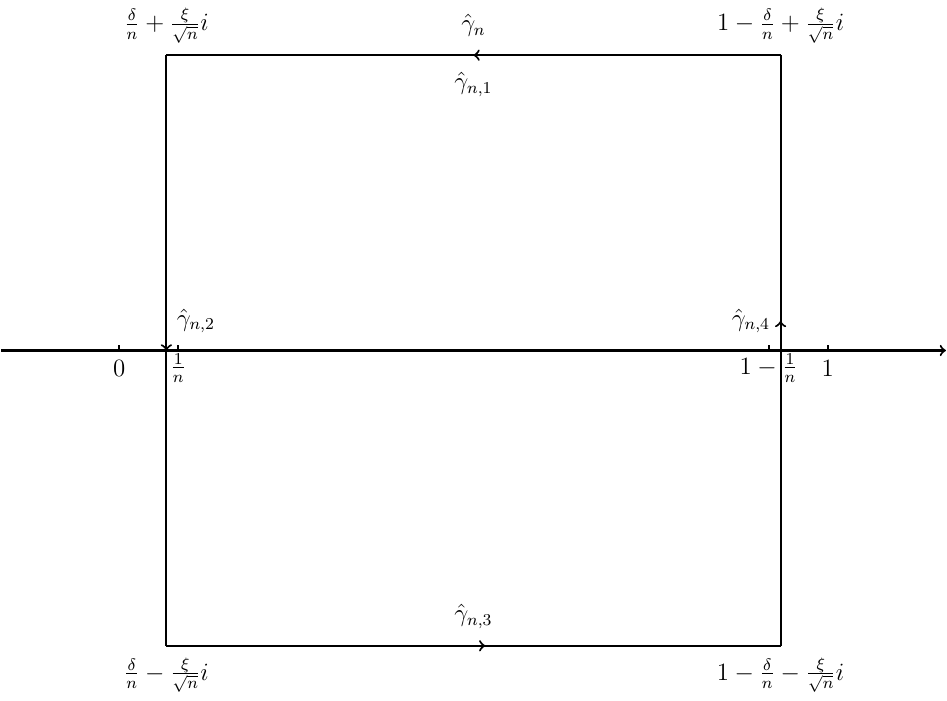}
		\caption{The rectangular contour $\hat{\gamma}_n$ is obtained from $\gamma_n$ via a change of variable.}
		\label{fig:hat_gamma_n}
	\end{figure}
	
	\noindent
	Recall the well-known Stirling's approximation for the Gamma function when $\textnormal{Re}(x) >0$. There exists $C>0$ such that
	\begin{align*}
		 \vert \Gamma(x + 1) - \sqrt{2\pi x} e^{x(\log x - 1)} \vert < C \frac{\vert \sqrt{2\pi x} \vert e^{\textnormal{Re}(x(\log x - 1))}}{\vert x \vert}.
		\end{align*}
	When computing the leading order asymptotics of $I_n$, we can replace $\Gamma(x + 1)$ with the following asymptotics
	\begin{align}
		\Gamma(x + 1) \sim \sqrt{2\pi x} e^{x(\log x - 1)}, \label{eq:gammaapprox}
	\end{align}
	since the error term decays with one more negative power of $n$, and thus can be ignored.
	Plugging \eqref{eq:gammaapprox} into $I_n$ yields
	\begin{align*}
		I_n 
		& \sim  n^2 \int_{\hat{\gamma}_n}  \frac{\left(\sqrt{y} - \sqrt{p_J } \right)^2}{2\pi \sqrt{ny} \sqrt{n(1-y)}} e^{\alpha ny - ny(\log ny - 1) -  n(1-y)(\log n(1-y) - 1)} \cot(n \pi y) \, dy \notag \\
		& = \frac{n}{2\pi} \int_{\hat{\gamma}_n}  \frac{\left(\sqrt{y} - \sqrt{p_J } \right)^2}{\sqrt{y(1-y)}} e^{\alpha ny - ny(\log ny - 1) -  n(1-y)(\log n(1-y) - 1)} \cot(n \pi y) \, dy  \notag \\
		& = \frac{n}{2\pi}  e^{-n \log n} \int_{\hat{\gamma}_n}  \frac{\left(\sqrt{y} - \sqrt{p_J } \right)^2}{\sqrt{y(1-y)}} e^{n[\alpha y - y(\log y - 1) -  (1-y)(\log (1-y) - 1)]} \cot(n \pi y) \, dy
	\end{align*}
	For convenience, let us define
	\begin{align*}
		\psi(y) &= \alpha y - y(\log y - 1) -  (1-y)(\log (1-y) - 1), \\
		\rho(y) & = \frac{\left(\sqrt{y} - \sqrt{p_J } \right)^2}{\sqrt{y(1-y)}},
	\end{align*}
	where $\alpha$ is given as in \eqref{eq:alpha}.
	Then,
	\begin{align}
		I_n &\sim \frac{n}{2\pi}  e^{-n \log n} \int_{\hat{\gamma}_n}  \rho(y) e^{n\psi(y)} \cot(n \pi y) \, dy. \label{eq:ejj4} 
	\end{align}
	
	We now turn to examine the integral
	\begin{align*}
		H_n = \int_{\hat{\gamma}_n}  \rho(y) e^{n\psi(y)} \cot(n \pi y) \, dy.
	\end{align*}
	Notice that the contour $\hat{\gamma}_n$ has four components: $\hat{\gamma}_{n,1}$, $\hat{\gamma}_{n,2}$, $\hat{\gamma}_{n,3}$ and $\hat{\gamma}_{n,4}$ given as in Figure \ref{fig:hat_gamma_n}, and we can evaluate the integral over each component. We define 
	\begin{align*}
		H_n^j = \int_{\hat{\gamma}_{n,j}}  \rho(y) e^{n\psi(y)} \cot(n \pi y) \, dy, ~~~~j = 1,2,3,4.
	\end{align*}
	Then, clearly, 
	\begin{align*}
		H = \sum_{j=1}^4 H_n^j.
	\end{align*}
	
	First, for $H_n^1$, notice that
	\begin{align*}
		H_n^1 = - \int_{-\hat{\gamma}_{n,1}}  \rho(y) e^{n\psi(y)} \cot(n \pi y) \, dy,
	\end{align*}
	where $-\hat{\gamma}_{n,1}$ denotes path with the opposite orientation of $\hat{\gamma}_{n,1}$ and $-\hat{\gamma}_{n,1}$ has the following parametrization
	\begin{align*}
		y = x + \frac{\xi}{\sqrt{n}} i, ~~~~ \frac{\delta}{n}\leq x \leq  1- \frac{\delta}{n}.
	\end{align*}
	From Lemma \ref{lem:cotan}, 
	\begin{align}
		\vert \cot(\pi n y) - (-i) \vert \leq \frac{2}{e^{2\pi n \textnormal{Im}(y)} - 1} = \frac{2}{e^{2\pi \xi \sqrt{n}} - 1}. \label{eq:cotan}
	\end{align}
	Then, we can replace $\cot(n\pi y)$ with following asymptotics
	\begin{align}
		\cot(n \pi y) =  \cot(n \pi x + \sqrt{n} \pi \xi i) \sim -i, 
		\label{eq:cotan2}
	\end{align}
	since the error term decays exponentially with $\sqrt{n}$. We thus obtain
	\begin{align}
		H_n^1 & \sim i \int_{-\hat{\gamma}_{n,1}}  \rho(y) e^{n\psi(y)} \, dy. \label{eq:starthn1}
	\end{align}
	Next, we continuously deform $\hat{\gamma}_{n,1}$ into a new curve as illustrated in Figure \ref{fig:deform}. Since the contribution over the vertical components can be handled together with $H_n^2$ and $H_n^4$, which we will soon prove negligible, we focus on the component over the real line. We then have
	\begin{align}
		\int_{-\hat{\gamma}_{n,1}}  \rho(y) e^{n\psi(y)} \, dy \sim \int^{1- \frac{\delta}{n}}_{\frac{\delta}{n}}  \rho(x) e^{n\psi(x)} \, dx. \label{eq:definite_inte}
	\end{align}
	
	\begin{figure}
		\centering
		\includegraphics[width=0.7\linewidth]{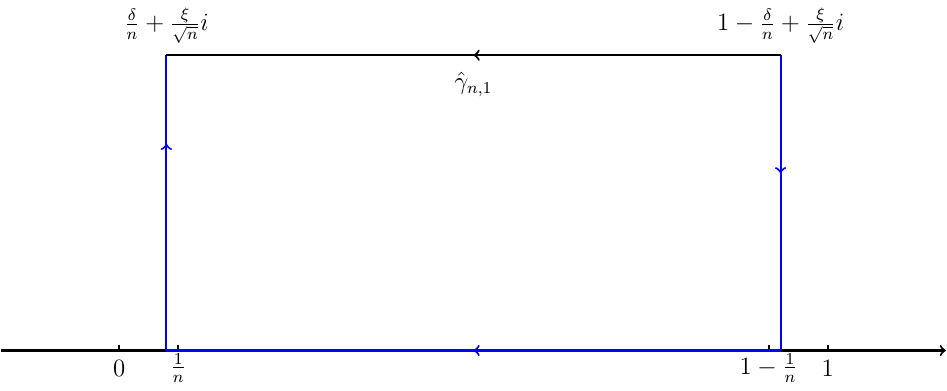}
		\caption{$\hat{\gamma}_{n,1}$ can be continuously deformed into a new curve (shown in blue) with the same orientation and end points. The new curve has two vertical components and a horizontal component contained completely within the real line.}
		\label{fig:deform}
	\end{figure}
	
	To analyze the asymptotic behavior of the real definite integral in \eqref{eq:definite_inte}, we apply the saddle point analysis, which says that the leading order asymptotics are determined by the behavior of the integrand near the critical points of $\psi(x)$. Towards this,
	we find the critical point $x_0$ by solving
	\begin{align*}
		\psi'(x) = \alpha - \log x + \log(1-x) = \log(\frac{p_J}{1-p_J}) - \log(\frac{x}{1-x}) =0,
	\end{align*}
	and we get
	$$x_0 = p_J.$$ 
	Note that $\gamma$ is real along $\left[ \frac{\delta}{n}, 1-\frac{\delta}{n} \right] \subseteq \mathbb{C}.$
	Further, $\psi'(x) >0$ for $x< x_0$, and  $\psi'(x) < 0$ for $x > x_0$. Thus, $\psi(x)$ decays as $x$ moves away from $x_0$. 
	We now expand $\rho(x) e^{n\psi(x)}$ around $x = x_0$. We write
	\begin{align}
		\int_{\frac{\delta}{n}}^{1- \frac{\delta}{n}} \rho(x) e^{n\psi(x)} \, dx = e^{n \psi(x_0)}\int_{\frac{\delta}{n}}^{1- \frac{\delta}{n}} \rho(x) e^{n (\psi(x)- \psi(x_0))} \, dx. \label{eq:integrand}
	\end{align}
	Since $\rho(x)$ vanishes to the second order at $x_0$, we have to go to higher orders in the standard stationary phase approximation to \eqref{eq:integrand}. 
	We thus introduce the following change of variable, where we write $\psi(x) - \psi(x_0)$ as the negative square of a formal power series. Precisely, let
	\begin{align}
		z(x) = \sum_{k = 1}^\infty c_k (x - x_0)^k. \label{eq:zpower1}
	\end{align}
	We require that
	\begin{align}
		-z(x)^2 = \psi(x) - \psi(x_0) =  \sum_{m =2}^\infty \frac{\psi^{(m)}(x_0)}{m!}(x -x_0)^m. \label{eq:zpower2}
	\end{align}
	The Taylor expansion starts only at $m = 2$, since the first derivative of $\psi$ vanishes at $x = x_0$. 
	Now we compute $z(x)^2$ explicitly using the power series expression and get
	\begin{align}
		z(x)^2 & = \sum_{k_1=1}^{\infty}  \sum_{k_2=1}^{\infty} c_{k_1} c_{k_2} (x-x_0)^{k_1 + k_2} \notag \\
		& = \sum_{m=2}^\infty \left(\sum_{k = 1}^{m-1} c_k c_{m-k} \right) (x-x_0)^{m} \label{eq:zpower2b}
	\end{align}
	Combining \eqref{eq:zpower2} and \eqref{eq:zpower2b}, we obtain
	\begin{align*}
		\left(\sum_{k = 1}^{m-1} c_k c_{m-k} \right) = - \frac{\psi^{(m)}(x_0)}{m!}
	\end{align*}
	and this allows us to compute the $c_k$'s explicitly. Particularly, for $m = 2$, we have
	\begin{align*}
		c_1^2 = - \frac{\psi^{(2)}(x_0)}{2!} = \frac{1}{2(1-x_0) x_0}.
	\end{align*}
	Therefore,
	\begin{align*}
		c_1 = \frac{1}{\sqrt{2(1-p_J) p_J}}.
	\end{align*}
	We now define the inversion series of \eqref{eq:zpower1} using
	\begin{align}
		x - x_0 = \sum_{l = 1}^\infty a_l z^l. \label{eq:zpower3}
	\end{align}
	Plugging \eqref{eq:zpower3} into \eqref{eq:zpower1}, we see that
	$a_l$'s must satisfy 
	\begin{align*}
		z & = \sum_{k = 1}^\infty c_k \left( \sum_{l = 1}^\infty a_l z^l \right)^k  = c_1 a_1 z + (c_1 a_2 + c_2 a_1^2) z^2 + (c_1a_3 + 2c_2 a_1 a_2 + c_3a_1^3) z^3 + \dots
	\end{align*}
	Hence, we must have
	\begin{gather*}
		c_1 a_1 = 1 \\
		c_1 a_2 + c_2 a_1^2 = 0 \\
		c_1a_3 + 2c_2 a_1 a_2 + c_3a_1^3 =0 
	\end{gather*}
	and so on and so forth. Specially, we can solve for $a_1$ such that
	\begin{align}
		a_1 = \sqrt{2(1-p_J) p_J}. \label{eq:c1a1}
	\end{align}
	We further note that
	\begin{align*}
		dx = \left(\sum_{l= 1}^\infty a_l l z^{l-1} \right) dz.
	\end{align*} 
	Then, apply the change of variable in \eqref{eq:zpower3} to \eqref{eq:integrand}, the integrand becomes
	\begin{align}
		\rho(x) e^{n (\psi(x) - \psi(x_0))} \, dx 
		& = \left[\rho(x_0) + \sum_{m=1}^\infty \frac{\rho^{(m)}(x_0)}{m!} \left( \sum_{l = 1}^\infty a_l z^l\right)^m \right] e^{-nz^2}   \left(\sum_{l= 1}^\infty a_l l z^{l-1} \right) \, dz \label{eq:dzdx}
	\end{align} 
	Now we organize the terms in \eqref{eq:dzdx} in ascending order of $z^le^{-nz^2}$. Since
	\begin{align*}
		\int_{-\infty}^\infty z^le^{-nz^2} \, dz = \left( \frac{1}{\sqrt{n}}\right)^{l+1} \int_{-\infty}^\infty z^le^{-z^2} \, dz,
	\end{align*}
	and since $\rho(x_0) = \rho'(x_0)=0$,
	\begin{align}
		\int_{\frac{\delta}{n}}^{1- \frac{\delta}{n}} \rho(x) e^{n\psi(x)} \, dx & \sim 
		e^{n \psi(x_0)}  \frac{a_1^3}{n} \frac{\rho^{(2)}(x_0)}{2!} \int_{-\infty}^{\infty}  z^2 e^{-z^2}  \, dz \notag \\
		& =  \frac{\sqrt{\pi}}{8} \left(\sqrt{\frac{2}{n}} \right)^3 (1-p_J) \frac{e^n}{(1-p_J)^n} \label{eq:integral_eval}
	\end{align}
	In \eqref{eq:integral_eval}, we have used
	\begin{gather*}
		\int_{-\infty}^{\infty}  z^2 e^{-z^2}=\sqrt{\pi}/2 \\ \frac{\rho^{(2)}(x_0)}{2!} = \frac{1}{4p_J \sqrt{(1-p_J)p_J}} \\
		e^{n \psi(x_0)} = \frac{e^n}{(1-p_J)^n}.
	\end{gather*}
	In summary, from \eqref{eq:starthn1}, \eqref{eq:definite_inte} and \eqref{eq:integral_eval}, we obtain
	\begin{align*}
		H_n^1 \sim i \frac{\sqrt{\pi}}{8} \left(\sqrt{\frac{2}{n}} \right)^3 (1-p_J) \frac{e^n}{(1-p_J)^n}.
	\end{align*}
	The computation for $H_n^3$ is completely analogous. 
	On $\hat{\gamma}_{n,3}$, we use $i$ (instead of $-i$) to approximate for $\cot(n \pi y)$. With an additional minus sign coming from the parametrization of $\hat{\gamma}_{n,3}$, which is in the opposite direction of $\hat{\gamma}_{n,1}$, we get the same leading order asymptotic for $H_n^3$. 
	
	The contribution from $H_n^2$ and $H_n^4$ can be ignored. To see this, we parametrize $\hat{\gamma}_{n,2}$ by
	\begin{align*}
		y = \frac{\delta}{n} - t i, ~~~~ -\frac{\xi}{\sqrt{n}} \leq t \leq \frac{\xi}{\sqrt{n}}.
	\end{align*}
	From the triangle inequality,
	\begin{align*}
		\vert H_n^2  \vert & \leq \int_{\hat{\gamma}_{n,2}}  \vert \rho(y) \vert e^{n\textnormal{Re}(\psi(y))}  \vert \cot(n \pi y) \vert \, dy \\
		& \leq \max_{\hat{\gamma}_{n,2}} \vert \rho(y) \vert e^{n\textnormal{Re}(\psi(y))}  \vert \cot(n \pi y) \vert \int_{\hat{\gamma}_{n,2} } 1 \, \vert dy \vert.
	\end{align*}
	Of course,
	\begin{align*}
		\int_{\hat{\gamma}_{n,2} } 1 \, \vert dy \vert = \frac{2\xi}{\sqrt{n}}.
	\end{align*}
	We focus on bounding the function $\vert \rho(y) \vert e^{n\textnormal{Re}(\psi(y))}  \vert \cot(n \pi y) \vert$. 
	It is enough to show that there exists positive constants $C_0$ and $\theta <1$ such that
	\begin{align}
		 \max_{\hat{\gamma}_{n,2}} \vert \rho(y) \vert e^{n\textnormal{Re}(\psi(y))}  \vert \cot(n \pi y) \vert \leq C_0 e^{\pi \sqrt{n} \xi} \frac{e^n \theta^n}{(1-p_J)^n}, \label{eq:triangleuseful1}
	\end{align}
	since \eqref{eq:triangleuseful1} clearly implies that $\vert H_n^2  \vert$ decays exponentially faster than $H_n^1$. Let $y \in \hat{\gamma}_{n,2}$. From Lemma \ref{lem:cotan2}, 
	\begin{align*}
		\vert \cot(n\pi y) \vert \leq  1 + \tfrac{2}{\sqrt{1-\cos(2\pi \delta)^2}}.
	\end{align*}
	We further compute that
	\begin{align*}
		\vert \rho(y) \vert \leq \frac{\left( \sqrt{\textnormal{Re}(y) } + \sqrt{p_J} \right)^2}{\sqrt{\textnormal{Re}(y) } \sqrt{\textnormal{Re}(1-y) }}= \frac{\left( \left( \frac{\delta^2 }{n^2} + \frac{\xi^2}{n} \right)^{\frac{1}{4}} + \sqrt{p_J}\right)^2}{\sqrt{\tfrac{\delta}{n}} \sqrt{1- \tfrac{\delta}{n}}}
	\end{align*}
	where the right hand side is bounded above by a polynomial of $n$. Then, to prove \eqref{eq:triangleuseful1}, it suffices to show that for $n$ large enough there exists positive constants $C_1$ and $\theta <1$ such that
	\begin{align}
		 e^{n\textnormal{Re}(\psi(y))} \leq C_1 e^{\pi \sqrt{n} \xi} \frac{e^n \theta^n}{(1-p_J)^n}. \label{eq:triangleuseful2}
	\end{align}
	From Lemma \ref{lem:realpart}, we get
	\begin{align*}
		\textnormal{Re}(\psi(y)) &\leq \alpha \textnormal{Re}(y) - \textnormal{Re}(y)\log(\textnormal{Re}(y)) - \textnormal{Re}(1-y)\log(\textnormal{Re}(1-y)) \\
		& \quad + \frac{\pi}{2} \vert \textnormal{Im}(y) \vert + \frac{\pi}{2} \vert \textnormal{Im}(1-y) \vert + 1.
	\end{align*}
	Since $\textnormal{Re}(y) = \frac{\delta}{n}$, 
	\begin{align*}
		\textnormal{Re}(\psi(y)) &\leq  \alpha \frac{\delta}{n} - \frac{\delta}{n}\log(\frac{\delta}{n}) - \left(1 - \frac{\delta}{n} \right) \log\left(1 - \frac{\delta}{n} \right) \\
		& \quad + \pi \frac{\xi}{\sqrt{n}} + 1.
	\end{align*}
	Hence,
	\begin{align*}
		e^{n\textnormal{Re}(\psi(y))} \leq \frac{e^{\alpha \delta + \pi \sqrt{n} \xi} e^n}{ \left( \frac{\delta}{n}\right)^{\delta} \left(1- \tfrac{\delta}{n} \right)^{n - \delta}},
	\end{align*}
	which clearly satisfies \eqref{eq:triangleuseful2}. Therefore, $H_n^2$ is negligible.
	Using the same kind of techniques, we establish that $H_n^{4}$ can also be ignored. 
	
	In total, we have
	\begin{align}
		H_n \sim  i  \frac{\sqrt{\pi}}{4} \left(\sqrt{\frac{2}{n}} \right)^3 (1-p_J) \frac{e^n}{(1-p_J)^n}. \label{eq:ejj5} 
	\end{align}
	Inserting \eqref{eq:ejj5} into \eqref{eq:ejj4}, we obtain
	\begin{align*}
		I_n \sim \frac{i}{2} \frac{1}{\sqrt{2\pi n}}
		  e^{-n \log n} (1-p_J) \frac{e^n}{(1-p_J)^n}.
	\end{align*}
	Finally, from \eqref{eq:ejj2}, we get
	\begin{align*}
		e_{J, J} & \sim \frac{1}{2in} (1-p_J)^{n} \Gamma(n+1) \\
		& \sim \frac{1-p_J}{4n}.
	\end{align*}
\end{proof}

\begin{lemma}
	\label{lem:ejjprime}
	Let $e_{J, J'}$ be defined as in \eqref{eq:ejjdef}
	with $J \neq J'$. Then the large $n$ leading order asymptotics of $e_{J, J'}$ is
	\begin{align}
		e_{J, J'} \sim -\frac{\sqrt{p_J p_{J'}} }{4n}. \label{eq:ejjasym}
	\end{align}
\end{lemma}

\begin{proof}
	We begin by noticing that
	\begin{align*}
		e_{J, J'}
		& = \sum_{\substack{m_1, m_2 = 0 \\m_1 + m_2 \leq n}} \left( \sqrt{\tfrac{m_1}{n}} - \sqrt{p_J} \right) \left( \sqrt{\tfrac{m_2}{n}} - \sqrt{p_{J'}} \right) P(S^{(J)}_n = m_1, S^{(J')}_n = m_2),
	\end{align*}
	and
	\begin{align*}
		& P(S^{(J)}_n = m_1, S^{(J')}_n = m_2) =  \binom{n}{m_1} \binom{n-m_1}{m_2} p_J^{m_1} p_{J'}^{m_2}  (1-p_J-p_{J'})^{n-m_1-m_2} \\
		& = \frac{(1-p_J-p_{J'})^{n} \Gamma(n+1) }{\Gamma(m_1+1) \Gamma(m_2 +1) \Gamma(n-m_1-m_2+1)} \left(\frac{p_J}{1-p_J - p_{J'}} \right)^{m_1} \left(\frac{p_{J'}}{1-p_J - p_{J'}} \right)^{m_2}.
	\end{align*}
	For simplicity of notation, let us define
	\begin{align*}
		&\alpha = \log(\frac{p_J}{1-p_J - p_{J'}}), \\
		&\beta = \log(\frac{p_{J'}}{1-p_J - p_{J'}}),
	\end{align*}
	and
	\begin{align}
		t_{m_1, m_2} = \frac{(1-p_J-p_{J'})^{n}  \Gamma(n+1)}{n} \, \frac{\left( \sqrt{m_1} - \sqrt{n p_J} \right) \left( \sqrt{m_2} - \sqrt{n p_{J'}} \right)  e^{\alpha m_1 + \beta m_2}}{\Gamma(m_1+1) \Gamma(m_2 +1) \Gamma(n-m_1-m_2+1)}. \label{eq:tm1m2}
	\end{align}
	Clearly, we have
	\begin{align*}
		e_{J, J'}
		& = \sum_{\substack{m_1, m_2 = 0 \\m_1 + m_2 \leq n}} 	t_{m_1, m_2}.
	\end{align*}
	
	We now decompose the sum into the following components according to their asymptotic behavior
	\begin{align*}
		e_{J, J'} =
		& t_{0,0} + t_{0, n} + t_{n, 0} + \sum_{\substack{m_1 = 0 \\1 \leq m_2 \leq n-1}} t_{m_1, m_2} + \sum_{\substack{m_2 = 0 \\1 \leq m_1 \leq n-1}} t_{m_1, m_2} + \sum_{(m_1, m_2) \in P_n} t_{m_1, m_2}
	\end{align*}
	where
	\begin{align*}
		P_n = \{ (m_1, m_2) \in \mathbb{Z}_+ \mid 1 \leq m_1 \leq n-2, 1 \leq m_2 \leq n-2, m_1 + m_2 < n \}.
	\end{align*}
	The terms $ t_{0,0}, t_{0, n}, t_{n, 0}$ decay exponentially with $n$. This follows from the facts that
	\begin{align*}
		t_{0,0} &= (1-p_J-p_{J'})^{n} \sqrt{p_J p_{J'}} \\
		t_{0,n} &= p_{J'}^n  \sqrt{p_J} ( \sqrt{p_{J'}} - 1) \\
		t_{n,0} &= p_{J}^n \sqrt{p_{J'}} ( \sqrt{p_J} - 1). 
	\end{align*}
	The terms $\displaystyle \sum_{\substack{m_1 = 0 \\1 \leq m_2 \leq n-1}} t_{m_1, m_2}$ and $\displaystyle \sum_{\substack{m_2 = 0 \\1 \leq m_1 \leq n-1}} t_{m_1, m_2}$ are effectively one-dimensional. By contour integral method similar to the proof of Lemma \ref{lem:ejjprime}, these terms also exhibit exponential decay in $n$. For details, see Lemma \ref{lem:ejjprimeplace1} and Corollary \ref{lem:ejjprimeplace2}. 
	The non-trivial contribution to the leading order asymptotics come from the two-dimensional sum  over $P_n$. The calculation proceeds similarly to previous cases except we now get two dimensional contour integrals. 
	This suggests that we define
	$$T_n = {\gamma}_n \times {\gamma}_n,$$
	where $\gamma_n$ is given as in Figure \ref{fig:gamma_n}.
	Clearly, 
	\begin{align*}
		T_n \subseteq \mathbb{C}_+^2 - \left( (\mathbb{C}_+ \times \mathbb{Z}_+) \cup (\mathbb{Z}_+ \times \mathbb{C}_+) \right).
	\end{align*}
	Consider $h \in \mathcal{M}(\mathbb{C}_+^2)$ such that
	\begin{align}
		h(x,y) = \frac{\left( \sqrt{x} - \sqrt{n p_J} \right) \left( \sqrt{y} - \sqrt{n p_{J'}} \right)  e^{\alpha x + \beta y}}{\Gamma(x+1) \Gamma(y +1) \Gamma(n-x-y+1)} \cot(\pi x) \cot(\pi y). \label{def:ejjprime_hfun}
	\end{align}
	and let
	\begin{align}
		h_n(x,y) = h(nx, ny). \label{def:ejjprime_hnfun}
	\end{align}
	Define
	\begin{align*}
		I_n = \int_{T_n} h(x,y) \dd x \dd y.
	\end{align*}
	Apply Cauchy's residue theorem one variable at a time, and we obtain that
	\begin{align*}
		I_n = & \int_{\gamma_n} \int_{\gamma_n} h(x,y) dx dy \notag \\
		= & {2i}\int_{\gamma_n} \sum_{m_2 =1}^{n-1}  \frac{\left( \sqrt{x} - \sqrt{n p_J} \right) \left( \sqrt{m_2} - \sqrt{n p_{J'}} \right)  e^{\alpha x + \beta m_2}}{\Gamma(x+1) \Gamma(m_2 +1) \Gamma(n-x-m_2+1)} \cot(\pi x) \\
		= & {(2i)^2} \sum_{m_1 =1}^{n-1} \sum_{m_2 =1}^{n-1} \frac{\left( \sqrt{m_1} - \sqrt{n p_J} \right) \left( \sqrt{m_2} - \sqrt{n p_{J'}} \right)  e^{\alpha m_1 + \beta m_2}}{\Gamma(m_1+1) \Gamma(m_2 +1) \Gamma(n-m_1- m_2+1)}
	\end{align*}
	Note that, for all $m_2 + m_2 > n$,
	\begin{align*}
		\frac{1}{\Gamma(n-m_1-m_2 + 1)} = 0.
	\end{align*}
	We thus get that
	\begin{align*}
		I_n = & {(2i)^2} \sum_{\substack{1 \leq m_1 \leq n- 1 \\m_2 = n-m_1}} \frac{\left( \sqrt{m_1} - \sqrt{n p_J} \right) \left( \sqrt{m_2} - \sqrt{n p_{J'}} \right)  e^{\alpha m_1 + \beta m_2}}{\Gamma(m_1+1) \Gamma(m_2 +1) \Gamma(n-m_1- m_2+1)}.
	\end{align*}
	Therefore, 
	\begin{align}
		e_{J, J'} \sim \sum_{(m_1, m_2) \in P_n} t_{m_1, m_2} = \frac{1}{(2i)^2} \frac{(1-p_J-p_{J'})^{n}  \Gamma(n+1)}{n} \, I_n. \label{eq:ejjprimelast}
	\end{align}
	
	It remains to compute the leading order asymptotics for $I_n$. Towards this, we would like to use Stirling's approximation to give asymptotics for $\Gamma(x+1)$, $\Gamma(y+1)$ and $\Gamma(n-x-y+1)$. However, there exists $(x,y)\in T_n$ such that $\textnormal{Re}(n-x-y) <0$, and hence Stirling's approximation does not apply. To address this, we divide $T_n$ into two parts based on the applicability of Stirling's approximation and apply different strategies to each part. 
	For convenience, we first perform a change of variable $x \mapsto nx$ and $y \mapsto ny$ to get
	\begin{align*}
		I_n = n^2 \int_{\hat{T}_n} h_n(x, y) \, dx dy.
	\end{align*}
	Here,
	\begin{align*}
		\hat{T}_n = \hat{\gamma}_n \times \hat{\gamma}_n,
	\end{align*}
	and $ \hat{\gamma}_n$ is defined as in Figure \ref{fig:hat_gamma_n}. 
	Since $\hat{T}_n$ consists of the following faces
	\begin{align*}
		\hat{T}_n = \bigcup_{i,j=1,2,3,4} \hat{\gamma}_{n, i} \times \hat{\gamma}_{n,j},
	\end{align*}
	we get that
	\begin{align*}
		\int_{\hat{T}_n} h_n(x, y) \, dx dy = \sum_{i = 1}^4  \sum_{j = 1}^4 \int_{ \hat{\gamma}_{n, i} \times \hat{\gamma}_{n,j}} h_n(x, y) \, dx dy.
	\end{align*}
	The non-trivial contribution to the leading order asymptotics comes from the integral over $\hat{\gamma}_{n, i} \times \hat{\gamma}_{n, j}$ where $i,j \in \{1,3\}$, and the rest are negligible as we will prove below.
	
	We focus on the integral over $\hat{\gamma}_{n, i} \times \hat{\gamma}_{n, j}$ with $i = j = 1$.
	First, let $\phi(x,y) \in \mathcal{M}(\mathbb{C}_+^2)$ be defined by
	\begin{align}
		\phi(x,y) = \frac{\left( \sqrt{nx} - \sqrt{n p_J} \right) \left( \sqrt{ny} - \sqrt{n p_{J'}} \right)  e^{\alpha nx + \beta ny}}{\Gamma(nx+1) \Gamma(ny +1) \Gamma(n-nx-ny+1)}. \label{def:ejjprime_phifun}
	\end{align}
	From Lemma \ref{lem:cotan}, we see that
	\begin{align*}
		\int_{\hat{\gamma}_{n, 1} \times \hat{\gamma}_{n, 1}} h_n(x, y) \, dx dy & = \int_{\hat{\gamma}_{n, 1} \times \hat{\gamma}_{n, 1}} \phi(x, y) \cot(n\pi x) \cot(n\pi y) \, dx dy \\
		& \sim (-i)^2 \int_{\hat{\gamma}_{n, 1} \times \hat{\gamma}_{n, 1}} \phi(x, y) \, dx dy,
	\end{align*}
	Let us further define
	\begin{align*}
		W_{n,0} &= \left\{(x,y) \in \mathbb{R}^2 \mid \frac{\delta}{n} \leq x,y \leq 1 - \frac{\delta}{n}  \right\}, \\
		W_{n,1} & = \left\{ (x,y) \in \mathbb{C}^2 \mid x = \frac{\delta}{n} + si, \, y = t + si, \right . \\ 
		& \hspace{1in} \left. \frac{\delta}{n} \leq t \leq 1 - \frac{\delta}{n}, \,  0 \leq s \leq \frac{\xi}{\sqrt{n}}  \right\}, \\
		W_{n,2} & = \left\{ (x,y) \in \mathbb{C}^2 \mid x = t + si, \, y = \frac{\delta}{n} + si, \right . \\ 
		& \hspace{1in} \left. \frac{\delta}{n} \leq t \leq 1 - \frac{\delta}{n}, \,  0 \leq s \leq \frac{\xi}{\sqrt{n}}  \right\}, \\
		W_{n,3} & = \left\{ (x,y) \in \mathbb{C}^2 \mid x = 1 -\frac{\delta}{n} + si, \, y = t + si, \right . \\ 
		& \hspace{1in} \left. \frac{\delta}{n} \leq t \leq 1 - \frac{\delta}{n}, \,  0 \leq s \leq \frac{\xi}{\sqrt{n}}  \right\}, \\
		W_{n,4} & = \left\{ (x,y) \in \mathbb{C}^2 \mid x = t + si, \, y = 1 -\frac{\delta}{n} + si, \right . \\ 
		& \hspace{1in} \left. \frac{\delta}{n} \leq t \leq 1 - \frac{\delta}{n}, \,  0 \leq s \leq \frac{\xi}{\sqrt{n}}  \right\},
	\end{align*}
	Since $\phi(x,y)$ is holomorphic for all $(x,y) \in \mathbb{C}$ such that $\textnormal{Re}(x) >0$ and $\textnormal{Re}(y)>0$, and since $\hat{\gamma}_{n, 1} \times \hat{\gamma}_{n, 1} $ is homologous to $W_n = \displaystyle \bigcup_{l=0}^4 W_{n,l}$ relative to their common boundary. 
	we obtain that
	\begin{align*}
		\int_{\hat{\gamma}_{n, 1} \times \hat{\gamma}_{n, 1}} \phi(x, y) \, dx dy = \sum_{l = 0}^4 \int_{W_{n,l}} \phi(x, y) \, dx dy.
	\end{align*}
	We further divide $W_{n,0}$ into a lower triangular piece $L^-_n$, a thin strip $U_n$ along the diagonal, and an upper triangular piece $L^+_n$. That is
	\begin{align*}
		W_{n,0} = L^-_n \cup U_n \cup L^+_n,
	\end{align*}
	where
	\begin{align*}
		L^-_n &=  \left\{(x,y) \in W_{n,0} \mid x + y \leq 1- \frac{1}{2n} \right\}, \\
		U_n &= \left\{(x,y) \in W_{n,0} \mid 1- \frac{1}{2n} < x + y \leq 1+ \frac{1}{n} \right\}, \\
		L^+_n &=  \left\{(x,y) \in W_{n,0} \mid x + y > 1+ \frac{1}{n} \right\}. \\
	\end{align*}
	Therefore, 
	\begin{align*}
		\int_{W_{n,0}} \phi(x, y) \, dx dy   = \int_{L^-_n} \phi(x, y) \, dx dy + \int_{U_n} \phi(x, y) \, dx dy + \int_{L^+_n} \phi(x, y) \, dx dy .
	\end{align*}
	
	Below we compute the leading order asymptotics of the integral over $L^-_n$.
	For $(x,y) \in L^-_n$, the Stirling's approximation of the Gamma function \eqref{eq:gammaapprox} yields
	\begin{align*}
		&\Gamma(nx+1) \Gamma(ny +1) \Gamma(n-nx-ny+1) \\
		& \hspace{.1in} \sim (\sqrt{2\pi})^3(\sqrt{n})^3 \sqrt{x y (1-x-y)} e^{n\log n - n} e^{n(x\log x + y \log y + (1-x-y)\log(1-x-y)}.
	\end{align*}
	We therefore obtain
	\begin{align}
		\int_{L^-_n} \phi(x, y) \, dx dy \sim e^{n-n\log n} \left(\frac{1}{2 \pi}\right)^{\frac{3}{2}} \frac{1}{\sqrt{n}} \int_{L^-_n}  \rho(x,y) e^{n \psi(x,y)} \, dx dy, \label{eq:inteshortcuta}
	\end{align}
	where
	\begin{align*}
		\psi(x,y) &= \alpha x + \beta y - x\log x - y \log y - (1-x-y)\log(1-x-y), \\
		\rho(x,y) &= \frac{\left( \sqrt{x} - \sqrt{p_J} \right) \left( \sqrt{y} - \sqrt{p_{J'}} \right)}{ \sqrt{x y (1-x-y)} }.
	\end{align*}
	Furthermore,
	\begin{align*}
		\int_{L^-_n}  \rho(x,y) e^{n \psi(x,y)} \, dx dy = \int_{\frac{\delta}{n}}^{1 - \frac{1}{2n}} \int_{\frac{\delta}{n}}^{1 - \frac{1}{2n} -y}  \rho(x, y) e^{n\psi(x, y)} \, dx dy.
	\end{align*}
	We proceed by applying the saddle point method to the integral above. Note that $\psi$ is real on $W_{n, 0}$.
	The critical point $(x_0, y_0)$ of $\psi(x,y)$ is obtained from solving the following system of equations
	\begin{align*}
		\frac{\partial \psi}{\partial x}= \alpha - \log(x) + \log(1-x-y) = 0, \\
		\frac{\partial \psi}{\partial y}= \beta - \log(y) + \log(1-x-y)  =  0.
	\end{align*}
	The system has a unique solution $(x_0, y_0) = (p_J, p_{J'}).$ Since $p_J + p_{J'} \leq 1$, $(x_0, y_0) \in L^-_n.$
	Furthermore, the Hessian matrix is given by
	\begin{align*}
		H = \begin{bmatrix}
			\frac{\partial^2 \psi}{\partial x^2} &  \frac{\partial \psi}{\partial x \partial y} \\
			\frac{\partial \psi}{\partial y \partial x} & \frac{\partial^2 \psi}{\partial y^2}
		\end{bmatrix} = 
		\begin{bmatrix}
			-\frac{1}{x} - \frac{1}{1-x-y} &  - \frac{1}{1-x-y} \\
			- \frac{1}{1-x-y} & -\frac{1}{y} - \frac{1}{1-x-y}
		\end{bmatrix}.
	\end{align*}
	Let $\lambda_1$ and $\lambda_2$ denote the eigenvalues of $H$. Since for all $x, y \in L^-_n$, 
	\begin{align*}
		\lambda_1 \lambda_2 = \det(H) = \frac{1 - y - x}{xy (1-x-y)^2} >0,
	\end{align*}
	and since
	\begin{align*}
		\lambda_1 + \lambda_2 = \textnormal{Tr}(H) =  -\frac{1}{x} -\frac{1}{y} - \frac{2}{1-x-y} < 0,
	\end{align*}
	both $\lambda_1$ and $\lambda_2$ are strictly negative, and hence the Hessian $H$ is negative definite over $L^-_n$. 
	Therefore, the function $\psi(x,y)$ attains its maximum value at the critical point  $(x_0, y_0) = (p_J, p_{J'})$.
	
	Let
	\begin{align*}
		z = (z_1, z_2)
	\end{align*}
	be new coordinates on $L_n^{-}$ near $(x_0, y_0)$, 
	where
	\begin{align}
		z_i(x, y) = \sum_{l_1, l_2 = 0}^\infty c^{(i)}_{l_1, l_2} (x-x_0)^{l_1}(y-y_0)^{l_2}, ~~~~~ i = 1, 2. \label{eq:z1z2}
	\end{align}
	Let
	\begin{align*}
		\Sigma = \begin{bmatrix}
			-\frac{\partial^2 \psi}{\partial x^2}(x_0, y_0) &  -\frac{\partial \psi}{\partial x \partial y} (x_0, y_0) \\
			-\frac{\partial \psi}{\partial y \partial x} (x_0, y_0) & -\frac{\partial^2 \psi}{\partial y^2}(x_0, y_0)
		\end{bmatrix} = \begin{bmatrix}
			\frac{1- p_{J'}}{p_J(1-p_J-p_{J'})} & \frac{1}{1-p_J-p_{J'}} \\
			\frac{1}{1-p_J-p_{J'}} & \frac{1- p_{J}}{p_{J'}(1-p_J-p_{J'})}
		\end{bmatrix}.
	\end{align*}
	We require that 
	\begin{align}
		-\frac{1}{2}z^\intercal \Sigma z  = \psi(x, y) - \psi(x_0, y_0). \label{eq:intermediatecomparison}
	\end{align}
	Taylor expansion of the right-hand side of \eqref{eq:intermediatecomparison} yields
	\begin{align*}
		&\psi(x, y) - \psi(x_0, y_0) = \frac{\psi_{xx}(x_0, y_0)}{2} (x-x_0)^2 \\
		&\hspace{0.2in} + \psi_{xy}(x_0, y_0)(x-x_0)(y-y_0)+ \frac{\psi_{yy}(x_0, y_0)}{2} (y-y_0)^2 + \dots.
	\end{align*}
	For simplicity, we write $\psi_{xx} = \psi_{xx}(x_0, y_0)$, $\psi_{xy} = \psi_{xy}(x_0, y_0)$ and $\psi_{yy}=\psi_{yy}(x_0, y_0)$. 
	On the left hand side of \eqref{eq:intermediatecomparison}, we have
	\begin{align*}
		-\frac{1}{2}z^\intercal \Sigma z & = \frac{\psi_{xx}}{2} z_1^2 + \psi_{xy}z_1z_2+ \frac{\psi_{yy}}{2} z_2^2  \\
		& = \frac{\psi_{xx}}{2} \left(\sum_{l_1, l_2 = 0}^\infty c^{(1)}_{l_1, l_2} (x-x_0)^{l_1}(y-y_0)^{l_2}\right)^2+  \\
		&~~~~ \psi_{xy} \left(\sum_{l_1, l_2 = 0}^\infty c^{(1)}_{l_1, l_2} (x-x_0)^{l_1}(y-y_0)^{l_2}\right)\left(\sum_{l_1, l_2 = 0}^\infty c^{(2)}_{l_1, l_2} (x-x_0)^{l_1}(y-y_0)^{l_2}\right)+  \\
		&~~~~ \frac{\psi_{yy}}{2} \left(\sum_{l_1, l_2 = 0}^\infty c^{(2)}_{l_1, l_2} (x-x_0)^{l_1}(y-y_0)^{l_2}\right)^2  \\
		& = \frac{\psi_{xx}}{2} \left(c^{(1)}_{0,0}\right)^2 + \psi_{xy}c^{(1)}_{0,0} c^{(2)}_{0,0}+ \frac{\psi_{yy}}{2} \left(c^{(2)}_{0,0}\right)^2 + \\
		&~~~~ \left( \frac{\psi_{xx}}{2} c^{(1)}_{0,0} c^{(1)}_{1,0} + \psi_{xy} c^{(1)}_{0,0}c^{(2)}_{1,0} + \psi_{xy} c^{(1)}_{1,0} c^{(2)}_{0,0} + \frac{\psi_{yy}}{2} c^{(2)}_{0,0} c^{(2)}_{1,0} \right) (x-x_0) + \\
		&~~~~ \left( \frac{\psi_{xx}}{2} c^{(1)}_{0,0} c^{(1)}_{0,1} + \psi_{xy} c^{(1)}_{0,0}c^{(2)}_{0,1} + \psi_{xy} c^{(1)}_{1,0} c^{(2)}_{0,0} + \frac{\psi_{yy}}{2} c^{(2)}_{0,0} c^{(2)}_{0, 1} \right) (y-y_0) + \\
		&~~~~ \left(\frac{\psi_{xx}}{2} \left(c^{(1)}_{1,0}\right)^2 + \psi_{xy}c^{(1)}_{1,0} c^{(2)}_{1,0}+ \frac{\psi_{yy}}{2} \left(c^{(2)}_{1,0}\right)^2 \right) (x-x_0)^2 + \\
		&~~~~ \left(\frac{\psi_{xx}}{2} \left(c^{(1)}_{0,1}\right)^2 + \psi_{xy}c^{(1)}_{0,1} c^{(2)}_{0,1}+ \frac{\psi_{yy}}{2} \left(c^{(2)}_{0,1}\right)^2 \right) (y-y_0)^2 + \\
		&~~~~ \left(\frac{\psi_{xx}}{2} c^{(1)}_{0,1} c^{(1)}_{1,0} + \psi_{xy} \left(c^{(1)}_{0,1} c^{(2)}_{1,0} + c^{(1)}_{1,0} c^{(2)}_{0,1} \right)+ \frac{\psi_{yy}}{2} c^{(2)}_{0,1}c^{(2)}_{1,0}  \right)(x-x_0)(y-y_0) + \dots
	\end{align*}
	Comparing both sides term by term, we obtain from the constant term 
	\begin{align}
		\frac{\psi_{xx}}{2} \left(c^{(1)}_{0,0}\right)^2 + \psi_{xy}c^{(1)}_{0,0} c^{(2)}_{0,0}+ \frac{\psi_{yy}}{2} \left(c^{(2)}_{0,0}\right)^2 =0. \label{eq:aa-equate}
	\end{align}
	Since $\Sigma$ is positive definite, the only solution to \eqref{eq:aa-equate} is $c^{(1)}_{0,0}= c^{(2)}_{0,0}=0$. 
	Therefore, the coefficients of the linear term on the left-hand side of \eqref{eq:intermediatecomparison} are all zero. In other words,
	\begin{align*}
		\frac{\psi_{xx}}{2} c^{(1)}_{0,0} c^{(1)}_{1,0} + \psi_{xy} c^{(1)}_{0,0}c^{(2)}_{1,0} + \psi_{xy} c^{(1)}_{1,0} c^{(2)}_{0,0} + \frac{\psi_{yy}}{2} c^{(2)}_{0,0} c^{(2)}_{1,0} = 0, \\
		\frac{\psi_{xx}}{2} c^{(1)}_{0,0} c^{(1)}_{0,1} + \psi_{xy} c^{(1)}_{0,0}c^{(2)}_{0,1} + \psi_{xy} c^{(1)}_{1,0} c^{(2)}_{0,0} + \frac{\psi_{yy}}{2} c^{(2)}_{0,0} c^{(2)}_{0, 1} = 0.
	\end{align*}
	For the quadratic terms, we obtain
	\begin{equation*}
		\begin{aligned}
			&\frac{\psi_{xx}}{2} \left(c^{(1)}_{1,0}\right)^2 + \psi_{xy}c^{(1)}_{1,0} c^{(2)}_{1,0}+ \frac{\psi_{yy}}{2} \left(c^{(2)}_{1,0}\right)^2 =\frac{\psi_{xx}}{2} \\
			&\frac{\psi_{xx}}{2} \left(c^{(1)}_{0,1}\right)^2 + \psi_{xy}c^{(1)}_{0,1} c^{(2)}_{0,1}+ \frac{\psi_{yy}}{2} \left(c^{(2)}_{0,1}\right)^2 =\frac{\psi_{yy}}{2} \\
			& \frac{\psi_{xx}}{2} c^{(1)}_{0,1} c^{(1)}_{1,0} + \psi_{xy} \left(c^{(1)}_{0,1} c^{(2)}_{1,0} + c^{(1)}_{1,0} c^{(2)}_{0,1} \right)+ \frac{\psi_{yy}}{2} c^{(2)}_{0,1}c^{(2)}_{1,0} =\psi_{xy}
		\end{aligned}
	\end{equation*}
	One solution is given by
	\begin{align*}
		c^{(1)}_{0,1}= c^{(2)}_{1,0}=0, \\
		c^{(1)}_{1,0}= c^{(2)}_{0,1}=1.
	\end{align*}
	The rest of the coefficients can be computed in a similar way.
	
	We now define the inversion series of \eqref{eq:z1z2} to be 
	\begin{align*}
		x - x_0 & = \sum_{l_1 + l_2 \geq 1} a_{l_1, l_2} z_1^{l_1} z_2^{l_2} \\
		y - y_0 & = \sum_{l_1 + l_2 \geq 1} b_{l_1, l_2} z_1^{l_1} z_2^{l_2}
	\end{align*}
	Plugging this into \eqref{eq:z1z2}, we get that $a_{l_1, l_2}$'s and $b_{l_1, l_2}$'s must satisfy
	\begin{align*}
		z_1 &= c^{(1)}_{1,0} \sum_{l_1 + l_2 \geq 1} a_{l_1, l_2} z_1^{l_1} z_2^{l_2} + \dots  = c^{(1)}_{1,0} a_{1,0} z_1 +c^{(1)}_{1,0} a_{0, 1} z_2 \dots \\
		z_2 &= c^{(2)}_{0,1} \sum_{l_1 + l_2 \geq 1} b_{l_1, l_2} z_1^{l_1} z_2^{l_2} + \dots  = c^{(2)}_{1,0} b_{1,0} z_1+ c^{(2)}_{0,1} b_{0,1} z_2 + \dots
	\end{align*}
	Therefore, 
	\begin{align*}
		&a_{1,0} = b_{0,1} = 1, \\
		&a_{0,1} = b_{1,0} = 0.
	\end{align*}
	The rest of the coefficients can be computed similarly. 
	
	We further note that 
	\begin{align*}
		dx \wedge dy = \sum_{\substack{l_1 + l_2 \geq 1 \\ m_1 + m_2 \geq 1}} a_{l_1, l_2} b_{m_1, m_2} (l_1m_2 - l_2m_1)z_1^{l_1 + m_1 -1}z_2^{l_2+m_2-1} dz_1 \wedge dz_2. 
	\end{align*}
	Particularly, when $l_1 + m_1 =1$ and $l_2 + m_2 =1$, the coefficient is 
	\begin{align*}
		a_{1,0}b_{0,1} + a_{0, 1}b_{1,0} = 1. 
	\end{align*}
	Then, using this change of variable, we can re-write the integral as
	\begin{align}
		& \int_{\frac{\delta}{n}}^{1 - \frac{1}{2n}} \int_{\frac{\delta}{n}}^{1 - \frac{1}{2n} -y}   \rho(x, y) e^{n\psi(x, y)} \, dx dy \notag \\
		= & e^{n \psi(x_0, y_0)}  \int_{\frac{\delta}{n}}^{1 - \frac{1}{2n}} \int_{\frac{\delta}{n}}^{1 - \frac{1}{2n} -y}  \rho(x,y) e^{n (\psi(x,y) - \psi(x_0,y_0))} \, dx dy \notag \\
		= &  e^{n \psi(x_0, y_0)} \int_{\frac{\delta}{n}}^{1 - \frac{1}{2n}} \int_{\frac{\delta}{n}}^{1 - \frac{1}{2n} -y}  ( \rho(x_0, y_0) + \rho_{x}(x_0, y_0)(x-x_0) + \rho_{y}(x_0, y_0)(y-y_0)  + \notag \\
		\quad & \frac{\rho_{xx}(x_0, y_0)}{2} (x-x_0)^2 
		+ \rho_{xy}(x_0, y_0)(x-x_0)(y-y_0)+ \dots) \notag \\
		~~& e^{n (\psi(x,y) - \psi(x_0,y_0))} dx dy \notag \\
		\sim &  e^{n \psi(x_0, y_0)} \int_{\mathbb{R}^2} \left( \rho(x_0, y_0) + \rho_{x}(x_0, y_0) \left(\sum_{l_1 + l_2 \geq 1} a_{l_1, l_2} z_1^{l_1} z_2^{l_2} \right)  + \dots\right) e^{-\frac{n}{2}z^\intercal \Sigma z } \notag \\
		& \sum_{l_1 + l_2 \geq 1} \sum_{m_1 + m_2 \geq 1} a_{l_1, l_2} b_{m_1, m_2} (l_1m_2 - l_2m_1)z_1^{l_1 + m_1 -1}z_2^{l_2+m_2-1} dz_1 dz_2.  \label{eq:dz1dz2}
	\end{align}
	Now we organize the terms in \eqref{eq:dz1dz2} in ascending order of the exponents of $z_1^{r_1}z_2^{r_2}e^{-\frac{n}{2}z^\intercal \Sigma z }$. Since
	\begin{align*}
		\int_{\mathbb{R}^2}  z_1^{r_1}z_2^{r_2}e^{-\frac{n}{2}z^\intercal \Sigma z } \, dz_1 dz_2 = \left( \frac{1}{n}\right)^{\frac{r_1 + r_2}{2} + 1} \int_{\mathbb{R}^2}  z_1^{r_1}z_2^{r_2}e^{-\frac{1}{2}z^\intercal \Sigma z } \, dz_1 dz_2
	\end{align*}
	and since $\rho(x_0,y_0)$, $\rho_x(x_0,y_0)$, $\rho_y(x_0,y_0) $, $\rho_{xx}(x_0,y_0)$ and $\rho_{yy}(x_0,y_0)$ are all zero, the leading term of \eqref{eq:dz1dz2} becomes
	\begin{align}
		\int_{\mathbb{R}^2}  \rho(x, y) e^{n\psi(x, y)} \, dx dy \, \sim & e^{n \psi(x_0, y_0)}  \frac{\rho_{xy}(x_0,y_0)}{n^2} \int_{\mathbb{R}^2} z_1 z_2 e^{-\frac{1}{2}z^\intercal \Sigma z } \, dz_1dz_2 \notag \\
		= & - \frac{\pi}{2 n^2} \frac{\sqrt{p_Jp_{J'}}}{(1-p_J-p_{J'})^n} . \label{eq:integral_eval2}
	\end{align}
	In \eqref{eq:integral_eval2}, we have used the facts
	\begin{gather*}
		\int_{\mathbb{R}^2} z_1 z_2 e^{-\frac{1}{2}z^\intercal \Sigma z } \, dz_1dz_2 = 2\pi \sqrt{\det(\Sigma^{-1})}\Haf(\Sigma^{-1}) = -2\pi (\sqrt{p_Jp_{J'}})^3 \sqrt{1- p_J -p_{J'}} \\ 
		\rho_{xy}(x_0,y_0) = \frac{1}{4p_Jp_{J'} \sqrt{1- p_J -p_{J'}}} \\
		e^{n \psi(x_0, y_0)} = \frac{1}{(1-p_J-p_{J'})^n}.
	\end{gather*}
	Finally, we obtain from \eqref{eq:inteshortcuta} and \eqref{eq:integral_eval2} that 
	\begin{align}
		\int_{L^-_n} \phi(x, y) \, dx dy  \sim  -\frac{1}{4} \frac{1}{\sqrt{2\pi n}} \frac{1}{n^2} \sqrt{p_Jp_{J'}} \frac{e^n}{(1-p_J-p_{J'})^n} \frac{1}{n^n}. \label{eq:inttriangle}
	\end{align}
	
	The integral of $\phi(x,y)$ over $U_n$ and $L_n^+$ are  negligible compared to the right hand side of \eqref{eq:inttriangle}. These results are proved in Lemmas \ref{lem:ejjprime_un} and \ref{lem:ejjprime_lnplus}. 
	In Lemma \ref{lem:ejjprime_wn1}, we prove that the contribution from $W_{n,1}$ can also be ignored. 
	Similarly, the contribution from $W_{n,2}$ is negligible, as the computation is identical to that of $W_{n,1}$ after swapping $x$ and $y$. The integral over $W_{n,3}$ is analyzed in Lemma \ref{lem:ejjprime_wn3}, and it does not contribute to the leading order asymptotics. Similarly, by swapping $x$ and $y$, it follows immediately from Lemma \ref{lem:ejjprime_wn3} that the contribution from $W_{n,4}$ is negligible.
	In total, we have 
	\begin{align*}
		\int_{\hat{\gamma}_{n, 1} \times \hat{\gamma}_{n, 1}} \phi(x, y) \, dx dy \sim  -\frac{1}{4} \frac{1}{\sqrt{2\pi n}} \frac{1}{n^2} \sqrt{p_Jp_{J'}} \frac{e^n}{(1-p_J-p_{J'})^n} \frac{1}{n^n},
	\end{align*}
	and thus
	\begin{align*}
		\int_{\hat{\gamma}_{n, 1} \times \hat{\gamma}_{n, 1}} h_n(x, y) \, dx dy \sim  \frac{1}{4} \frac{1}{\sqrt{2\pi n}} \frac{1}{n^2} \sqrt{p_Jp_{J'}} \frac{e^n}{(1-p_J-p_{J'})^n} \frac{1}{n^n}.
	\end{align*}
	By an analogous argument, we obtain the same leading order asymptotics for the integrals over $\hat{\gamma}_{n, 1} \times \hat{\gamma}_{n, 3}$, $\hat{\gamma}_{n, 3} \times \hat{\gamma}_{n, 1}$ and $\hat{\gamma}_{n, 3} \times \hat{\gamma}_{n, 3}$. This follows from appropriately adjusting the signs of the asymptotics of $\cot(n\pi x)$ and $\cot(n \pi y)$ depending on whether $x$ and $y$ lie above or below the real axis, as well as the orientation of the surface. 
	The remaining $\hat{\gamma}_{n, i} \times \hat{\gamma}_{n, j}$'s do not contribute to the leading order asymptotics. We provide proofs for a few representative cases to illustrate the key techniques. See Lemmas \ref{lem:ejjprime_rn12}, \ref{lem:ejjprime_rn14}, and \ref{lem:ejjprime_rn22}. The proof for the rest of are completely similar to the cases treated in these lemmas. 
	
	In the end, we see that
	\begin{align}
		I_n \sim  \frac{1}{\sqrt{2\pi n}} \sqrt{p_Jp_{J'}} \frac{e^n}{(1-p_J-p_{J'})^n} \frac{1}{n^n}. \label{eq:ejjprimeIn1}
	\end{align}
	Inserting \eqref{eq:ejjprimeIn1} into \eqref{eq:ejjprimelast}, we get the leading order asymptotics is 
	\begin{align*}
		e_{J, J'} & \sim \frac{1}{(2i)^2} \frac{\sqrt{p_Jp_{J'}}}{n} \Gamma(n+1) \frac{1}{\sqrt{2\pi n}} \frac{e^n}{n^n} \\
		& \sim -\frac{\sqrt{p_Jp_{J'}}}{4n}.
	\end{align*}
\end{proof}

\subsection{The importance sampling method: GBS-I}
The GBS-I algorithm solves the special cases $\mathcal{I}^{\times}_{\Haf^2}(\epsilon, \delta)$ using the idea of importance sampling.
Let $I_1, \dots, I_n$ be $n$ independent draws from $P_{tB}$ where $m_{tB} = 2K$. We define the GBS-I estimator to be
\begin{align*}
	\mathcal{E}^\textnormal{GBS-I}_n = \frac{1}{n} \sum_{i = 1}^n \frac{ I_i!}{d_t} a_{I_i} t^{-2K}.
\end{align*}
It is straightforward to show that the GBS-I estimator is unbiased. In other words, for all $n$
\begin{align*}
	\mathbb{E}[\mathcal{E}_n^\textnormal{GBS-I}] = \mu_{\Haf^2}
\end{align*}
By WLLN, we get
\begin{equation}
	\lim_{n \rightarrow \infty} P( \vert \mathbb{E}[\mathcal{E}^\textnormal{GBS-I}_n] - \mu^k_{\Haf^2} \vert > \epsilon) = 0.
	\label{eq:asymbehavior} 
\end{equation}
Thus, GBS-I solves $\mathcal{I}^\times_{\Haf^2}(\epsilon, \delta)$ for large enough $n$.
Furthermore, from Chevyshev's inequality, a \textit{guaranteed sample size} $n_\text{GBS-I}$ can be computed explicitly. 

\begin{theorem}
	\label{thrm:GBS-Icov}
	The GBS-I estimator $\mathcal{E}_n^\textnormal{GBS-I}$ solves $\mathcal{I}^\times_{\Haf^2}(\epsilon, \delta)$
	with
	\begin{align*}
		n \geq \frac{	V^{\textnormal{GBS-I}}_{\Haf^2} }{\delta \epsilon^2 \vert \mu_{\Haf^2}  \vert^2} := n^{\textnormal{GBS-I}}_{\Haf^2}
	\end{align*}
	where
	\begin{align*}
		V^{\textnormal{GBS-I}}_{\Haf^2} = \left(\sum_{\vert J \vert = 2K} a^2_J t^{-4K}J! \Haf(B_I)^2  - \vert \mu_{\Haf^2} \vert^2  \right).
	\end{align*}
\end{theorem}

\begin{proof}
	The key ingredient of this proof is to compute the variance of the GBS-I estimator.
	\begin{align*}
		\textnormal{Var}(\mathcal{E}^{\textnormal{GBS-I}}_n) &= \mathbb{E}[(\mathcal{E}^{\textnormal{GBS-I}}_n - \vert \mu_{\Haf^2} \vert^2 )] \\
		& =  \frac{1}{n} \left(\frac{1}{d_t} \sum_{\vert J \vert = 2K} a^2_J t^{-4K}J! \Haf(B_J)^2  - \vert \mu_{\Haf^2} \vert^2  \right)\\
		& = \frac{1}{n} V^{\textnormal{GBS-I}}_{\Haf^2}.
	\end{align*}
	The rest of the proof is straightforward from Chebyshev's inequality. 
\end{proof}

Similar to Corollary \ref{coro:optimalt}, we can also show the optimal $t_0$ for the GBS-I estimator is again when $m_{t_0B} = 2K$. 

\begin{corollary}
	A real value $t_0$ is a solution to
	\begin{align*}
		t_0 = \argmin_{t\in (0, \lambda_1^{-1})} \mathbb{E} \left[ \vert \mathcal{E}^{\textnormal{GBS-I}}_n  - \mu_{\Haf} \vert^2 \right]
	\end{align*}
	if and only if $m_{t_0 B} = 2K$.  
\end{corollary}

\begin{proof}
	The proof is omitted due to similarity to Corollary \ref{coro:optimalt}.
\end{proof}

\subsection{Comparison with the plain Monte Carlo}
We use a plain Monte Carlo (MC) estimator as a baseline for comparison. The plain MC estimator solves the approximation problems by taking the average of $f$ evaluated at $n$ random samples from the multivariate Gaussian distribution $g$. 
For $\mathcal{I}^{\times}_{\Haf, K}(\epsilon, \delta)$, the functions $f$ and $g$ are defined as in \eqref{eq:I}. For $\mathcal{I}^{\times}_{\Haf^2, K}(\epsilon, \delta)$, the functions  $f$ and $g$ are defined as in \eqref{eq:I2}. 
Let 
\begin{align*}
	\mathcal{E}^{\textnormal{MC}}_n = \sum_{i=1}^n \frac{1}{n} f(X_i)
\end{align*}
with $X_1, X_2, \dots, X_n$ being i.i.d samples with probability density $g$. Again, by WLLN, the plain MC estimator solves $\mathcal{I}^\times_{\Haf}(\epsilon, \delta)$ or  $\mathcal{I}^\times_{\Haf^2}(\epsilon, \delta)$ for large $n$. Similarly, we can derive the guaranteed sample size using Chebyshev's inequality as Theorem \ref{thrm:GBS-Icov}. 
For $\mathcal{I}^\times_{\Haf}(\epsilon, \delta)$, we have
\begin{align*}
	n_{\Haf}^{\textnormal{MC}} = \frac{V^{\textnormal{MC}}_{\Haf} }{\delta \epsilon^2 \vert \mu_{\Haf} \vert^2}, ~~~~~ V^{\textnormal{MC}}_{\Haf} = \sum_{\substack{|J|, |J'| = 2K}}  a_J a_{J'} \Haf(B_{J+J'}) - \vert \mu_{\Haf} \vert^2.
\end{align*}
For $\mathcal{I}^\times_{\Haf^2}(\epsilon, \delta)$, we have
\begin{align*}
	n_{\Haf^2}^{\textnormal{MC}} = \frac{V^{\textnormal{MC}}_{\Haf^2} }{\delta \epsilon^2 \vert \mu_{\Haf^2} \vert^2}, ~~~~~ V^{\textnormal{MC}}_{\Haf^2} = \sum_{\substack{|J|, |J'| = 2K}} a_J a_{J'} \Haf(B_{J + J'})^2 - \vert \mu_{\Haf^2} \vert^2.
\end{align*}
To compare the efficiency of our GBS algorithms and MC estimator, we look at their corresponding guaranteed sample size. 
For the rest of the paper, we describe how to estimate the percentage of the Gaussian expectation problems where $n_{\Haf}^{\textnormal{GBS-P}} \leq n_{\Haf}^{\textnormal{MC}}$ for $\mathcal{I}^\times_{\Haf}(\epsilon, \delta)$ and $n_{\Haf^2}^{\textnormal{GBS-I}} \leq n_{\Haf^2}^{\textnormal{MC}}$ for $\mathcal{I}^\times_{\Haf^2}(\epsilon, \delta)$.

	\section{Percentage of GBS advantage}
\label{sec:percentage}
\subsection{Theoretical formulation}
\label{subsec:space}
We start by comparing GBS-P and MC. For an arbitrary $N$ and $K$, let us denote $\sigma(N,K) = \# \left\{\vert I \vert = 2K \vert \right\} $ to be the size of the coefficient set. Let $\mathbb{S}^{\sigma(N,K) -1}$ be the unit sphere such that
\begin{align*}
	\mathbb{S}^{\sigma(N,K) -1} = \left \{ (a_I)_{\vert I \vert = 2K} : \sum_{\vert I \vert = 2K} a^2_I = 1 \right \}.
\end{align*} 
Let $M(N)$ be the space of real, symmetric $N \times N$ matrices whose eigenvalues are bounded strictly between 0 and 1,
\begin{align*}
	M(N) = \left \{ M \in \textnormal{Sym}(N) : ~\text{spec}(M) \in (0, 1) \right \}.
\end{align*}
The problem space $\mathcal{P}_{\Haf}(N, K)$ is thus
\begin{align*}
	\mathcal{P}_{\Haf}(N, K) = \mathbb{S}^{\sigma(N,K) -1} \times M(N).
\end{align*}

To assign a volume form on $\mathcal{P}_{\Haf}$, we first choose a normalized rotational invariant measure $dS$ on $\mathbb{S}^{\sigma(N,K) -1}$. Recall that the Lebesgue measure on $\text{Sym}(N)$ is given on the matrix entries as
\begin{align*}
	dM = \prod_{i = 1}^N dM_{ii} \prod_{i = 1}^{N-1}  \prod_{j = i+1}^{N}  d M_{ij}.
\end{align*}
A real symmetric matrix admits a diagonalization $M = U^{-1} \Lambda U$, where $U \in SO(N)$ is an orthonormal matrix and $\Lambda = \text{diag}(\lambda_1, \dots, \lambda_N)$ consists of all the eigenvalues of $M$. 
The Lebesgue measure is invariant under conjugation by orthogonal matrices. Hence, via a change of variable, we obtain
\begin{align}
	dM = \vert \Delta(\lambda_1, \dots, \lambda_N)\vert \prod_{i}^N d\lambda_i dU.
	\label{eq:lebesgue}
\end{align}
Here, $\Delta(\lambda_1, \dots, \lambda_N)$ is the Vandermonde determinant such that $$\Delta(\lambda_1, \dots, \lambda_N) = \prod_{1 \leq i < j \leq N} (\lambda_i - \lambda_j)$$
and $dU$ is the normalized Haar measure on $SO(N)$. We scale $dM$ so as to make the measure integrate to 1. In other words, 
\begin{align*}
	 \int_{M(N)} dM = \int_{SO(N)} \int_{[0,1]^N} c_N \vert \Delta(\lambda_1, \dots, \lambda_N)\vert \prod_{i}^N d\lambda_i dU  = 1,
\end{align*}
which yields
\begin{align*}
	 c_N = \left( \int_{[0,1]^N} \vert \Delta(\lambda_1, \dots, \lambda_N)\vert \prod_{i}^N d\lambda_i \right)^{-1}.
\end{align*}
The volume form on $\mathcal{P}_{\Haf}$ is thus the product $dS \times dM$, which we denote for convenience as $d\mu$. 

Now we compute the percentage of GBS-P outperforms MC over the problem space $\mathcal{P}_{\Haf}$ as the following integral.
\begin{align*}
	P^{\textnormal{GBS-P}}_{\Haf} & = \int_{\mathcal{P}_{\Haf}} H(n^{\textnormal{MC}}_{\Haf} - n^{\textnormal{GBS-P}}_{\Haf}) \, d\mu
\end{align*}
Here, $H(x)$ denote the heavyside step function such that $H(x) = 1$ if $x \geq 0$ and 0 otherwise. We think of $n^{\textnormal{MC}}_{\Haf}$ and $n^{\textnormal{GBS-P}}_{\Haf}$ as functions of $\eta$ where $\eta \in \mathcal{P}_{\Haf}$. 
Since
\begin{align*}
	H(n^{\textnormal{MC}}_{\Haf}(\eta) - n^{\textnormal{GBS-P}}_{\Haf}(\eta)) = H(V^{\textnormal{MC}}_{\Haf}(\eta) - V^{\textnormal{GBS-P}}_{\Haf}(\eta)),
\end{align*}
$P^{\textnormal{GBS-P}}_{\Haf}$ is 
independent of the choice of $\epsilon$ and $\delta$. 

Similarly, for $\mathcal{I}^\times_{\Haf^2}(\epsilon, \delta)$, we find that the problem space $\mathcal{P}_{\Haf^2}$ is the same as $\mathcal{P}_{\Haf}$. Therefore, we compute the proportion of $\mathcal{P}_{\Haf^2}$ where GBS-I outperforms MC to be
\begin{align*}
	P^{\textnormal{GBS-I}}_{\Haf^2}
	& = \int_{\mathcal{P}_{\Haf^2}} H(n^{\textnormal{MC}}_{\Haf^2} - n^{\textnormal{GBS-I}}_{\Haf^2}) \, d\mu,
\end{align*}
and $P^{\textnormal{GBS-I}}_{\Haf^2}$ is also independent of $\epsilon$ and $\delta$.

\subsection{Numerical methods}
\label{subsec:emp}
We describe a simple numerical method for approximating $P^{\textnormal{GBS-P}}_{\Haf}$ and $P^{\textnormal{GBS-I}}_{\Haf^2}$. 
We begin with describing the method for $P^{\textnormal{GBS-P}}_{\Haf}$, noting that it can be easily adapted for $P^{\textnormal{GBS-I}}_{\Haf^2}$. 
Let $\eta \in \mathcal{P}_{\Haf}$  be expressed by 
\begin{align*}
	\eta = (a, B), ~~~~ a = (a_I)_{\vert I \vert = 2K} \in \mathbb{S}^{\sigma(N,K) -1},  ~~~~ B \in M(N)
\end{align*}
Let $ B= U^{-1} \Lambda U$ be the eigendecomposition of $B$, where $U \in SO(N)$ is an orthonormal matrix and $\Lambda = \text{diag}(\lambda_1, \dots, \lambda_N)$ is the diagonal matrix of all the eigenvalues of $B$. Then, we obtain
\begin{align}
	P^{\textnormal{GBS-P}}_{\Haf}
	& = c_N \int_{SO(N)} \int_{[0,1]^N} \int_{\mathbb{S}^{\sigma(N,K) -1}} 	\vert \Delta(\lambda_1, \dots, \lambda_N)\vert  \notag \\
	& \hspace{0.3in} H(V^{\textnormal{MC}}_{\Haf}(a, U^{-1}\Lambda U) - V^{\textnormal{GBS-P}}_{\Haf}(a, U^{-1}\Lambda U))
	\, dS \prod_{i}^N d\lambda_i dU. \label{eq:int_GBSP_final}
\end{align}

To compute this, we use the following MC method\footnote{It is a challenge to compute approximation of $P^{\textnormal{GBS-P}}_{\Haf}$ and $P^{\textnormal{GBS-I}}_{\Haf^2}$ for increasing problem sizes. We will explore the use of large-scale HPC and GBS to compute these terms for higher $(K,N)$ in future work.}. We first draw i.i.d samples $U^{(1)}, U^{(2)}, \dots, U^{(n_1)}$ with respect to the Haar measure on $SO(N)$. Let $\Lambda^{(l)} = \textnormal{diag}(\lambda^{(l)}_1, \dots, \lambda^{(l)}_N)$ be a diagonal matrix where each $\lambda^{(l)}_i$ is sampled i.i.d from the uniform distribution on $(0, 1)$. Then, construct $B^{(l)}  = \left(U^{(l)}\right)^{-1} \Lambda U^{(l)}$ and evaluate the vandemonde determinant $\vert \Delta(\lambda^{(l)}_1, \dots, \lambda^{(l)}_N)\vert$. Further, we use grid search to find the parameter $t^{(l)}$ such that $m_{ t^{(l)} B^{(l)} } = 2K$. Next, we compute $\Haf(B_J)$ and $\Haf(B_{J + J'})$ for all $J, J'$ such that $\vert J \vert = \vert J' \vert = 2K$. Here, the matrix hafnian is computed using the algorithm described as in \cite{bulmer2022boundary, gupt2019walrus}.
Then, for each $ B^{(l)}$, we draw $n_2$ i.i.d samples $a^{(l, 1)}, a^{(l,2)}, \dots, a^{(l, n_2)}$ from the uniform distribution on $\mathbb{S}^{\sigma(N,K) -1}$. Finally, we evaluate $V^{\textnormal{MC}}_{\Haf}(a^{(l,m)}, B^{(l)})$ and $V^{\textnormal{GBS-P}}_{\Haf}(a^{(l,m)}, B^{(l)})$.
The estimator for $P^{\textnormal{GBS-P}}_{\Haf}$ is thus defined as 
\begin{align}
	\mathcal{E}_{n_1, n_2} & = \frac{1}{n_1n_2}\sum_{l = 1}^{n_1} \sum_{m=1}^{n_2} \vert \Delta(\lambda^{(l)}_1, \dots, \lambda^{(l)}_N)\vert \notag  \\
	&\quad \quad H\left(  V^{\textnormal{MC}}_{\Haf}(a^{(l,m)}, B^{(l)}) - V^{\textnormal{GBS-P}}_{\Haf}(a^{(l,m)}, B^{(l)}) \right). \label{eq:volume_est}
\end{align}
It is clear that $\mathcal{E}_{n_1, n_2}$ is an ubiased estimator of  $P^{\textnormal{GBS-P}}_{\Haf}$. By WLLN,  $\mathcal{E}_{n_1, n_2}$ yields a good approximation for large $n_1$ and $n_2$. 

\begin{algorithm}[ht]
	\caption{Numerical Approximation to $P^{\textnormal{GBS-P}}_{\Haf}$}
	\begin{algorithmic}[1]
		\STATE Initialize $E \leftarrow 0$
		\FOR{each $l = 1, 2, \dots, n_1$}
		\STATE Sample $U^{(l)}$ i.i.d. from the Haar measure on $SO(N)$
		\STATE Sample $\lambda_1^{(l)}, \dots, \lambda_N^{(l)}$ i.i.d. from the uniform distribution on $(0, 1)$
		\STATE Construct $\Lambda^{(l)} \leftarrow \text{diag}(\lambda_1^{(l)}, \dots, \lambda_N^{(l)})$
		\STATE Compute $v \leftarrow \vert \Delta(\lambda_1^{(l)}, \dots, \lambda_N^{(l)}) \vert$
		\STATE Compute $B^{(l)} \leftarrow \left(U^{(l)}\right)^{-1} \Lambda^{(l)} U^{(l)}$
		\STATE Perform grid search to find $t^{(l)}$ such that $m_{t^{(l)} B^{(l)}} = 2K$
		\STATE Compute $\Haf(B_J)$ and $\Haf(B_{J+J'})$ for all $J, J'$ where $\vert J \vert = \vert J' \vert = 2K$
		\STATE Draw i.i.d. samples $a^{(l, 1)}, a^{(l, 2)}, \dots, a^{(l, n_2)}$ uniformly from $\mathbb{S}^{\sigma(N,K) - 1}$
		\FOR{each $m = 1, 2, \dots, n_2$}
		\STATE Compute $\tau \leftarrow H \left(V^{\text{MC}}_{\Haf}(a^{(l,m)}, B^{(l)}) - V^{\text{GBS-P}}_{\Haf}(a^{(l,m)}, B^{(l)}) \right)$
		\STATE Update $E \leftarrow E + v \tau$
		\ENDFOR
		\ENDFOR
		\STATE Return $\frac{E}{n_1 n_2}$
	\end{algorithmic}
\end{algorithm}

Figure \ref{fig:heatmap_GBSP} shows numerical approximations of $P^{\textnormal{GBS-P}}_{\Haf}$ for various $N$ and $K$, computed with $n_1 = 30$ and $n_2 = 100$. Convergence plots in Appendix C confirm that our numerical approximations stabilize for the chosen $n_1$ and $n_2$ and $P^{\textnormal{GBS-P}}_{\Haf}$ indeed has converged.

For $P^{\textnormal{GBS-I}}_{\Haf^2}$, we use an almost identical approach. Notice that
\begin{align}
	P^{\textnormal{GBS-I}}_{\Haf^2}
	& = c_N \int_{SO(N)} \int_{[0,1]^N} \int_{\mathbb{S}^{\sigma(N,K) -1}} 	\vert \Delta(\lambda_1, \dots, \lambda_N)\vert  \notag \\
	& \hspace{0.3in} H(V^{\textnormal{MC}}_{\Haf^2}(a, U^{-1}\Lambda U) - V^{\textnormal{GBS-I}}_{\Haf^2}(a, U^{-1}\Lambda U))
	\, dS \prod_{i}^N d\lambda_i dU. \label{eq:int_GBSI_final}
\end{align}
Then, the estimator is defined similarly to \eqref{eq:volume_est}, replacing the term inside the step function with $V^{\textnormal{MC}}_{\Haf^2}$ and $V^{\textnormal{GBS-I}}_{\Haf^2}$.

	\newpage
	\appendix
	\section{Code accessibility}
Python scripts used to generate the numerical results are available from the GitHub repository: \url{https://github.com/sshanshans/GBSPE}.

\section{Technical lemmas}
\begin{lemma}
	\label{lem:kslice}
	Suppose $e_k$ solves $\mathcal{I}^{\times}_{\Haf, k}(\frac{\epsilon}{K}\frac{\vert \mu_{\Haf}\vert}{\vert \mu_{\Haf, k}\vert}, 1 - \frac{1-\delta}{K})$ where $k = 1, \dots, K$. Then, $e = \sum_{k=0}^K e_k$ solves $\mathcal{I}^{\times}_{\Haf}(\epsilon, \delta)$.
\end{lemma}

\begin{proof}
	Let $A_k$ denote the event $\vert \mu_{\Haf, k} - e_k \vert < \frac{\epsilon}{K}{\vert \mu_{\Haf}\vert}$ and $B_k = \overline{A_k}$. Let $\tau = \frac{1-\delta}{K}$.
	Then, $e_k$ solves $\mathcal{I}^{\times}_{\Haf, k}(\frac{\epsilon}{K}\frac{\vert \mu_{\Haf}\vert}{\vert \mu_{\Haf, k}\vert}, 1 - \frac{1-\delta}{K})$ implies that $P(A_k) > 1-\tau$ and $P(B_k) \leq \tau$. Let $Q$ denote the event $\vert \mu_{\Haf} - e \vert < \epsilon \vert \mu_{\Haf}\vert$. Note that $A_1 \cap A_2 \dots \cap A_K$ implies $Q$. To see this, we note
	\begin{align*}
		\vert \mu_{\Haf} - e \vert = \lvert \sum_{k =0}^K \mu_{\Haf, k} - e_k \rvert \leq \sum_{k = 0}^K \lvert \mu_{\Haf, k} - e_k \rvert \leq \epsilon \vert \mu_{\Haf}\vert,
	\end{align*}
	which implies $Q$ is true. Therefore, $P(A_1 \cap A_2 \cap \dots A_K) < P(Q)$. 
	
	We further compute from  De Morgan's Laws and the union bound that
	\begin{align*}
		1 - P(A_1 \cap A_2 \cap \dots \cap A_K) & = P(B_1 \cup B_2 \dots \cup B_K) \leq \sum_{k=1}^K P(B_k) = K\tau. 
	\end{align*}
	Therefore,
	$$ P(A_1 \cap A_2 \cap \dots A_K)  \geq 1 - K \tau > \delta, $$
	which implies that $P(Q) > \delta$ and $e $ solves $\mathcal{I}^{+}_{\Haf}(\epsilon, \delta)$.
\end{proof}

\begin{lemma}
	\label{lem:meanphoton}
	If $\lambda_1, \dots, \lambda_N$ are eigenvalues of $B$, then
	\begin{align*}
		m_B =\sum_{n = 1}^N \frac{\lambda_n^2}{1- \lambda_n^2}.
	\end{align*}
\end{lemma}

\begin{proof}
	It follows from the proof of Theorem A.1 in \cite{shan2024companion} that
	the covariance matrix $\Sigma$ of $\hat{\rho}_B$ must be of the form 
	$$\Sigma = \begin{bmatrix}
		\Sigma_1 & \Sigma_2 \\
		\Sigma_2 & \Sigma_1
	\end{bmatrix}, $$
	where
	\begin{align*}
		\Sigma_1 &= U \left( \bigoplus_{n=1}^N \cosh^2(r_n) - \frac{1}{2} \right) U^\intercal \\
		\Sigma_2 &= U \left( \bigoplus_{n=1}^N \cosh(r_n)\sinh(r_n) - \frac{1}{2} \right) U^\intercal.
	\end{align*}
	In these expressions $U$ and $r_n$'s are from the eigendecomposition of $B$:
	\begin{align*}
		B = U \Lambda U^\intercal, ~~~~ \Lambda = \text{diag}(\lambda_1, \dots, \lambda_N), ~~~~ U U^\intercal = \mathbb{I}_N
	\end{align*}
	and $$r_n = \tanh^{-1}(\lambda_n).$$
	From Eq. (27) of \cite{vallone2019means}, the mean photon number $m_i$ in node $i$ of  $\hat{\rho}$ is given by
	\begin{align*}
		m_i = U_{ii} - \frac{1}{2}.
	\end{align*}
	Thus,  
	\begin{align*}
		m_B & = \sum_{i = 1}^N m_i = \text{Tr}(\Sigma_1) - \frac{N}{2} \\
		& = \sum_{n = 1}^N \cosh^2(r_n) - N \\
		& = \sum_{n = 1}^N \frac{1}{1- \tanh^2(r_n)} - N \\
		& = \sum_{n = 1}^N \frac{1}{1- \lambda_n^2} - N \\
		& = \sum_{n = 1}^N \frac{\lambda_n^2}{1- \lambda_n^2}
	\end{align*}.
\end{proof}

Let $p_J = P_B(J) = \frac{d}{J!} \Haf(B_J)^2$, and let 
\begin{align*}
	p_K = \sum_{\vert J \vert = 2K} p_J.
\end{align*}
In the next lemma, we give an explicit expression of $p_K$ and show that $p_K$ is bounded above by $\frac{1}{\sqrt{2\pi}}$. 
\begin{lemma}
	\label{lem:pJ}
	If $\lambda_1, \dots, \lambda_N$ are eigenvalues of $B$ and $m_B = 2K$, then
	\begin{align}
		p_K = \left(\prod_{l=1}^N {\sqrt{1 - \lambda_l^2}} \right) \sum_{k_1 + k_2 \dots + k_N = K } \prod_{l = 1}^N \left( \frac{ (2k_l)! }{4^{k_l}  (k_l!)^2 } (\lambda_l^2)^{k_l} \right) \label{eq:pK}
	\end{align} 
	and
	\begin{align}
		p_K < \frac{1}{\sqrt{2\pi}}. \label{eq:pK2}
	\end{align}
\end{lemma}

\begin{proof}
	We first prove \eqref{eq:pK}.
	From \eqref{eq:ddef}, we see that
	\begin{align*}
		\frac{1}{d} = \prod_{l=1}^N \frac{1}{\sqrt{1 - \lambda_l^2}}.
	\end{align*}
	Since
	\begin{align*}
		\frac{1}{\sqrt{1- x^2}} & = \sum_{k=0}^\infty \frac{(2k-1)!!}{2^k k!} (x^2)^k \\
				& = \sum_{k = 0}^\infty \frac{(2k)!}{2^{2k} (k!)^2} (x^2)^k, 
	\end{align*}
	we have that
	\begin{align}
		\prod_{l=1}^N \frac{1}{\sqrt{1 - \lambda_l^2}} & = \prod_{l = 1}^\infty \left( \sum_{k_l = 0}^\infty \frac{ (2k_l)! }{4^{k_l} (k_l!)^2 } (\lambda_l^2)^{k_l} \right) \notag \\
		& = \sum_{k = 0}^\infty \sum_{k_1 + k_2 \dots k_N = k } \prod_{l = 1}^N \left( \frac{ (2k_l)! }{4^{k_l}  (k_l!)^2 } (\lambda_l^2)^{k_l} \right).	 \label{eq:pp1}
	\end{align}
	Since $d$ is also the normalization factor to make the sum of the probabilities equal 1, we see that
	\begin{align}
		\frac{1}{d} = \sum_{k=0}^\infty \sum_{\vert I \vert = 2k } \frac{\Haf(B_I)^2}{I!}. \label{eq:pp2} 
	\end{align}
	Equating \eqref{eq:pp1} and \eqref{eq:pp2} yields
	\begin{align*}
		\sum_{\vert I \vert = 2k } \frac{\Haf(B_I)^2}{I!} = \sum_{k_1 + k_2 \dots + k_N = K } \prod_{l = 1}^N \left( \frac{ (2k_l)! }{4^{k_l}  (k_l!)^2  } (\lambda_l^2)^{k_l} \right).
	\end{align*}
	Therefore,
	\begin{align*}
		p_K = d  \sum_{\vert I \vert = 2K} \frac{\Haf(B_I)^2}{I!} = d  \sum_{k_1 + k_2 \dots + k_N = K } \prod_{l = 1}^N \left( \frac{ (2k_l)! }{4^{k_l}  (k_l!)^2 } (\lambda_l^2)^{k_l} \right).
	\end{align*}
	
	We now prove \eqref{eq:pK2}. Let $\mu$ be such that
	\begin{align*}
		2K = \frac{\mu}{1-\mu}.
	\end{align*}
	Then,
	\begin{align*}
		\mu = \frac{2K}{1+ 2K}.
	\end{align*}
	Let 
	\begin{align*}
		S = \left\{ (\lambda_1, \dots, \lambda_N) \in [0, \mu]^N \Big| \,  \sum_{l = 1} \frac{\lambda^2_l}{1-\lambda^2_l} = 2K \right\}. 
	\end{align*}
	In Lemma \ref{lem:boundaryterm}, we see that $\partial S \subseteq \partial\left( [0, \mu]^N  \right)$. 
	Therefore,
	\begin{align*}
		\partial S = S \cap \partial\left([0, \mu]^N \right) = \bigcup_{l = 1}^N \left(S_l^0 \cup S_l^{\mu} \right),
	\end{align*}
	where
	\begin{align*}
		S^{0}_l &= \{ (x_1, \dots, x_N) \in S \mid x_l = 0 \}, \\
		S^{\mu}_l &= \{ (x_1, \dots, x_N) \in S \mid x_l = \mu \}.
	\end{align*}
	Define
	\begin{align*}
		D_{ij} = \lambda_j (1-\lambda_i^2) \frac{\partial}{\partial \lambda_i} - \lambda_i (1-\lambda_j^2) \frac{\partial}{\partial \lambda_j},
	\end{align*}
	which is a vector field along $S$. 
	Now consider
	\begin{align*}
		f(\lambda) = f(\lambda_1, \dots, \lambda_N) = \prod_{l = 1}^N \left(1 - \lambda_l^2\right) 
	\end{align*}
	with $\lambda_l \in [0, \mu]$. 
	Then,
	\begin{align*}
		D_{ij}f(\lambda) & = -2 \lambda_i \lambda_j(1-\lambda_i^2) + 2 \lambda_j \lambda_i(1-\lambda_j^2) \\
		& = 2 \lambda_i \lambda_j(\lambda_i^2 - \lambda_j^2). 
	\end{align*}
	So, $D_{ij}f(\lambda) = 0$ for all $i, j \in \{ 1, \dots, N \}$ is equivalent to 
	\begin{align}
		\lambda_1 = \lambda_2 = \dots = \lambda_N = \lambda. \label{eq:station}
	\end{align}
	Hence, \eqref{eq:station} is a unique stationary point of $f\mid_S$. Since
	\begin{align*}
		2K = N\frac{\lambda^2}{1-\lambda^2},
	\end{align*}
	we thus have
	\begin{align*}
		\lambda = \sqrt{\frac{2K}{N + 2K}}.
	\end{align*}
	Let us also consider
	\begin{align*}
		h(\lambda) = h(\lambda_1, \dots, \lambda_N) =  \sum_{k_1 + k_2 \dots + k_N = K } \prod_{l = 1}^N \left( \frac{ (2k_l)! }{4^{k_l}  (k_l!)^2 } \lambda_l^{k_l} \right).
	\end{align*}
	Then, 
	\begin{align*}
		D_{ij}h(\lambda) = \sum_{k_1 + k_2 \dots + k_N = K } \left( \frac{\lambda_j}{\lambda_i} \left(1 - \lambda_i^2 \right) 2k_i - \frac{\lambda_i}{\lambda_j} \left(1 - \lambda_j^2 \right) 2k_j \right) \prod_{l = 1}^N \left( \frac{ (2k_l)! }{4^{k_l}  (k_l!)^2 } \lambda_l^{k_l} \right).
	\end{align*}
	Notice that if $k = (k_1, \dots, k_N)$ with $k_i = a$ and $k_j = b$ satisfies $k_1 + k_2 \dots + k_N = K$, then $	\tilde{k} = (\tilde{k}_1, \dots, \tilde{k}_N)$
	with $\tilde{k}_i = b$, $\tilde{k}_j = a$ 
	and for all other $l$, $\tilde{k}_l = k_l$, 
	must also satisfy $\tilde{k}_1 + \dots + \tilde{k}_N = K$. Since
	\begin{align*}
	\prod_{l = 1}^N \left( \frac{ (2\tilde{k}_l)! }{4^{\tilde{k}_l}  (\tilde{k}_l!)^2 } \lambda_l^{\tilde{k}_l} \right) = \prod_{l = 1}^N \left( \frac{ (2k_l)! }{4^{k_l}  (k_l!)^2 } \lambda_l^{k_l} \right),
	\end{align*}
	and since
	\begin{align*}
		& \frac{\lambda_j}{\lambda_i} \left(1 - \lambda_i^2 \right) 2a - \frac{\lambda_i}{\lambda_j} \left(1 - \lambda_j^2 \right) 2b + \frac{\lambda_j}{\lambda_i} \left(1 - \lambda_i^2 \right) 2b - \frac{\lambda_i}{\lambda_j} \left(1 - \lambda_j^2 \right) 2a \\
		= & (2a + 2b) \frac{\lambda_j^2\left(1-\lambda_i^2\right) - \lambda_i^2\left(1-\lambda_j^2\right)}{\lambda_i \lambda_j} \\
		=& (2a + 2b) \frac{\lambda_j^2 - \lambda_i^2}{\lambda_i \lambda_j},
	\end{align*}
	pairing $k$ and $\tilde{k}$ in the sum over $k_1 + k_2 \dots + k_N = K$ yields
	\begin{align*}
		D_{ij}h(\lambda) & = \sum_{\substack{0 \leq a < b \\
		a + b < K}} \sum_{\substack{k_1 + k_2 \dots + k_N = K \\ k_i = a, k_j = b}} \left( (2a + 2b) \frac{\lambda_j^2 - \lambda_i^2}{\lambda_i \lambda_j} \right) \prod_{l = 1}^N \left( \frac{ (2k_l)! }{4^{k_l}  (k_l!)^2 } \lambda_l^{k_l} \right) \\
	& + \sum_{a = 0}^K \sum_{\substack{k_1 + k_2 \dots + k_N = K \\ k_i = k_j = a}} \left( 2a \frac{\lambda_j^2 - \lambda_i^2}{\lambda_i \lambda_j} \right) \prod_{l = 1}^N \left( \frac{ (2k_l)! }{4^{k_l}  (k_l!)^2 } \lambda_l^{k_l} \right) \\
	& = \frac{\lambda_j^2 - \lambda_i^2}{\lambda_i \lambda_j} \left( \sum_{\substack{0 \leq a < b \\
			a + b < K}} \sum_{\substack{k_1 + k_2 \dots + k_N = K \\ k_i = a, k_j = b}} \left( 2a + 2b \right) \prod_{l = 1}^N \left( \frac{ (2k_l)! }{4^{k_l}  (k_l!)^2 } \lambda_l^{k_l} \right) \right. \\
	& \quad \left. + \sum_{a = 0}^K \sum_{\substack{k_1 + k_2 \dots + k_N = K \\ k_i = k_j = a}} \left( 2a \right) \prod_{l = 1}^N \left( \frac{ (2k_l)! }{4^{k_l}  (k_l!)^2 } \lambda_l^{k_l} \right)  \right). 
	\end{align*}
	Therefore,  $D_{ij}h(\lambda) = 0$ for all $i, j \in \{ 1, \dots, N \}$ is equivalent to 
	\begin{align*}
		\lambda_1 = \lambda_2 = \dots = \lambda_N = \lambda = \sqrt{\frac{2K}{N + 2K}}.
	\end{align*}
	which is the same as \eqref{eq:station}. Hence, $h\mid_S$ has a unique stationary point \eqref{eq:station}. 
	
	Since
	\begin{align*}
		p_K = \sqrt{f(\lambda)} h(\lambda),
	\end{align*}
	the maximum of $p_K$ over $S$ is bounded above by the values of $f$ and $h$ at stationary point \eqref{eq:station} or at the boundary $S_i$'s. 
	First, at $\eqref{eq:station}$, we have
	\begin{align}
		\sqrt{f(\lambda)} = \left( \frac{N}{N+ 2K}\right)^{\frac{N}{2}}, \label{eq:boundd}
	\end{align}
	and 
	\begin{align*}
		h(\lambda) &= \sum_{k_1 + k_2 \dots + k_N = K } \prod_{l = 1}^N \left( \frac{ (2k_l)! }{4^{k_l}  (k_l!)^2  } \left(\frac{2K}{N+2K} \right)^{k_l} \right). \\
		& =  \left(\frac{2K}{N+2K} \right)^K \sum_{k_1 + k_2 \dots + k_N = K } \prod_{l = 1}^N  \frac{ (2k_l)! }{4^{k_l}  (k_l!)^2  } \\
		& =  \left(\frac{2K}{N+2K} \right)^K \frac{\Gamma(K + \frac{N}{2})}{\Gamma(\frac{N}{2}) K!}.
	\end{align*}
	Therefore,
	\begin{align*}
		p_K & = \left( \frac{N}{N+ 2K}\right)^{\frac{N}{2}} \left(\frac{2K}{N+2K} \right)^K \frac{\Gamma(K + \frac{N}{2})}{\Gamma(\frac{N}{2}) K!} \\
		& \sim  \left( \frac{N}{N+ 2K}\right)^{\frac{N}{2}} \left(\frac{2K}{N+2K} \right)^K  \frac{\sqrt{2 \pi} \sqrt{K+ \frac{N}{2}} }{2 \pi \sqrt{\frac{N}{2}} \sqrt{K}} \frac{\left(K + \frac{N}{2}\right)^{\left(K + \frac{N}{2}\right)}}{K^K \left(\frac{N}{2}\right)^{\frac{N}{2}}} \\
		& = \frac{ \sqrt{K+ \frac{N}{2}} }{\sqrt{\frac{N}{2}} \sqrt{K}} \frac{1}{\sqrt{2\pi}} \\
		& \leq \frac{1}{\sqrt{2\pi}}.
	 \end{align*}
Second, at $S^\mu_l$, we have
\begin{align*}
	\sqrt{f(\lambda)} = \left( \frac{N}{N+ 2K}\right)^{\frac{N}{2}},
\end{align*}
and 
 \begin{align*}
 	h(\lambda) = \mu^K \frac{(2K)!}{4^K (K!)^2} \leq \frac{2K}{1+2K}^K \frac{1}{\sqrt{\pi K}}. 
 \end{align*}
 Therefore,
 \begin{align*}
 	p_K \leq  \left( \frac{N}{N+ 2K}\right)^{\frac{N}{2}} \left(\frac{2K}{1+2K}\right)^K \frac{1}{\sqrt{\pi K}} \leq \frac{1}{\sqrt{2 \pi}}. 
 \end{align*}
 Third, at $S^0_l$, this is reduced to the $N-1$ case. For $N = 1$, $S^0_l$ is empty. 
 The rest follows by induction.
\end{proof}

Let $S$ and $\mu$ be defined as in the proof of Lemma \ref{lem:pJ}. 
We now prove that $\partial S \subseteq \partial\left( [0, \mu]^N  \right)$. 
\begin{lemma}
	\label{lem:boundaryterm}
	The boundary of $S$ must lie within the boundary of $\left( [0, \mu]^N  \right)$. In other words,
	$\displaystyle \partial S \cap (0, \mu)^N = \phi.$ 
\end{lemma}

\begin{proof}
	Suppose there exists $(\lambda_1, \dots, \lambda_N) \in  \partial S \cap (0, \mu)^N.$
	Let $v \in T_{(\lambda_1, \dots, \lambda_N)} S$ be a tangent vector. Let $D_{ij}$ be defined as in the proof of Lemma \ref{lem:pJ}. Since $D_{i, i+1}$ with $i = 1, \dots, N-1$ preserve the constraint $\sum_{l = 1}^N \frac{\lambda_l^2}{1- \lambda_l^2}$ defining $S$, they span $T_{(\lambda_1, \dots, \lambda_N)} S$ and thus there exists $a_i \in \mathbb{R}$ with $i = 1, \dots, N-1$ such that
	\begin{align*}
		v = \sum_{i = 1}^{N-1} a_i D_{i, i+1}.
	\end{align*}
	Let $\gamma$ be a curve in $(0, \mu)^N$ starting at $(\lambda_1, \dots, \lambda_N)$ and satisfying
	\begin{align*}
		\gamma'(t) = \sum_{i = 1}^{N-1} a_i(D_{i, i+1})(\gamma(t)).
	\end{align*}
	Since $ \sum_{i = 1}^{N-1} a_i(D_{i, i+1})$ does not change $\sum_{l = 1}^N \frac{\lambda_l^2}{1- \lambda_l^2}$, $\gamma$ must be inside $S$ for $t \in [0, \epsilon]$ for some small $\epsilon >0$. Thus, $(\lambda_1, \dots, \lambda_N) \notin  \partial S \cap (0, \mu)^N$, and we get a contradiction. 
\end{proof}

\begin{lemma}
	\label{lem:cotan}
	Let $x \in \mathbb{C}$. If $\textnormal{Im}(x) >0$, then
	\begin{align*}
		\vert \cot(\pi n x) - (-i) \vert \leq \frac{2}{e^{2\pi n \textnormal{Im}(x)} - 1}.
	\end{align*}
	If $\textnormal{Im}(x) <0$, then
	\begin{align*}
		\vert \cot(\pi n x) - i \vert \leq \frac{2}{e^{-2\pi n \textnormal{Im}(x)} - 1}.
	\end{align*}
\end{lemma}

\begin{proof}
	We first prove the case when $\textnormal{Im}(x) >0$. Note that
	\begin{align*}
		\cot(\pi n x ) & = i \frac{e^{i\pi n x} + e^{-i\pi n x}}{e^{i\pi n x} - e^{-i\pi n x}} \\
		& =  -i \frac{e^{-i\pi n x} + e^{i\pi n x} }{e^{-i\pi n x} - e^{i\pi n x}} \\
		& = - i \left(1 + \frac{2}{e^{-2 i\pi n x} - 1} \right).
	\end{align*}
	Then, 
	\begin{align*}
		\vert \cot(\pi n x) - (-i) \vert \leq  \left\vert \frac{2}{e^{-2 i\pi n x} - 1}  \right\vert.
	\end{align*}
	From the triangle inequality, we obtain
	\begin{align*}
		\vert e^{-2 i\pi n x} - 1 \vert \geq \vert e^{-2 i\pi n x} \vert - \vert 1 \vert = e^{2\pi n \textnormal{Im}(x)} - 1.
	\end{align*}
	Therefore, 
	\begin{align*}
		\vert \cot(\pi n x) - (-i) \vert \leq \frac{2}{e^{2\pi n \textnormal{Im}(x)} - 1}.
	\end{align*}
	
	When $\textnormal{Im}(x) <0$, the proof is completely analogous. Note that
	\begin{align*}
		\cot(\pi n x ) &  = i \left(1 + \frac{2}{e^{2 i\pi n x} - 1} \right).
	\end{align*}
	Then, 
	\begin{align*}
		\vert \cot(\pi n x) - i \vert \leq  \left\vert \frac{2}{e^{2 i\pi n x} - 1}  \right\vert.
	\end{align*}
	Again, from the triangle inequality, we obtain
	\begin{align*}
		\vert e^{2 i\pi n x} - 1 \vert \geq \vert e^{2 i\pi n x} \vert - \vert 1 \vert = e^{-2\pi n \textnormal{Im}(x)} - 1.
	\end{align*}
	Therefore, 
	\begin{align*}
		\vert \cot(\pi n x) - i \vert \leq \frac{2}{e^{-2\pi n \textnormal{Im}(x)} - 1},  ~~~~\textnormal{Im}(x) <0
	\end{align*}
\end{proof}

\begin{lemma}
	\label{lem:cotan2}
	Let $x \in \mathbb{C}$ and $n \in \mathbb{Z}$. If $\textnormal{Re}(x) \notin \frac{\mathbb{Z}}{n}$, then
	\begin{align*}
		\vert \cot(\pi n x) \vert \leq 1 + \frac{2}{\sqrt{1 - \cos(2\pi n \textnormal{Re}(x))^2}}.
	\end{align*}
\end{lemma}

\begin{proof}
	From the proof of Lemma \ref{lem:cotan}, we have that
	\begin{align*}
		\vert \cot(\pi n x) \vert & = \left\vert -i \left(1 + \frac{2}{e^{-2i\pi n x} - 1}\right) \right\vert \\
		& = 1 + \frac{2}{\vert e^{-2i\pi n x} - 1 \vert}.
	\end{align*}
	Let $x = a + bi$. We compute that
	\begin{align*}
		\vert e^{-2i\pi n x} - 1\vert  & = \vert  e^{-2i\pi n a}  e^{2\pi n b} - 1\vert \\
		& = e^{2\pi n b} \vert  e^{-2i\pi n a}  - e^{-2\pi n b} \vert \\
		& = e^{2\pi n b} \vert \cos(2\pi n a) -  e^{-2\pi n b}  - i \sin(2\pi a) \vert \\
		& = e^{2\pi n b} \sqrt{\left(\cos(2\pi n a) -  e^{-2\pi n b} \right)^2 + \sin(2 \pi a)^2} \\
		& = e^{2\pi n b} \sqrt{e^{-4\pi \sqrt{n} b}  - 2 e^{-2\pi n b} \cos(2\pi n a) + 1} \\
		& = \sqrt{1  - 2 e^{2\pi n b} \cos(2\pi n a) + e^{4\pi n b}} \\
		& =  \sqrt{\left( \cos(2\pi n a) - e^{4\pi n b} \right)^2 + 1 - \cos(2\pi n a)^2} \\
		& \geq  \sqrt{1 - \cos(2\pi n a)^2}.
	\end{align*}
	We thus have that
	\begin{align*}
		\vert \cot(\pi n x) \vert \leq 1 + \frac{2}{\sqrt{1 - \cos(2\pi n a)^2}}, ~~~~ \textnormal{for } \textnormal{Re}(x) \notin \frac{\mathbb{Z}}{n}
	\end{align*}
\end{proof}

\begin{lemma}
	\label{lem:realpart}
	Let $x$ be a complex number with  $\textnormal{Re}(x) >0$. Then, 
	\begin{align*}
		\textnormal{Re}(x\log x) \geq \textnormal{Re}(x) \log(\textnormal{Re}(x)) - \frac{\pi}{2} \vert \textnormal{Im}(x) \vert.
	\end{align*}
\end{lemma}

\begin{proof}
	The proof is a trivial computation.
\end{proof}

The following lemmas are used in the proof of Lemma \ref{lem:ejjprime} for establishing the leading order asymptotic of $e_{J, {J'}}$. Let $t_{m_1, m_2}$ be defined as in \eqref{eq:tm1m2}.

\begin{lemma}
	\label{lem:ejjprimeplace1}
	There exists positive constants $C$ and $\theta < 1$ such that
	$$\left\vert \sum_{\substack{m_1 = 0 \\1 \leq m_2 \leq n-1}} t_{m_1, m_2} \right\vert  \leq C \theta^n.$$
\end{lemma}

\begin{proof}
	We first use Cauchy's residue theorem to turn the discrete sum into a one-dimensional contour integral. Recall
	\begin{align*}
		\beta = \log(\frac{p_{J'}}{1-p_J-p_{J'}}).
	\end{align*}
	We see that
	\begin{align}
		\sum_{\substack{m_1 = 0 \\1 \leq m_2 \leq n-1}} t_{m_1, m_2} & = \sum_{m_2 = 1}^{n-1} t_{0, m_2} \notag \\
		& = \frac{(-\sqrt{p_J})(1-p_J-p_{J'})^{n}  \Gamma(n+1)}{\sqrt{n}} \,  \sum_{m_2 = 1}^{n-1}  \frac{ \left( \sqrt{m_2} - \sqrt{n p_{J'}} \right)  e^{\beta m_2}}{\Gamma(m_2+1) \Gamma(n-m_2+1)} \notag \\
		& = \frac{(-\sqrt{p_J})(1-p_J-p_{J'})^{n}  \Gamma(n+1)}{2i  \sqrt{n}} I_n \label{eq:jjprimein2}
	\end{align}
	where
	\begin{align*}
		I_n & = \int_{\gamma_n}  \frac{ \left( \sqrt{y} - \sqrt{n p_{J'}} \right)  e^{\beta y}}{\Gamma(y+1) \Gamma(n-y+1)} \cot(\pi y) \, dy \\
		& = n^{\tfrac{3}{2}} \int_{\hat{\gamma}_n} \frac{ \left( \sqrt{y} - \sqrt{p_{J'}} \right)  e^{n\beta y}}{\Gamma(ny+1) \Gamma(n-ny+1)} \cot(n \pi y) \, dy
	\end{align*}
	Here, $\gamma_n$ and $\hat{\gamma}_n$ are defined as in Figures \ref{fig:gamma_n} and \ref{fig:hat_gamma_n}. We then apply Stirling's approximation \eqref{eq:gammaapprox} to get
	\begin{align}
		I_n \sim \frac{\sqrt{n}}{2\pi}  e^{-n \log n} \int_{\hat{\gamma}_n}  \rho(y) e^{n\psi(y)} \cot(n \pi y) \, dy, \label{eq:Injjprime}
	\end{align}
	where
	\begin{align*}
		\psi(y) &= \beta y - y(\log y - 1) -  (1-y)(\log (1-y) - 1), \\
		\rho(y) & = \frac{\sqrt{y} - \sqrt{p_{J'} }}{\sqrt{y(1-y)}}.
	\end{align*}
	We now examine the integral.
	\begin{align*}
		H_n = \int_{\hat{\gamma}_n}  \rho(y) e^{n\psi(y)} \cot(n \pi y) \, dy.
	\end{align*}
	As in the proof of Lemma \ref{lem:ejj}, we divide the contour into four components $\hat{\gamma}_{n, 1}, \hat{\gamma}_{n, 2}, \hat{\gamma}_{n, 3}, \hat{\gamma}_{n, 4}$, and define the corresponding integrals $H_n^1, H_n^2, H_n^3, H_n^4$. 
	For $H_n^1$, it follows from Lemma \ref{lem:cotan} that
	\begin{align*}
		H_n^1 \sim (-i) \int_{\hat{\gamma}_{n, 1}} \rho(y)  e^{n\psi(y)} \, dy.
	\end{align*}
	We then deform $\hat{\gamma}_{n, 1}$ as illustrated in Figure \ref{fig:deform}. As in Lemma \ref{lem:ejj}, we focus on the horizontal component that is completely contained within the real line.
	Then, we get
	\begin{align*}
		H_n^1 \sim  i \int_{\frac{\delta}{n}}^{1 - \frac{\delta}{n}} \rho(x) e^{n\psi(x)} \, dx,
	\end{align*} 
	and proceed with saddle point analysis. 
	Note that $\psi$ is real along the horizontal component. From elementary calculus, it is straightforward to see that $x_0 = p_{J'}$ is the unique critical point of $\psi(x)$, and $\psi(x)$ decays as $x$ moves away from $x_0$ along the entire segment. We further compute that
	\begin{align*}
		e^{n \psi(x_0)} = \frac{e^n}{(1-p_{J'})^n}.
	\end{align*}
	From similar analysis as in the proof of Lemma \ref{lem:ejj}, we obtain that
	\begin{align*}
		\int_{\frac{\delta}{n}}^{1 - \frac{\delta}{n}} \rho(x) e^{n\psi(x)} \, dx \sim \frac{e^n}{(1-p_{J'})^n} q(n)
	\end{align*}
	where $q(n)$ is a polynomial of $n$. Therefore, 
	\begin{align*}
		H_n^1 \sim i \frac{e^n}{(1-p_{J'})^n}q(n).
	\end{align*}
	The computation for $H_n^3$ is entirely analogous, and one can easily verify that we get the same leading order asymptotics for $H_n^3$. The contributions from $H_n^2$ and $H_n^4$ are negligible, and they can be analyzed using techniques very similar to the proof of Lemma \ref{lem:ejj}.
	
	In total, we get
	\begin{align*}
		H_n  \sim 2i \frac{e^n}{(1-p_{J'})^n} q(n).
	\end{align*}
	From \eqref{eq:Injjprime}, we obtain
	\begin{align*}
		I_n \sim \frac{i \sqrt{n}}{\pi}  e^{-n \log n} \frac{e^n}{(1-p_{J'})^n} q(n).
	\end{align*}
	From \eqref{eq:jjprimein2}, we obtain
	\begin{align*}
		\sum_{\substack{m_1 = 0 \\1 \leq m_2 \leq n-1}} t_{m_1, m_2} & \sim  \frac{\pi}{2} (-\sqrt{p_J}) e^n e^{-n \log n} \Gamma(n+1) \frac{(1-p_J-p_{J'})^{n}}{(1-p_{J'})^n} q(n) \\
		& \leq C \theta^n,
	\end{align*}
	for some positive constants $C$ and $\theta <1$. 
	This completes the proof.  
\end{proof}

\begin{corollary}
	\label{lem:ejjprimeplace2}
	There exists positive constants $C$ and $\theta < 1$ such that
	$$ \left\vert \sum_{\substack{m_2 = 0 \\1 \leq m_1 \leq n-1}} t_{m_1, m_2} \right\vert \leq C \theta^n.$$
\end{corollary}

\begin{proof}
	The proof is analogous to Lemma \ref{lem:ejjprimeplace1} by swapping $m_1$ and $m_2$ and we leave the details to the readers.
\end{proof}

Let $\phi(x,y)$ be defined as in \eqref{def:ejjprime_phifun} and let $h_n(x,y)$ be defined as in \eqref{def:ejjprime_hnfun}. The following lemmas establish that the integrals of $\phi(x,y)$ or $h_n(x,y)$ over the regions $U_n$, $L_n^+$, $W_{n,1}$, $W_{n,3}$ are negligible compared to the integral over $L_n$. The overall strategy of these proofs is to upper bound the integral by the product of the maximum value of the integrand and the volume of the region, and then show that this upper bound decays exponentially faster than the leading-order asymptotics given as in \eqref{eq:inttriangle}. We first prove that the integral of $\phi(x,y)$ over $U_n$ is negligible. Recall that
\begin{align*}
	U_n = \left\{(x,y) \in \mathbb{R}^2 \mid \frac{\delta}{n} \leq x,y \leq 1 - \frac{\delta}{n}, 1- \frac{1}{2n} < x + y \leq 1+ \frac{1}{n} \right\}.
\end{align*}

\begin{lemma}
	\label{lem:ejjprime_un}
	There exists positive constants $C$ and $\theta <1$ such that
	\begin{align*}
		\left\vert \int_{U_n} \phi(x,y) \, dx dy \right\vert \leq C \left(\frac{e \theta}{n(1-p_J -p_{J'})}\right)^n.
	\end{align*}
\end{lemma}

\begin{proof}
	From the triangle inequality, we get that
	\begin{align*}
		\left\vert \int_{U_n} \phi(x,y) \, dx dy \right\vert  & \leq \int_{U_n} \vert \phi(x,y) \vert \, dx dy \\
		& \leq \max_{U_n} \vert \phi(x,y) \vert \int_{U_n} 1 \, dx dy. 
	\end{align*}
	Since the volume of $U_n$ is bounded above by some positive constant, it is then enough to show that there exists positive constants $C_0$ and $\theta <1$ such that 
	\begin{align*}
		\max_{U_n} \vert \phi(x,y) \vert \leq C_0 \left(\frac{e \theta}{n(1-p_J -p_{J'})}\right)^n.
	\end{align*}
	Let $x,y \in U_n$. We compute that
	\begin{align*}
		\vert \phi(x,y) \vert 
		& \leq \frac{n \vert \sqrt{x} - \sqrt{p_J}  \vert \vert \sqrt{y} - \sqrt{p_{J'}} \vert}{\Gamma(n(1-x-y) + 1)}  \frac{ e^{\alpha nx + \beta ny} }{\Gamma(nx+1) \Gamma(ny +1) }.
	\end{align*}
	To proceed, we first establish that
	\begin{align}
		\frac{n \vert \sqrt{x} - \sqrt{p_J}  \vert \vert \sqrt{y} - \sqrt{p_{J'}} \vert}{\Gamma(n(1-x-y) + 1)} \leq c n, \label{eq:primeusefulform3}
	\end{align}
	for some constant $c >0$. Since $\vert \sqrt{x} - \sqrt{p_J}  \vert \vert \sqrt{y} - \sqrt{p_{J'}} \vert$ is clearly bounded above by a constant, we need only show that $1/\Gamma(n(1-x-y) + 1)$ is bounded above by a constant. However, since there exists $x,y \in U_n$ such that $n(1-x-y) <0$, we cannot use Stirling's approximation directly. Instead, we use the following recursive formula of the Gamma function
	\begin{align}
		\Gamma(w) = \frac{\Gamma(w+1)}{w}. \label{eq:gammarecursive}
	\end{align}
	Let $w= 1-x-y$. Then, $-1 \leq nw \leq \frac{1}{2}$ and $nw + 2 \geq 1$. The recursive formula yields
	\begin{align*}
		\frac{1}{\Gamma(nw + 1)}
		& =  \frac{ nw+1}{\Gamma(nw+1+1)} = \frac{\left( nw +1 \right) \left( nw+2\right)}{\Gamma(nw + 2 + 1)}.
	\end{align*}
	Clearly,
	\begin{align*}
		(nw + 1)(nw + 2) &\leq \frac{3}{2} \times \frac{5}{2} = \frac{15}{4}.
	\end{align*}
	Using Stirling's approximation \eqref{eq:gammaapprox} and the fact that $log(x) - x$ is monotone increasing on $x \geq 1$, we get that
	\begin{align*}
		\Gamma(nw + 2 + 1) &\sim \sqrt{2\pi (n w + 2)} e^{ (nw + 2)(\log(nw + 2) - (nw + 2))} \geq \sqrt{2\pi} e.
	\end{align*}
	Therefore,
	\begin{align*}
		\frac{1}{\Gamma(nw + 1)} \leq \frac{15}{4 \sqrt{2\pi} e}.
	\end{align*}
	This completes the proof of \eqref{eq:primeusefulform3}.
	
	It remains to show that there exists some positive constants $C_1$ and $\theta <1$ such that
	\begin{align}
		\frac{ e^{\alpha nx + \beta ny} }{\Gamma(nx+1) \Gamma(ny +1) } \leq C_1 \left(\frac{e \theta}{n(1-p_J -p_{J'})}\right)^n. \label{eq:primeusefulform4}
	\end{align}
	From Stirling's approximation \eqref{eq:gammaapprox}, we get
	\begin{align*}
		&\Gamma(nx+1) \Gamma(ny +1) \sim 2\pi n \sqrt{x y} e^{n (x + y) \log n} e^{n(x\log x - x + y\log y -y)}.
	\end{align*}
	Thus, 
	\begin{align*}
		\frac{ e^{\alpha nx + \beta ny} }{\Gamma(nx+1) \Gamma(ny +1) } \sim \frac{1}{2\pi n} \frac{e^{-n(x+y)\log n}}{\sqrt{xy}} e^{n\psi(x,y)}
	\end{align*}
	where $\psi: \mathbb{R}_+ \times \mathbb{R}_+ \rightarrow \mathbb{R}$ is given by
	\begin{align*}
		\psi(x,y) = \alpha x + \beta y - x\log x - y \log y + x + y.
	\end{align*}
	Since $x + y \geq 1 - \frac{1}{2n}$,
	\begin{align*}
		e^{-n(x+y)\log n} &\leq \frac{\sqrt{n}}{n^n}.
	\end{align*}
	Since $x, y \geq \frac{\delta}{n}$, 
	\begin{align*}
		\frac{1}{\sqrt{xy}} &\leq \frac{n}{\delta}.
	\end{align*}
	We now seek an upper bound for $\psi(x,y)$. Note that it suffices to prove there exists some positive constants $C_2$ and $\theta <1$ such that
	\begin{align}
		e^{n\psi(x,y) } \leq  C_2 \left(\frac{e \theta}{1-p_J -p_{J'}}\right)^n. 
		\label{eq:unform1}
	\end{align}
	Towards this, we compute the critical point $(x_0, y_0)$ of $\psi(x,y)$ by solving
	\begin{align*}
		\frac{\partial \psi}{\partial x} &= \alpha - \log x = 0, \\ 
		\frac{\partial \psi}{\partial y} &= \beta - \log y = 0.
	\end{align*}
	The system has a unique solution, that is
	\begin{align*}
		x_0 &= e^{\alpha}=\frac{p_J}{1 - p_J - p_{J'}},\\
		y_0 &= e^{\beta} =\frac{p_{J'}}{1 - p_J - p_{J'}} .
	\end{align*}
	The Hessian matrix of $\psi(x,y)$ is
	\begin{align*}
		H = \begin{bmatrix}
			-\frac{1}{x} & 0 \\
			0 & -\frac{1}{y}
		\end{bmatrix}.
	\end{align*}
	Since $H$ is negative definite over $x, y >0$, $\psi(x,y)$ is concave down and the critical point $(x_0, y_0)$ is the global maximum over the entire positive quadrant. Hence,
	\begin{align*}
		\psi(x_0, y_0) = x_0 + y_0 = \frac{p_J + p_{J'}}{1 - p_J - p_{J'}}.
	\end{align*}
	In Lemma \ref{lem:pJ}, we see that
	\begin{align*}
		p_J - p_{J'} < \frac{1}{2},
	\end{align*}
	and hence
	\begin{align*}
		\frac{p_J + p_{J'}}{1 - p_J - p_{J'}} < 1.
	\end{align*}
	For $n$ sufficiently large, $(x_0, y_0) \notin U_n$, and the maximum value of $\psi(x,y)$ must be attained on the boundary of $U_n$. There are four components, and it suffices to prove that the maximum over each component satisfies \eqref{eq:unform1}.
	\begin{enumerate}
		\item  Along $x + y = 1+ \frac{1}{n}$, the maximum value of $\psi(x,y)$ occurs at 
		$$x_0=\frac{p_J}{p_J + p_{J'}} \left(1 + \frac{1}{n}\right), ~~~~ y_0 =\frac{p_{J'}}{p_J + p_{J'}} \left(1 + \frac{1}{n}\right)$$
		as determined using basic calculus. Furthermore,
		\begin{align*}
			\psi(x_0, y_0) = \left(1 + \frac{1}{n}\right) \left( \log(\frac{p_J + p_{J'}}{1-p_J - p_{J'}})- \log\left(1 + \frac{1}{n}\right) + 1 \right).
		\end{align*}
		Therefore,
		\begin{align*}
			e^{n\psi(x,y)} \leq e^{n \psi(x_0,y_0)} =  \left(\frac{p_J + p_{J'}}{1-p_J - p_{J'}}\right)^{n+1} \left(1 + \frac{1}{n}\right)^{-(n+1)} e^{n+1},
		\end{align*}
		which clearly satisfies \eqref{eq:unform1}.
		
		\item Along $x + y = 1- \frac{1}{2n}$, the maximum value of $\psi(x,y)$ occurs at 
		$$x_0=\frac{p_J}{p_J + p_{J'}} \left(1 - \frac{1}{2n}\right), ~~~~ y_0 =\frac{p_{J'}}{p_J + p_{J'}} \left(1 - \frac{1}{2n}\right).$$ Furthermore, 
		\begin{align*}
			\psi(x_0, y_0) = \left(1 - \frac{1}{2n}\right) \left( \log(\frac{p_J + p_{J'}}{1-p_J - p_{J'}}) - \log\left(1 - \frac{1}{2n}\right) + 1 \right).
		\end{align*}
		Therefore,
		\begin{align*}
			e^{n\psi(x,y)} \leq e^{n \psi(x_0,y_0)} =  \left(\frac{p_J + p_{J'}}{1-p_J - p_{J'}}\right)^{n-\frac{1}{2}} \left(1 - \frac{1}{2n}\right)^{-(n-\frac{1}{2})} e^{n-\frac{1}{2}}, 
		\end{align*}
		which clearly satisfies \eqref{eq:unform1}.
		
		\item Along $x = \frac{\delta}{n}$ and $1 - \frac{1}{2n}-\frac{\delta}{n} \leq  y \leq 1-\frac{\delta}{n}$, we see that
		\begin{align*}
			\psi(x,y) = \alpha \frac{\delta}{n} - \frac{\delta}{n} \log(\frac{\delta}{n}) + \frac{\delta}{n} + \beta y -y \log(y) + y.
		\end{align*}
		Since 
		\begin{align*}
			e^{n \left(\alpha \frac{\delta}{n} - \frac{\delta}{n} \log(\frac{\delta}{n}) + \frac{\delta}{n} \right)} = \left(\frac{p_{J}}{1- p_J - p_{J'}}\right)^\delta \left(\frac{\delta}{n}\right)^{-\delta}e^\delta
		\end{align*}
		is bounded above by a polynomial of $n$, we only need to give an upper bound for 
		$$\rho(y) = \beta y -y \log(y) + y.$$ The
		maximum value of $\rho(y)$  over the positive real line occurs at
		\begin{align*}
			y_0  = \frac{p_{J'}}{1- p_J - p_{J'}} = \exp(\beta).
		\end{align*}
		Since $p_J + p_{J'} < \frac{1}{2}$,
		$$\frac{p_{J'}}{1- p_J - p_{J'}} < 1. $$
		For $n$ sufficiently large,
		$$y_0 \notin \left[1 - \frac{1}{2n}-\frac{\delta}{n},  1-\frac{\delta}{n} \right].$$ 
		The maximum value of $\rho(y)$ is attained at the end points of the interval. 
		At $y = 1 - \frac{1}{2n}-\frac{\delta}{n}$, we get
		\begin{align*}
			e^{n \rho(y)}  \leq \left(\frac{e p_J }{1-p_J-p_{J'}}\right)^{n - \frac{1}{2} - \delta} \left(1 - \frac{1}{2n}-\frac{\delta}{n}\right)^{n - \frac{1}{2} - \delta}.
		\end{align*}
		At $y =1-\frac{\delta}{n}$, we get
		\begin{align*}
			e^{n \rho(y)}  \leq \left(\frac{e p_J }{1-p_J-p_{J'}}\right)^{n - \delta} \left(1 -\frac{\delta}{n}\right)^{n - \delta}.
		\end{align*}
		In both cases, \eqref{eq:unform1} is clearly met.
		
		\item Along $y = \frac{\delta}{n}$ and $1 - \frac{1}{2n}-\frac{\delta}{n} \leq  x \leq 1-\frac{\delta}{n}$, the analysis is completely analogous to Case 3 with $x$ and $y$ swapped.
	\end{enumerate}
	This completes the proof of Lemma \ref{lem:ejjprime_un}.
\end{proof}

\begin{lemma}
	\label{lem:ejjprime_lnplus}
	There exists positive constants $C$ and $\theta <1$ such that
	\begin{align*}
		\left\vert \int_{L_n^+} \phi(x,y) \, dx dy \right\vert \leq C \left(\frac{e \theta}{n(1-p_J -p_{J'})}\right)^n.
	\end{align*}
\end{lemma}

\begin{proof}
	As in Lemma \ref{lem:ejjprime_un}, the triangle inequality implies that it is enough to prove that there exists positive constants $C_0$ and $\theta <1$ such that
	\begin{align*}
		\max_{L_n^+} \vert \phi(x,y) \vert \leq C_0 \left(\frac{e \theta}{n(1-p_J -p_{J'})}\right)^n.
	\end{align*}
	Let $(x,y) \in L_n^+$. Recall that Euler's reflection formula \eqref{eq:euler} states
	\begin{align*}
		\Gamma(z)\Gamma(1-z) = \frac{\pi}{\sin(\pi z)}, ~~~~\forall z \in \mathbb{C} - \mathbb{Z}.
	\end{align*}
	So, we see that
	\begin{align}
		\label{eq:euler}
		\frac{1}{\Gamma(z)\Gamma(1-z)} = \frac{\sin(\pi z)}{\pi}, ~~~~ \forall z \in \mathbb{C}.
	\end{align}
	Then,
	\begin{align*}
		\frac{1}{\Gamma(n(1-x-y)+1)} & = \frac{1}{\Gamma(1-n(x+y-1))} \\
		& = \frac{1}{\pi} \sin(\pi n (x+y-1)) \Gamma(n (x+y-1)).
	\end{align*}
	Furthermore, it follows from the recursive formula \eqref{eq:gammarecursive} that
	\begin{align*}
		\Gamma(n (x+y-1)) = \frac{\Gamma(n (x+y-1) + 1) }{n (x+y-1)}. 
	\end{align*}
	Thus,
	\begin{align*}
		\frac{1}{\Gamma(n(1-x-y)+1)} = \frac{\sin(\pi n (x+y-1)) \Gamma(n (x+y-1) + 1) }{\pi n (x+y-1)},
	\end{align*}
	and therefore,
	\begin{align*}
		\vert \phi(x,y) \vert 
		& \leq \frac{\vert \sqrt{x} - \sqrt{p_J}  \vert \vert \sqrt{y} - \sqrt{p_{J'}} \vert}{\pi \vert x+y-1\vert}  \frac{ e^{\alpha nx + \beta ny} \Gamma(n(x+y-1)+1) }{\Gamma(nx+1) \Gamma(ny +1) }.
	\end{align*}
	Using Stirling's approximation \eqref{eq:gammaapprox}, we obtain that
	\begin{align*}
			\vert \phi(x,y) \vert  \sim \frac{e^n}{n^n} \rho(x,y) e^{n\psi(x,y)}
	\end{align*}
	for some $\rho(x,y)$ bounded above by a polynomial of $n$ and
	\begin{align*}
		\psi(x,y) = \alpha x + \beta y + (x+y-1)\log(x+y-1) - x\log x- y\log y.
	\end{align*}
	Therefore, it suffices to prove that there exists positive constants $C_1$ and $\theta <1$ such that
	\begin{align}
		\max_{L_n^+} e^{n\psi(x,y)} \leq C_1 \left(\frac{\theta}{1-p_J -p_{J'}}\right)^n. \label{eq:lnplususefn1}
	\end{align}
	Towards this, we find the critical point $(x_0, y_0)$ of $\psi(x,y)$ by solving 
	\begin{align*}
		&\frac{\partial \psi}{\partial x} = \alpha + \log(x+y-1) -\log x = 0 \\
		&\frac{\partial \psi}{\partial y} = \beta + \log(x+y-1) -\log y = 0 
	\end{align*}
	and we get a unique solution
	\begin{align*}
		(x_0, y_0) = (p_J, p_{J'}) \notin L_n^+.
	\end{align*}
	Therefore, the maximum of $\psi(x,y)$ over $L_n^+$ must be attained at the boundary. There are three components, and it suffices to prove that the maximum over each component satisfies \eqref{eq:lnplususefn1}. 
	\begin{enumerate}
		\item Along $x + y = 1 + \frac{1}{n}$, we obtain
		\begin{align*}
			\psi(x,y) & = \alpha \left( 1+ \frac{\delta}{n} -y \right) + \beta y + \frac{1}{n} \left( \frac{1}{n}\right) - \left( 1+ \frac{\delta}{n} -y \right) \log\left( 1+ \frac{\delta}{n} -y \right)  - y\log y.
		\end{align*}
		Define
		\begin{align*}
			\eta(y) = (\beta - \alpha) y - y\log y - \left( 1+ \frac{\delta}{n} -y \right) \log\left( 1+ \frac{\delta}{n} -y \right).
		\end{align*}
		We find the critical point $y_0$ of $\eta(y)$ by solving
		\begin{align*}
			\eta'(y) = \beta - \alpha + \log\left( 1+ \frac{\delta}{n} -y \right) - \log y = 0,
		\end{align*}
		and we get a unique solution
		\begin{align*}
			y_0 = \left(1 + \frac{1}{n}\right) \frac{p_{J'}}{p_J + p_{J'}}. 
		\end{align*}
		Since for large $n$
		\begin{align*}
			y_0 \in \left[ \frac{\delta}{n}, 1 -\frac{\delta}{n} \right],
		\end{align*}
		and for all $y \in  \left[ \frac{\delta}{n}, 1 -\frac{\delta}{n} \right]$,
		\begin{align*}
			\eta''(y) = - \frac{ 1+ \frac{\delta}{n}}{y \left(1 + \frac{\delta}{n} -y \right)} < 0,
		\end{align*}
		we see that $y_0$ is a global maximum. Hence,
		\begin{align*}
			\eta(y) \leq - \left( 1 + \frac{1}{n}\right)\log( 1 + \frac{1}{n}) - \left( 1 + \frac{1}{n}\right)\log(\frac{p_{J'}}{p_J + p_{J'}} ),
		\end{align*}
		and therefore
		\begin{align*}
			e^{n\psi(x,y)} \leq  \frac{1}{n} \left(\frac{p_{J}}{1- p_J - p_{J'}} \right)^{n+1}   \left(\frac{p_{J}}{p_J + p_{J'}} \right)^{n+1}.
		\end{align*}
		Clearly,  \eqref{eq:lnplususefn1} is satisfied.
		
		\item Along $x = 1-\frac{\delta}{n}$ and $\frac{1}{n}+\frac{\delta}{n} \leq  y \leq 1-\frac{\delta}{n}$, we obtain
		\begin{align*}
			\psi(x,y) & = \alpha \left(1-\frac{\delta}{n} \right) + \beta y + \left(y-\frac{\delta}{n}\right)\log(y-\frac{\delta}{n}) \\
			& \quad -  \left(1-\frac{\delta}{n}\right)\log(1-\frac{\delta}{n}) - y\log y.
		\end{align*}
		Define 
		\begin{align*}
			\nu(y) = \beta y +  \left(y-\frac{\delta}{n}\right)\log(y-\frac{\delta}{n}) - y\log y.
		\end{align*}
		From
		\begin{align*}
		\nu'(y) = \beta + \log(y-\frac{\delta}{n}) - \log y = 0,
		\end{align*}
		we get
		\begin{align}
			\frac{p_{J'}}{1-p_J-p_{J'}} = \frac{y}{y-\frac{\delta}{n}}, \label{eq:lnplususefn2}
		\end{align}
		Lemma \ref{lem:pJ} implies that $p_J+p_{J'} < \frac{1}{2}$, and thus
		\begin{align*}
			\frac{p_{J'}}{1-p_J-p_{J'}} < 1.
		\end{align*}
		On the other hand, 
		\begin{align*}
			\frac{y}{y-\frac{\delta}{n}} > 1.
		\end{align*}
		Therefore, there is no solution to \eqref{eq:lnplususefn2} and the function $\nu(y)$ has no critical points. The maximum of $\nu(y)$ must be attained at one of the two end points of the interval $\left[ \frac{1}{n}+\frac{\delta}{n}, 1-\frac{\delta}{n}\right].$ At $y =  \frac{1}{n}+\frac{\delta}{n}$, we see that 
		\begin{align*}
			\nu(y) = \beta \frac{1 + \delta}{n} + \frac{1}{n}\log(\frac{1}{n}) - \frac{1 + \delta}{n} \log(\frac{1 + \delta}{n}),
		\end{align*}
		and hence
		\begin{align*}
			e^{n\psi(x,y)} & \leq \left(\frac{p_{J}}{1-p_J-p_{J'}}\right)^{n - \delta} \left(\frac{p_{J'}}{1-p_J-p_{J'}} \right)^{1 + \delta} \\
			& \quad \left(\frac{1}{n}\right)  \left(1-\frac{\delta}{n}\right)^{\delta - n} \left(\frac{1 + \delta}{n}\right)^{1 + \delta}. 
		\end{align*}
		Clearly, \eqref{eq:lnplususefn1} is satisfied.
		At $y =1-\frac{\delta}{n}$, we see that
		\begin{align*}
			\nu(y) = \beta \left(1-\frac{\delta}{n}\right) + \left(1-\frac{2\delta}{n}\right)\log(1-\frac{2\delta}{n}) - \left(1-\frac{\delta}{n} \right) \log(1-\frac{\delta}{n}),
		\end{align*}
		and hence
		\begin{align*}
			e^{n\psi(x,y)} & \leq \left(\frac{p_{J}}{1-p_J-p_{J'}}\right)^{n - \delta} \left(\frac{p_{J'}}{1-p_J-p_{J'}} \right)^{n- \delta} \\
			& \quad \left(1-\frac{2\delta}{n}\right)^{n-2\delta} \left(1-\frac{\delta}{n}\right)^{\delta - n} \left(1-\frac{\delta}{n}\right)^{\delta - n}. 
		\end{align*}
		Lemma \ref{lem:pJ} implies that $p_J + p_{J'} < \frac{1}{2}$ and thus 
		\begin{align*}
			\frac{p_J + p_{J'}}{1-p_J-p_{J'}} < \frac{1}{2}. 
		\end{align*}
		Therefore, \eqref{eq:lnplususefn1} is satisfied.
		\item Along $y = 1-\frac{\delta}{n}$ and $\frac{1}{n}+\frac{\delta}{n} \leq  x \leq 1-\frac{\delta}{n}$, the analysis is nearly identical to the previous case after swapping $x$ and $y$. We therefore omit its proof to avoid repetition.
	\end{enumerate}
	
\end{proof}

We now show that the integral of $\phi(x,y)$ over $W_{n,1}$ is negligible. Recall that 
\begin{align*}
	W_{n,1} & = \left\{ (x,y) \in \mathbb{C}^2 \mid x = \frac{\delta}{n} + si, \, y = t + si, \right . \\ 
	& \hspace{1in} \left. \frac{\delta}{n} \leq t \leq 1 - \frac{\delta}{n}, \,  0 \leq s \leq \frac{\xi}{\sqrt{n}}  \right\}.
\end{align*}

\begin{lemma}
	\label{lem:ejjprime_wn1}
	There exists positive constants $C$ and $\theta <1$ such that
	\begin{align*}
		\left\vert \int_{W_{n,1}} \phi(x,y) \, dx dy \right\vert  \leq C e^{2\pi \sqrt{n}\xi} \left(\frac{e \theta}{n(1-p_J -p_{J'})}\right)^n.
	\end{align*}
\end{lemma}

\begin{proof}
	As before, the triangle inequality yields that
	\begin{align*}
		\left\vert \int_{W_{n,1}} \phi(x,y) \, dx dy \right\vert
		& \leq  \int_{W_{n,1}} \left\vert \phi(x,y) \right\vert  \, dx dy \\
		& \leq \max_{W_{n,1}} \left\vert \phi(x,y) \right\vert \int_{W_{n,1}} 1 \, \vert dx dy \vert.  
	\end{align*}
	Since
	\begin{align*}
		\int_{W_{n,1}} 1 \, \vert dx dy \vert & =  \int_{\frac{\delta}{n}}^{1- \frac{\delta}{n}} \int_{0}^{\frac{\xi}{\sqrt{n}}} 1 \, ds dt 
		\leq \frac{\xi}{\sqrt{n}},
	\end{align*}
	it is enough to show that there exists positive constants $C_0$ and $\theta <1$ such that
	\begin{align*}
		\max_{W_{n,1}} \left\vert \phi(x,y) \right\vert \leq C_0  e^{2\pi \sqrt{n}\xi}  \left(\frac{e \theta}{n(1-p_J -p_{J'})}\right)^n.
	\end{align*}
	
	Let $(x,y) \in W_{n,1}$. It follows from the recursive formula \eqref{eq:gammarecursive} and the Stirling's approximation \eqref{eq:gammaapprox} of the Gamma function that
	\begin{align*}
		\vert \phi(x,y) \vert & \leq \frac{n \vert \sqrt{x} - \sqrt{p_J} \vert \vert \sqrt{y} - \sqrt{p_{J'}}\vert \vert n(1-x-y)+1 \vert e^{\textnormal{Re}(\alpha nx + \beta ny )}}{\vert \Gamma(nx + 1) \Gamma(ny + 1) \Gamma(n(1-x-y)+2)\vert } \\
		& \sim  \rho_n(x,y) e^{\textnormal{Re}(\psi_n(nx,ny))}
	\end{align*}
	where
	\begin{align*}
		\rho_n(x,y) &= \frac{\vert \sqrt{x} - \sqrt{p_J} \vert \vert \sqrt{y} - \sqrt{p_{J'}}\vert \vert n(1-x-y)+1 \vert}{(\sqrt{2\pi})^3 \vert \sqrt{xy} \sqrt{n(1-x-y)+1}\vert)},
	\end{align*}
	and
	\begin{align*}
		\psi_n(x,y) & = \alpha x + \beta y - x\log(x) - y \log(y) - (n-x-y + 1)\log(n-x-y + 1) + n + 1.
	\end{align*}
	Since $\rho_n(x,y)$ is clearly bounded above by a polynomial of $n$, it is then enough to show that there exists positive constants $C_1$ and $\theta <1$ such that
	\begin{align}
		e^{\textnormal{Re}(\psi_n(nx,ny))} \leq C_1  e^{2\pi \sqrt{n}\xi}  \left(\frac{e \theta}{n(1-p_J -p_{J'})}\right)^n. \label{eq:ejjwn1usefulfn1}
	\end{align}
	From Lemma \ref{lem:realpart}, we get that
	\begin{align*}
		\textnormal{Re}(\psi_n(nx,ny)) & \leq  \textnormal{Re}(\alpha nx + \beta ny) - \textnormal{Re}(nx)\log(\textnormal{Re}(nx) ) - \textnormal{Re}(ny)\log(\textnormal{Re}(ny)) \\
		& \quad - \textnormal{Re}(n(1-x-y)+1)\log( \textnormal{Re}(n(1-x-y) + 1)) + n + 1 \\
		& \quad + \frac{\pi}{2} \left( \vert \textnormal{Im}(nx) \vert + \vert \textnormal{Im}(ny) \vert + \vert \textnormal{Im}(n(1-x-y) + 1) \vert \right).
	\end{align*}
	Since $x = \frac{\delta}{n} + si$, $y = t + si$ with $\frac{\delta}{n} \leq t \leq 1 - \frac{\delta}{n}$ and $0 \leq s \leq \frac{\xi}{\sqrt{n}}$, we obtain
	\begin{align*}
		\textnormal{Re}(\psi_n(nx,ny))
		& \leq \alpha \delta + \beta n t - \delta \log \delta - nt \log(nt) \\
		& \quad - (n-\delta -nt + 1) \log(n-\delta -nt +1) + n + 1 + 2\pi \sqrt{n} \xi \\
		& = \alpha \delta + \beta n t - \delta \log \delta - (n-\delta+1)\log n - nt \log(t) \\
		& \quad - n(1-\frac{\delta}{n} -t +\frac{1}{n}) \log(1-\frac{\delta}{n} -t +\frac{1}{n}) + n + 1 + 2\pi \sqrt{n} \xi.
	\end{align*}
	For convenience, let us define
	\begin{align*}
		\nu_n(t) = \beta t - t \log(t) - \left(1-\frac{\delta}{n} -t +\frac{1}{n} \right) \log(1-\frac{\delta}{n} -t +\frac{1}{n}). 
	\end{align*}
	Then,
	\begin{align*}
		\textnormal{Re}(\psi_n(nx,ny)) 
		\leq  \alpha \delta - \delta \log \delta - (n-\delta+1)\log n  + n + 1 + 2\pi \sqrt{n} \xi  + n \nu_n(t).
	\end{align*}
	From elementary calculus, the maximum value of $\nu_n(t)$ over the positive real line occurs at
	\begin{align*}
		t_0 = \frac{p_{J'}}{1- p_J} \left(1-\frac{\delta}{n} + \frac{1}{n}\right).
	\end{align*}
	Furthermore, 
	\begin{align*}
		\nu_n(t_0) =  \left(1-\frac{\delta}{n} + \frac{1}{n}\right) \log(\frac{1-p_J}{1- p_J - p_{J'}}) - \left(1-\frac{\delta}{n} + \frac{1}{n}\right) \log\left(1-\frac{\delta}{n} + \frac{1}{n}\right).
	\end{align*}
	Hence,
	\begin{align*}
		e^{\textnormal{Re}(\psi_n(nx,ny))} \leq \frac{e^{\alpha \delta + 1}}{\delta^\delta} e^{2\pi \sqrt{n} \xi} \frac{e^n}{n^{n- \delta + 1}}  \left(\frac{1-p_J}{1- p_J - p_{J'}} \right)^{n - \delta + 1} \left(1-\frac{\delta}{n} + \frac{1}{n} \right)^{n - \delta + 1} 
	\end{align*}
	Clearly, \eqref{eq:ejjwn1usefulfn1} is satisfied and this completes the proof.
\end{proof}

\begin{lemma}
	\label{lem:ejjprime_wn3}
	There exists positive constants $C$ and $\theta <1$ such that
	\begin{align*}
		\left\vert \int_{W_{n,3}} \phi(x,y) \, dx dy \right\vert  \leq C e^{2\pi \sqrt{n}\xi} \left(\frac{e \theta}{n(1-p_J -p_{J'})}\right)^n.
	\end{align*}
\end{lemma}

\begin{proof}
	From the triangle inequality,
	\begin{align*}
		\left\vert \int_{W_{n,3}} \phi(x,y) \, dx dy \right\vert
		& \leq  \int_{W_{n,3}} \left\vert \phi(x,y) \right\vert  \, \vert dx dy \vert \\
		& \leq \max_{W_{n,3}} \left\vert \phi(x,y) \right\vert \int_{W_{n,3}} 1 \, \vert dx dy \vert.  
	\end{align*}
	Since
	\begin{align*}
		\int_{W_{n,3}} 1 \, \vert dx dy \vert & =  \int_{\frac{\delta}{n}}^{1- \frac{\delta}{n}} \int_{0}^{\frac{\xi}{\sqrt{n}}} 1 \, ds dt 
		\leq \frac{\xi}{\sqrt{n}},
	\end{align*}
	it suffices to prove that there exists positive constants $C_0$ and $\theta <1$ such that
	\begin{align}
		\max_{W_{n,3}} \left\vert \phi(x,y) \right\vert \leq C_0  e^{2\pi \sqrt{n}\xi}  \left(\frac{e \theta}{n(1-p_J -p_{J'})}\right)^n. \label{eq:wn3ufn1}
	\end{align}
	Let us further divide $W_{n,3}$ into two parts
	\begin{align*}
		W_{n,3}^- &= \{  (x, y) \in W_{n,3} \mid \textnormal{Re}(x) + \textnormal{Re}(y) \leq 1 + \frac{1}{ n}\}, \\
		W_{n,3}^+ &= \{  (x, y) \in W_{n,3} \mid \textnormal{Re}(x) + \textnormal{Re}(y) \geq 1 + \frac{1}{ n}\}.
	\end{align*}
	It is then enough to show that the maximum of $\vert \phi(x,y) \vert$ over each of the two regions satisfies \eqref{eq:wn3ufn1}. 
	
	Let $(x,y) \in W_{n,3}^+$ and write
	\begin{align*}
		x &= 1- \frac{\delta}{n} + si, \\
		y &= t + si
	\end{align*}
	with $ \frac{1 + \delta}{n} \leq t \leq 1- \frac{\delta}{n}$ and $0 \leq s \leq \frac{\xi}{\sqrt{n}}$.
	It follows from Euler's reflection formula \eqref{eq:euler} and the Stirling's approximation \eqref{eq:gammaapprox} that
	\begin{align*}
		\vert \phi(x,y) \vert \sim \frac{e^n}{n^n} \rho(x,y) e^{n \textnormal{Re} \psi(x,y)}
	\end{align*}
	for some $\rho(x,y)$ bounded above by a polynomial of $n$ and
	\begin{align*}
		\psi(x,y) = \alpha x + \beta y - x\log x- y\log y  + (x+y-1)\log(x + y-1). 
	\end{align*}
	Therefore, it suffices to prove that there exists positive constants $C_1$ and $\theta <1$ such that
	\begin{align}
		\max_{W_{n,3}^+} e^{n\textnormal{Re}(\psi(x,y))} \leq C_1  e^{2\pi \sqrt{n}\xi}  \left(\frac{\theta}{1-p_J -p_{J'}}\right)^n. \label{eq:wn3ufn2}
	\end{align}
	From Lemma \ref{lem:realpart}, we see that
	\begin{align*}
		\textnormal{Re}(\psi(x,y)) &\leq \textnormal{Re}(\alpha x + \beta y) - \textnormal{Re}(x) \log(\textnormal{Re}(x)) - \textnormal{Re}(y) \log(\textnormal{Re}(y)) \\
		& \quad + \textnormal{Re}(x+y -1) \log(\textnormal{Re}(x+y -1)) \\
		& \quad \frac{\pi}{2} \left( \vert \textnormal{Im}(x) \vert + \vert \textnormal{Im}(y) \vert + \vert \textnormal{Im}(x+y-1) \vert \right) \\
		& \leq \alpha \left(1- \frac{\delta}{n}\right) + t\beta - \left(1- \frac{\delta}{n}\right) \log(1- \frac{\delta}{n}) - t\log t \\
		& \quad + \left(t - \frac{\delta}{n}\right)\log(t - \frac{\delta}{n}) + \frac{2\pi \xi }{\sqrt{n}}.
	\end{align*}
	Define
	\begin{align*}
		\nu(t) = t \beta  - t \log(t) + \left(t - \frac{\delta}{n}\right)\log(t - \frac{\delta}{n}).
	\end{align*} 	
	One can easily verify that
	\begin{align*}
		\nu''(t) = \frac{\frac{\delta}{n} }{t \left(t - \frac{\delta}{n}\right)} >0
	\end{align*}
	for all $\frac{1 + \delta}{n} \leq t \leq 1 - \frac{\delta}{n}$.
	Therefore, the maximum of $\nu(t)$ must be attained at the boundary points. At $t = \frac{1 + \delta}{n}$, we have
	\begin{align*}
			\nu(t) = \frac{1 + \delta}{n} \beta  - \frac{1 + \delta}{n} \log(\frac{1 + \delta}{n}) + \frac{1}{n} \log(\frac{1}{n} ),
	\end{align*}
	and thus
	\begin{align*}
		e^{n\textnormal{Re}(\psi(x,y))} & \leq \left(\frac{p_J}{1 - p_J - p_{J'}}\right)^{n-\delta} \left(\frac{p_{J'}}{1 - p_J - p_{J'}}\right)^{1+\delta} \left(1 - \frac{\delta}{n}\right)^{\delta - n} \\
		& \quad \left(\frac{1 + \delta}{n} \right)^{-1-\delta} \frac{e^{2\pi \sqrt{n}\xi}}{n}.
	\end{align*}
	Clearly, \eqref{eq:wn3ufn2} is satisfied.
	At $t = 1 - \frac{\delta}{n}$, we have
	\begin{align*}
		\nu(t) = \left(1 - \frac{\delta}{n}\right) \beta  - \left(1 - \frac{\delta}{n}\right) \log(1 - \frac{\delta}{n}) + \left(1 - \frac{2\delta}{n} \right) \log(1 - \frac{2\delta}{n}),
	\end{align*}
	and thus
	\begin{align*}
		e^{n\textnormal{Re}(\psi(x,y))} & \leq \left(\frac{p_J}{1 - p_J - p_{J'}}\right)^{n-\delta} \left(\frac{p_{J'}}{1 - p_J - p_{J'}}\right)^{n-\delta} \\
		& \quad \left(1 - \frac{\delta}{n} \right)^{-2n + 2\delta} \left( 1 - \frac{2\delta}{n}\right)^{n-2\delta} e^{2\pi \sqrt{n}\xi}.
	\end{align*}
	Since Lemma \ref{lem:pJ} implies that $p_J + p_{J'} < \frac{1}{2}$, we get that
	\begin{align*}
		\frac{p_J p_{J'}}{1 - p_J - p_{J'}} < 1.
	\end{align*}
	Therefore, \eqref{eq:wn3ufn2} is satisfied.
	
	Let $(x,y) \in W_{n,3}^-$ and write
	\begin{align*}
		x &= 1- \frac{\delta}{n} + si, \\
		y &= t + si
	\end{align*}
	with $\frac{\delta}{n} \leq t \leq \frac{1 + \delta}{n} $ and $0 \leq s \leq \frac{\xi}{\sqrt{n}}$.
	It follows from Stirling's approximation \eqref{eq:gammaapprox} that
	\begin{align*}
		\vert \phi(x,y) \vert & \leq \frac{n \vert \sqrt{x} - \sqrt{p_J} \vert \vert \sqrt{y} - \sqrt{p_{J'}} \vert}{\vert \Gamma(n(1-x-y) + 1) \vert} \frac{e^{\textnormal{Re}(\alpha x + \beta y)}}{\vert \Gamma(nx + 1) \vert \vert \Gamma(ny + 1) \vert} \\
		& \sim \frac{e^{-n \log n \textnormal{Re}(x + y)}}{\vert \Gamma(n(1-x-y) + 1) \vert} \rho_d(x,y) e^{n\textnormal{Re} (\psi_d(x,y))}
	\end{align*}
	for some $\rho_d(x,y)$ bounded above by a polynomial of $n$ and
	\begin{align*}
		\psi_d(x,y) = \alpha x + \beta y - x\log x- y\log y  + x +y.
	\end{align*}
	Since $\textnormal{Re}(x + y) \leq 1$, we get
	\begin{align*}
		e^{-n \log n \textnormal{Re}(x + y)} \leq \frac{1}{n^n}.
	\end{align*}
	Using the recursive formula \eqref{eq:gammarecursive} and the Stirling's approximation \eqref{eq:gammaapprox}, we obtain
	\begin{align*}
		\frac{1}{\vert  \Gamma(n(1-x-y) + 1) \vert} & = \frac{\vert n(1-x-y) + 1\vert \vert n(1-x-y) + 2 \vert }{ \vert \Gamma(n(1-x-y) + 2 + 1) \vert} \\
		& \sim \rho_g(x,y) e^{\textnormal{Re}(\psi_g(x,y))}
	\end{align*}
	for some $\rho_g(x,y)$ bounded above by a polynomial of $n$ and 
	\begin{align*}
		\psi_g(x,y) = n(1-x-y) + 2 - \left(n(1-x-y) + 2\right) \log(n(1-x-y) + 2).
	\end{align*}
	Since
	\begin{align*}
		\textnormal{Re}(\psi_g(x,y)) &\leq \textnormal{Re}(n(1-x-y) + 2) - \textnormal{Re}(n(1-x-y) + 2) \log(\textnormal{Re}(n(1-x-y) + 2)) \\
		&\quad + \frac{\pi}{2} \vert \textnormal{Im}(n(1-x-y) + 2) \vert \\
		& \leq (\delta - nt) + 2 - ((\delta - nt) + 2 ) \log((\delta - nt) + 2) + \pi \sqrt{n} \xi \\
		& \leq 1 + \pi \sqrt{n} \xi,
	\end{align*}
	we get that
	\begin{align*}
		e^{\textnormal{Re}(\psi_g(x,y))} \leq e^{1 + \pi \sqrt{n} \xi}.
	\end{align*}
	Therefore, it suffices to prove that there exists positive constants $C_2$ and $\theta <1$ such that
	\begin{align}
		\max_{W_{n,3}^-} e^{n\textnormal{Re}(\psi_d(x,y))} \leq C_2  e^{\pi \sqrt{n}\xi}  \left(\frac{e \theta}{1-p_J -p_{J'}}\right)^n. \label{eq:wn3ufn3}
	\end{align}
	From Lemma \ref{lem:realpart}, we obtain
	\begin{align*}
		\textnormal{Re}(\psi_d(x,y) ) & \leq 	\textnormal{Re}(\alpha x + \beta y) - 	\textnormal{Re}(x)\log(	\textnormal{Re}(x) ) - 	\textnormal{Re}(y)\log(	\textnormal{Re}(y))  \\
		& \quad + \textnormal{Re}(x + y) + \frac{\pi}{2} \left( \vert \textnormal{Im}(x) \vert +  \vert \textnormal{Im}(y) \vert \right) \\
		& \leq \alpha \left(1 - \frac{\delta}{n} \right) + \beta t - \left(1 - \frac{\delta}{n} \right) \log\left(1 - \frac{\delta}{n} \right)  - t\log(t) \\
		& \quad + \left(1 - \frac{\delta}{n} \right)  + t + \pi \frac{\xi}{\sqrt{n}}.
	\end{align*}
	Now let us define
	\begin{align*}
		\eta(t) = t \beta - t \log(t) + t.
	\end{align*}
	Since the domain of $t$, given by $\left[ \frac{\delta}{n}, \frac{1 + \delta}{n} \right]$ shrinks as $n$ grows, the maximum of $\eta(t)$ over this interval must be attained at its end points. At $t = \frac{\delta}{n}$, we have 
	\begin{align*}
		\eta(t) = \frac{\delta}{n} \beta - \frac{\delta}{n} \log(\frac{\delta}{n}) + \frac{\delta}{n},
	\end{align*}
	and thus
	\begin{align*}
		e^{n \textnormal{Re}(\psi_d(x,y) )} & \leq \left(\frac{p_J}{1 - p_J - p_{J'}}\right)^{n-\delta} \left(\frac{p_{J'}}{1 - p_J - p_{J'}}\right)^{\delta} \\
		& \quad \left(1 - \frac{\delta}{n} \right)^{\delta - n} \left(\frac{\delta}{n} \right)^{-\delta} e^{n - \delta} e^\delta e^{\pi \sqrt{n} \xi}. 
	\end{align*}
	Clearly, \eqref{eq:wn3ufn3} is met. 
	At $t =\frac{1 + \delta}{n}$, we have
	\begin{align*}
		\eta(t) = \frac{1 + \delta}{n}  \beta - \frac{1 + \delta}{n} \log(\frac{1 + \delta}{n}) + \frac{1 + \delta}{n},
	\end{align*}
	and thus
	\begin{align*}
		e^{n \textnormal{Re}(\psi_d(x,y) )} & \leq \left(\frac{p_J}{1 - p_J - p_{J'}}\right)^{n-\delta} \left(\frac{p_{J'}}{1 - p_J - p_{J'}}\right)^{1+\delta} \\
		& \quad \left(1 - \frac{\delta}{n} \right)^{\delta - n} \left( \frac{1 + \delta}{n} \right)^{-1-\delta} e^{n - \delta} e^{1 +\delta} e^{\pi \sqrt{n} \xi}. 
	\end{align*}
	Clearly, \eqref{eq:wn3ufn3} is met.  This completes the proof of Lemma \ref{lem:ejjprime_wn3}. 
\end{proof}

We now turn to examine the integral of $h_n(x,y)$ over $\hat{\gamma}_{n, i} \times \hat{\gamma}_{n, j}$ for a few representative cases in $i, j \in \{1, 2, 3, 4\}$. 

\begin{lemma}
	\label{lem:ejjprime_rn12}
	There exists positive constants $C$ and $\theta <1$ such that
	\begin{align*}
		\left\vert \int_{\hat{\gamma}_{n, 1} \times \hat{\gamma}_{n, 2}} h_n(x,y) \, dx dy \right\vert \leq C e^{2\pi \sqrt{n}\xi} \left(\frac{e \theta}{n(1-p_J -p_{J'})}\right)^n.
	\end{align*}
\end{lemma}

\begin{proof}
	From the triangle inequality, we get
	\begin{align*}
		\left\vert \int_{\hat{\gamma}_{n, 1} \times \hat{\gamma}_{n, 2}} h_n(x,y) \, dx dy \right\vert & \leq \int_{\hat{\gamma}_{n, 1} \times \hat{\gamma}_{n, 2}} \vert h_n(x,y) \vert \, dx dy \\
		& \leq \max_{\hat{\gamma}_{n, 1} \times \hat{\gamma}_{n, 2}} \vert \phi(x,y) \vert \vert \cot(n\pi x) \vert \vert \cot(n\pi y) \vert \int_{\hat{\gamma}_{n, 1} \times \hat{\gamma}_{n, 2}} 1 \, \vert dx dy \vert. 
	\end{align*}
	From Lemma \ref{lem:cotan}, we see that for all $x \in \hat{\gamma}_{n, 1}$, 
	\begin{align*}
		\vert \cot(n\pi x) \vert  \sim 1.
	\end{align*} 
	From  Lemma \ref{lem:cotan2}, we see that for all $y \in \hat{\gamma}_{n, 2}$,
	\begin{align*}
		\vert \cot(n\pi y) \vert \leq  1 + \tfrac{2}{\sqrt{1-\cos(2\pi \delta)^2}}.
	\end{align*}
	Since 
	\begin{align*}
		\int_{\hat{\gamma}_{n, 1} \times \hat{\gamma}_{n, 2}} 1 \, \vert dx dy \vert \leq  2 \frac{\xi}{\sqrt{n}}, 
	\end{align*}
	it is therefore enough to prove that
	there exists positive constants $C_0$ and $\theta <1$ such that
	\begin{align*}
		\max_{\hat{\gamma}_{n, 1} \times \hat{\gamma}_{n, 2}} \vert \phi(x,y) \vert \leq C_0 e^{2\pi \sqrt{n}\xi} \left(\frac{e \theta}{n(1-p_J -p_{J'})}\right)^n.
	\end{align*}
	Let us divide $\hat{\gamma}_{n, 1} \times \hat{\gamma}_{n, 2}$ into two parts
	\begin{align*}
		\left(\hat{\gamma}_{n, 1} \times \hat{\gamma}_{n, 2}\right)^- = \left\{ (x,y) \in \hat{\gamma}_{n, 1} \times \hat{\gamma}_{n, 2} \mid \textnormal{Re}(x) + \textnormal{Re}(y) \leq 1 - \frac{1}{2n} \right\}, \\
		\left(\hat{\gamma}_{n, 1} \times \hat{\gamma}_{n, 2}\right)^+ = \left\{ (x,y) \in \hat{\gamma}_{n, 1} \times \hat{\gamma}_{n, 2} \mid \textnormal{Re}(x) + \textnormal{Re}(y) > 1 - \frac{1}{2n} \right\}.
	\end{align*}
	
	First, let $(x,y) \in \left(\hat{\gamma}_{n, 1} \times \hat{\gamma}_{n, 2}\right)^-$. Then, we write
	\begin{align*}
		x &= t + \frac{\xi}{\sqrt{n}} i, \\
		y &= \frac{\delta}{n} + si
	\end{align*}
	with $ \frac{\delta}{n} \leq t \leq 1 - \frac{1}{2n} - \frac{\delta}{n}$ and $-\frac{\xi}{\sqrt{n}}  \leq s \leq \frac{\xi}{\sqrt{n}}$. 
	Using Stirling's approximation \eqref{eq:gammaapprox}, we get that
	\begin{align*}
		\vert \phi(x,y) \vert = \frac{e^n}{n^n} \rho(x,y) e^{n \textnormal{Re}(\psi(x,y))}
	\end{align*}
	for some $\rho(x,y)$ bounded above by a polynomial of $n$ and
	\begin{align*}
		\psi(x,y) = \alpha x + \beta y - x\log x - y\log y - (1-x-y)\log(1-x-y).
	\end{align*}
	Hence,
	it suffices to prove that there exists positive constants $C_1$ and $\theta <1$ such that
	\begin{align}
		\max_{\left(\hat{\gamma}_{n, 1} \times \hat{\gamma}_{n, 2}\right)^-} e^{n\textnormal{Re}(\psi(x,y))} \leq C_1  e^{2\pi \sqrt{n}\xi}  \left(\frac{\theta}{1-p_J -p_{J'}}\right)^n. \label{eq:gamma12ufn2}
	\end{align}
	From Lemma \ref{lem:realpart}, we obtain that
	\begin{align*}
		\textnormal{Re}(\psi(x,y)) & \leq \textnormal{Re}(\alpha x + \beta y) -\textnormal{Re}(x)\log(\textnormal{Re}(x)) - \textnormal{Re}(y)\log(\textnormal{Re}(y))  \\
		 & \quad -\textnormal{Re}(1-x-y)\log(\textnormal{Re}(1-x-y)) + 
		 \frac{\pi }{2} \left( \vert \textnormal{Im}(x) \vert + \vert \textnormal{Im}(y) \vert + \vert \textnormal{Im}(1-x-y) \vert \right) \\
		 & \leq \alpha t + \beta \frac{\delta}{n} - t\log(t) - \frac{\delta}{n}  \log(\frac{\delta}{n} ) - \left(1-t-\frac{\delta}{n}  \right)\log(1-t-\frac{\delta}{n} )  +\frac{2\pi \xi}{\sqrt{n}}.
	\end{align*}
	Let us define 
	\begin{align*}
		\nu(t) = \alpha t - t \log(t) - \left(1-t-\frac{\delta}{n}  \right)\log(1-t-\frac{\delta}{n} ).
	\end{align*}
	From elementary calculus, we see that $\nu(t)$ has a unique critical point at
	\begin{align*}
		t_0 = \frac{p_J}{1-p_{J'}} \left(1-t-\frac{\delta}{n}  \right)
	\end{align*}
	by solving
	\begin{align*}
		\nu'(t) = \alpha + \log(1-t-\frac{\delta}{n} ) - \log(t) = 0.
	\end{align*}
	Since 
	\begin{align*}
		\nu''(t) = \frac{1 - \frac{\delta}{n} }{t \left(t - \left(1-\frac{\delta}{n} \right)\right)} < 0 
	\end{align*}
	for all $ t \leq 1-\frac{\delta}{n}$, we know that the critical point is in fact a global maximum. Therefore, 
	\begin{align*}
		\nu(t) \leq - \left(  1-\frac{\delta}{n} \right) \log( 1-\frac{\delta}{n} ) -  \left(  1-\frac{\delta}{n} \right) \log \left( \frac{ 1- p_J -p_{J'}}{1-p_{J'}}\right),
	\end{align*}
	and hence,
	\begin{align*}
	e^{n\textnormal{Re}(\psi(x,y))} \leq \left( \frac{p_J}{1- p_J -p_{J'}}\right)^\delta \left(\frac{\delta}{n} \right)^{-\delta} \left(1 - \frac{\delta}{n} \right)^{\delta - n } \left( \frac{1- p_{J'}}{1- p_J -p_{J'} }\right)^{n-\delta} e^{2\pi \xi \sqrt{n}}.
	\end{align*}
	Clearly, \eqref{eq:gamma12ufn2} is met. 
	
	The maximum value of the function  $\vert \phi(x,y) \vert$ over  $ \left(\hat{\gamma}_{n, 1} \times \hat{\gamma}_{n, 2}\right)^+$ can be analyzed using similar techniques as in the case of $W_{n,3}^-$ of Lemma \ref{lem:ejjprime_wn3}. 
\end{proof}

\begin{lemma}
	\label{lem:ejjprime_rn14}
	There exists positive constants $C$ and $\theta <1$ such that
	\begin{align*}
		\left\vert \int_{\hat{\gamma}_{n, 1} \times \hat{\gamma}_{n, 4}} h_n(x,y) \, dx dy \right\vert \leq C e^{2\pi \sqrt{n}\xi} \left(\frac{e \theta}{n(1-p_J -p_{J'})}\right)^n.
	\end{align*}
\end{lemma}

\begin{proof}
	The proof proceeds similarly to that of Lemma \ref{lem:ejjprime_rn12}. From the triangle inequality and the approximations of the $\cot$ function, it suffices to prove that there exists positive $C$ and $\theta <1$ such that
	\begin{align*}
		\max_{\hat{\gamma}_{n, 1} \times \hat{\gamma}_{n, 4}} \vert \phi(x,y) \vert  \leq C e^{2\pi \sqrt{n}\xi} \left(\frac{e \theta}{n(1-p_J -p_{J'})}\right)^n.
	\end{align*}
	The rest of the proof is entirely analogous to Lemma \ref{lem:ejjprime_wn3}, and we omit the details to avoid repetition.
\end{proof}

\begin{lemma}
	\label{lem:ejjprime_rn22}
	There exists positive constants $C$ and $\theta <1$ such that
	\begin{align*}
		\left\vert \int_{\hat{\gamma}_{n, 2} \times \hat{\gamma}_{n, 2}} h_n(x,y) \, dx dy \right\vert \leq C e^{2\pi \sqrt{n}\xi} \left(\frac{e \theta}{n(1-p_J -p_{J'})}\right)^n.
	\end{align*}
\end{lemma}

\begin{proof}
	Again, using the triangle inequality, we get
	\begin{align*}
		\left\vert \int_{\hat{\gamma}_{n, 2} \times \hat{\gamma}_{n, 2}} h_n(x,y) \, dx dy \right\vert & \leq \int_{\hat{\gamma}_{n, 2} \times \hat{\gamma}_{n, 2}} \vert h_n(x,y) \vert \, dx dy \\
		& \leq \max_{\hat{\gamma}_{n, 2} \times \hat{\gamma}_{n, 2}} \vert \phi(x,y) \vert \vert \cot(n\pi x) \vert \vert \cot(n\pi y) \vert \int_{\hat{\gamma}_{n, 2} \times \hat{\gamma}_{n, 2}} 1 \, \vert dx dy \vert. 
	\end{align*}
	From Lemma \ref{lem:cotan2}, we obtain for all $(x,y) \in \hat{\gamma}_{n, 2} \times \hat{\gamma}_{n, 2}$
	\begin{align*}
		\vert \cot(n\pi x) \vert \leq  1 + \tfrac{2}{\sqrt{1-\cos(2\pi \delta)^2}}, \\
		\vert \cot(n\pi y) \vert \leq  1 + \tfrac{2}{\sqrt{1-\cos(2\pi \delta)^2}}.
	\end{align*}
	Since
	\begin{align*}
		\int_{\hat{\gamma}_{n, 2} \times \hat{\gamma}_{n, 2}} 1 \, \vert dx dy \vert \leq  2, 
	\end{align*}
	it suffices to prove that there exists positive constants $C_0$ and $\theta <1$ such that
	\begin{align*}
		\max_{\hat{\gamma}_{n, 2} \times \hat{\gamma}_{n, 2}} \vert \phi(x,y) \vert \leq C_0 e^{2\pi \sqrt{n}\xi} \left(\frac{e \theta}{n(1-p_J -p_{J'})}\right)^n.
	\end{align*}
	
	Let $(x,y) \in \hat{\gamma}_{n, 2} \times \hat{\gamma}_{n, 2}$. We write
	\begin{align*}
		x = \frac{\delta}{n} - si, \\
		y = \frac{\delta}{n} - ti.
	\end{align*}
	with $-\frac{\xi}{\sqrt{n}} \leq s, t \leq \frac{\xi}{\sqrt{n}}.$
	Then, from the Stirling's approximation \eqref{eq:gammaapprox} of the Gamma function, we get that
	\begin{align*}
		\left\vert \phi(x,y) \right\vert \sim e^{n - n\log(n)} \rho_n(x,y) e^{n\textnormal{Re}(\psi(x,y))}, 
	\end{align*}
	where
	\begin{align*}
		\rho_n(x,y) &= \frac{\vert \sqrt{x} - \sqrt{p_J} \vert \vert \sqrt{y} - \sqrt{p_{J'}}\vert}{(\sqrt{2\pi})^3 \vert \sqrt{xy} \sqrt{n(1-x-y)}\vert)},
	\end{align*}
	and 
	\begin{align*}
		\psi(x,y) & = \alpha x + \beta y - x\log(x) - y \log(y) - (1-x-y) \log(1-x-y).
	\end{align*}
	It is straightforward that $\rho_n(x,y)$ is bounded above by some polynomial of $n$. Hence, it suffices to prove that there exist positive constants $C_1$ and $\theta <1$ such that
	\begin{align}
		e^{n\textnormal{Re}(\psi(x,y))} \leq C_1 e^{2\pi \sqrt{n}\xi} \left(\frac{\theta}{1-p_J -p_{J'}}\right)^n. \label{eq:rjjuseful1}
	\end{align}
	From Lemma \ref{lem:realpart}, we get that 
	\begin{align*}
		\textnormal{Re}(\psi(x,y))
		& \leq  \textnormal{Re}(\alpha x + \beta y) - \textnormal{Re}(x)\log(\textnormal{Re}(x) ) - \textnormal{Re}(y)\log(\textnormal{Re}(y)) \\
		& \quad - \textnormal{Re}(1-x-y)\log( \textnormal{Re}(1-x-y)) \\
		& \quad + \frac{\pi}{2} \left( \vert \textnormal{Im}(x) \vert + \vert \textnormal{Im}(y) \vert + \vert \textnormal{Im}(1-x-y) \vert \right).
	\end{align*}
	Since $x = \frac{\delta}{n} - ti$, $y = \frac{\delta}{n} - si$ with $-\frac{\xi}{\sqrt{n}} \leq t,s \leq \frac{\xi}{\sqrt{n}}$, 
	\begin{align*}
		\textnormal{Re}(\psi(x,y))
		& \leq \frac{\delta}{n} (\alpha + \beta) - \frac{2\delta}{n}\log(\frac{\delta}{n}) - \left(1- \frac{2\delta}{n} \right)\log(1- \frac{2\delta}{n}) +  2\pi \frac{\xi}{\sqrt{n}}.
	\end{align*}
	Hence,
	\begin{align*}
		e^{n\textnormal{Re}(\psi(x,y))} \leq \frac{e^{\delta(\alpha + \beta) + 2\pi \sqrt{n} \xi}}{ \left( \frac{\delta}{n}\right)^{2 \delta} \left(1- \tfrac{2\delta}{n} \right)^{n - 2\delta}} = \frac{n^{2\delta} \left(1- \tfrac{2\delta}{n} \right)^{2\delta} e^{\delta(\alpha + \beta)}}{ \delta^{2\delta}} \frac{e^{2\pi \sqrt{n} \xi}}{\left(1- \tfrac{2\delta}{n} \right)^{n}}.
	\end{align*}
	Clearly, \eqref{eq:rjjuseful1} is satisfied and this completes the proof. 
\end{proof}

	\section{Analysis of the efficiency gains}
\subsection{Comparison between GBS-P and MC}
\foreach \N in {3,4,5,6}{
	\foreach \K in {2,3,4,5,6,7,8}{
		\IfFileExists{fig/GBSP2/haf-N_\N-K_\K.png}{%
			\begin{figure}[H]
				\centering
				\includegraphics[width=\linewidth]{fig/GBSP2/haf-N_\N-K_\K.png}
				\caption{GBS-P vs MC with $N = \N$ and $K = \K$.}
			\end{figure}
		}{}
	}
}

\subsection{Comparison bewteen GBS-I and MC}

\foreach \N in {3,4,5,6}{
	\foreach \K in {2,3,4,5,6,7,8}{
		\IfFileExists{fig/GBSI2/hafsq-N_\N-K_\K.png}{%
			\begin{figure}[H]
				\centering
				\includegraphics[width=\linewidth]{fig/GBSI2/hafsq-N_\N-K_\K.png}
				\caption{GBS-I vs MC with $N = \N$ and $K = \K$.}
			\end{figure}
		}{}
	}
}

\section{Convergence analysis}

\subsection{Computing $P^{\textnormal{GBS-P}}_{\Haf}$}

\foreach \N in {3,4,5,6}{
	\foreach \K in {2,3,4,5,6,7,8}{
		\IfFileExists{fig/GBSP/haf-N_\N-K_\K.png}{%
			\begin{figure}[H]
				\centering
				\includegraphics[width=0.5\linewidth]{fig/GBSP/haf-N_\N-K_\K.png}
				\caption{Convergence plot of $P^{\textnormal{GBS-P}}_{\Haf}$ with $N = \N$ and $K = \K$.}
			\end{figure}
		}{}
	}
}

\subsection{Computing $P^{\textnormal{GBS-I}}_{\Haf^2}$}

\foreach \N in {3,4,5,6}{
	\foreach \K in {2,3,4,5,6,7,8}{
		\IfFileExists{fig/GBSI/hafsq-N_\N-K_\K.png}{%
			\begin{figure}[H]
				\centering
				\includegraphics[width=0.5\linewidth]{fig/GBSI/hafsq-N_\N-K_\K.png}
				\caption{Convergence plot of $P^{\textnormal{GBS-I}}_{\Haf^2}$ with $N = \N$ and $K = \K$.}
			\end{figure}
		}{}
	}
}

	\clearpage
	\bibliographystyle{plain}
	\bibliography{haf}
	
\end{document}